\documentclass[10pt, draftclsnofoot, journal, onecolumn,romanappendices]{IEEEtran}
\usepackage{setspace}
\usepackage{cite}
\ifCLASSINFOpdf
   \usepackage[pdftex]{graphicx}
\else
  \usepackage[dvips]{graphicx}
\fi
\usepackage{amsmath,amssymb,amsthm,mathrsfs,amsfonts}
\usepackage{array}
\usepackage{dsfont}
\usepackage{amsbsy}
\usepackage{epstopdf}
\usepackage{graphicx,color}
\usepackage{hyperref}
\usepackage{subfig}
\usepackage{pgfplots}
\usepackage{graphicx}
\usepackage{booktabs}
\usepackage{xcolor}
\usepackage{soul}

\newcommand\sbullet[1][.7]{\mathbin{\vcenter{\hbox{\scalebox{#1}{$\bullet$\phantom{e}}}}}}

\newtheorem{lemma}{Lemma}
\newtheorem*{problem*}{Problem Formulation}
\newtheorem{proposition}{Proposition}

\newtheorem{corollary}{Corollary}
\long\def\symbolfootnote[#1]#2{\begingroup%
\def\thefootnote{\fnsymbol{footnote}}\footnote[#1]{#2}\endgroup}
\newtheorem{theorem}{Theorem}
\newtheorem{definition}{Definition}

\newtheorem{assumption}{Assumption}

\newtheorem{cov_model}{Covariance Matrix Model}
\newtheorem*{pilot_model}{Non-Orthogonal Pilots Model}

 \allowdisplaybreaks

\newcommand {\y}{\mathbf{y}}
\renewcommand {\v}{\mathbf{v}}

\newcommand {\x}{\mathbf{x}}
\newcommand {\e}{\mathbf{e}}
\newcommand {\n}{\mathbf{n}}
\renewcommand {\u}{\mathbf{u}}
\newcommand {\0}{\mathbf{0}}
\newcommand {\h}{\mathbf{h}}

\newcommand{\p}{\mathbf{p}}

\newcommand {\X}{\mathbf{X}}
\newcommand {\Y}{\mathbf{Y}}
\newcommand {\D}{\mathbf{D}}
\renewcommand {\S}{\mathbf{S}}
\newcommand {\W}{\mathbf{W}}

\newcommand {\E}{\mathbf{E}}

\newcommand {\U}{\mathbf{U}}
\newcommand {\V}{\mathbf{V}}
\newcommand {\Z}{\mathbf{Z}}
\newcommand {\I}{\mathbf{I}}

\newcommand {\C}{\mathbf{C}}
\newcommand {\A}{\mathbf{A}}
\newcommand {\B}{\mathbf{B}}
\newcommand {\M}{\mathbf{M}}
\newcommand {\N}{\mathbf{N}}
\renewcommand {\P}{\mathbf{P}}
\newcommand {\T}{\mathbf{T}}

\newcommand {\Hc}{\mathbf{H}}

\newcommand {\etav}{\boldsymbol{\eta}}

\newcommand {\betav}{\boldsymbol{\beta}}

\newcommand {\La}{\mathbf{\Lambda}}

\newcommand {\Sig}{\mathbf{\Sigma}}

\newcommand {\Ps}{\mathbf{\Psi}}

\newcommand {\Gam}{\mathbf{\Gamma}}

\newcommand {\De}{\mathbf{\Delta}}

\newcommand {\Diag}{\mathrm{diag}}

\newcommand {\Tr}{\mathrm{tr}}
\renewcommand {\vec}{\mathrm{vec}}

\newcommand {\Exp}{\mathrm{E}}

\newcommand {\CN} {\mathcal{CN}}

\renewcommand {\Pr} {\mathsf{Pr}}

\newcommand {\Compl}{\mbox{$\mathbb{C}$}}
\newcommand {\Real}{\mbox{$\mathbb{R}$}}

\newcommand{\snr}{\mathsf{snr}}
\newcommand{\sinr}{\mathsf{sinr}}

\newcommand{\mse}{\mathsf{mse}}

\newcommand{\tras}{\mathrm{T}}

\begin{document}

\title{ Covariance-Aided CSI Acquisition with Non-Orthogonal Pilots in Massive MIMO: \\A Large-System Performance Analysis}
\author{Alexis~Decurninge, Luis~G.~Ord\'o\~nez, and Maxime~Guillaud \thanks{The authors are with the Mathematics and Algorithmic Sciences Lab., Paris Research Center, Huawei Technologies France, Boulogne-Billancourt. Emails: \texttt{\{alexis.decurninge, luis.ordonez, maxime.guillaud\}@huawei.com}}
\thanks{This paper was presented in part at the 2019 IEEE Global Communications Conference (GLOBECOM), Hawaii, USA, 9--13 December 2019.}}
\maketitle

\begin{abstract}
Massive multiple-input multiple-output (MIMO) systems use antenna arrays with a large number of antenna elements to serve many different users simultaneously. The large number of antennas in the system makes, however, the channel state information (CSI) acquisition strategy design critical and particularly challenging. Interestingly, in the context of massive MIMO systems, channels exhibit a large degree of spatial correlation which results in strongly rank-deficient spatial covariance matrices at the base station (BS). With the final objective of analyzing the benefits of covariance-aided uplink multi-user CSI acquisition in massive MIMO systems, here we compare the channel estimation mean-square error (MSE) for (i) conventional CSI acquisition, which does not assume any knowledge on the user spatial covariance matrices and uses orthogonal pilot sequences; and (ii) covariance-aided CSI acquisition, which exploits the individual covariance matrices for channel estimation and enables the use of non-orthogonal pilot sequences. We apply a large-system analysis to the latter case, for which new asymptotic MSE expressions are established under various assumptions on the distributions of the pilot sequences and on the covariance matrices. We link these expressions to those describing the estimation MSE of conventional CSI acquisition with orthogonal pilot sequences of some equivalent length. This analysis provides insights on how much  training overhead can be reduced with respect to the conventional strategy when a covariance-aided approach is adopted.
\end{abstract}

\section{Introduction}

\noindent Massive multiple-input multiple-output (MIMO) \cite{Rusek_etal_SPMag2012,Larsson_MassiveMIMO_IEEE_CommMag2014,BjornsonHoydisSanguinetti} is considered to be one of the key technologies for realizing the performance targets expected from future wireless systems \cite{Boccardi2014_5G}. Massive MIMO base stations (BSs) use antenna arrays with the number of antenna elements being some orders of magnitude larger than classical MIMO technology. As a result, the system spectral efficiency can be effectively boosted by spatially multiplexing many different users in the same communication resource element \cite{JSAC2013HoydisDebbah}. This requires, however, accurate channel state information (CSI).  Given the large number of antennas at the BS and the large number of users simultaneously served, the CSI acquisition strategy design is critical and particularly challenging in massive MIMO systems.

Conventional cellular MIMO systems acquire CSI by sensing the channel with pilot signals, known at both sides of the communication link.  In the absence of prior information about the channel statistics or when the channels are independent and identically distributed (i.i.d.), it is well known that the length of these pilot sequences should at least coincide with the total number of transmit antennas in the system \cite{Hassibi2003} for the channels to be identifiable. Additionally, mutually orthogonal pilot sequences are preferred, since they result in a better channel estimation accuracy with covariance-agnostic channel estimators \cite{Hassibi2003,Zhan97}. For the problem we consider in this paper, i.e., uplink (UL) CSI acquisition in a massive MIMO system serving single-antenna users, these conditions impose the pilot length to be at least equal to the number of users for which the channel is being simultaneously estimated.  Depending on the coherence time of the channel or, more exactly, on the channel sensing periodicity, the transmission of long pilots sequences (to guarantee orthogonality) instead of data-bearing symbols can represent a significant loss in UL spectral efficiency.

Interestingly, in the context of massive MIMO systems, channels are far from being i.i.d and, in contrast, they exhibit a large degree of spatial correlation which results in a strongly rank-deficient spatial covariance matrix at the BS, as shown by numerous channel measurements campaigns as well as by theoretical channel models. For instance, experimental data in \cite{Payami_Tufvesson_Channel_measurements_analysis_2_6GHz_EUCAP12} show that the measured weakest channel singular values are significantly smaller than what would be expected under the Gaussian i.i.d. fading hypothesis \cite{Martinez_TAP18_experimental_MassiveMIMO_propagation}.  This phenomenon is due to the fact that a small number of specular paths dominate in the propagation scenario, as put in evidence e.g. in \cite{DanFei_etal_MassiveMIMO_measurements_analysis_3_3GHz_ChinaCom2015}. Independent experimental results in \cite{Hoydis_MMIMO_measurements_ISWCS2012} confirm a rank-deficient spatial covariance matrix by analyzing the channel correlation and channel matrix condition number observed with a cylindrical array of 112 elements in a mixed line-of-sight and non-line-of-sight scenario. Since the observed correlation is not mitigated by increasing the number of elements in the antenna array, we can conclude that it is a fundamental property of the propagation environment. Furthermore, this rank-deficient property holds irrespective of the massive MIMO array geometry: \cite{Payami_Tufvesson_Channel_measurements_analysis_2_6GHz_EUCAP12} uses a uniform linear array, \cite{Hoydis_MMIMO_measurements_ISWCS2012} a cylindrical array, and \cite{Martinez_TAP18_experimental_MassiveMIMO_propagation} a 2D array for the covariance eigenvalue profile analysis. On a more theoretical side, it was shown in \cite{Yin_Gesbert_etal_coordinated_estimation_JSAC2013,Adhikary_JSDM_IT2013} that if the support of the angle of arrival distribution is assumed to be bounded, the channel covariance matrix is rank-deficient with a rank depending on the tightness of the support bounds.

Consequently, by exploiting the rank-deficiency property, covariance-aided techniques are likely to be a key ingredient in the design of spectrally efficient massive MIMO systems.  Indeed,~\cite{Bjo18_unlimited_capacity} demonstrates that sharing (perfect) covariance information across different cells results in unbounded spectral efficiencies (as the number of BS antennas increases without bound) under a fairly mild assumption on the linear independence between the user covariance matrices. In the context of CSI acquisition, covariance information has been mainly exploited to propose  orthogonal pilot reuse strategies\cite{Yin_Gesbert_etal_coordinated_estimation_JSAC2013, Adhikary_JSDM_IT2013 ,You15, Hajri_etal_MMIMO_user_clustering_pilot_assignment_Allerton2016} or non-orthogonal pilot designs \cite{pilot_length_minimization_Asilomar2015,Jia15,Tom16} to mitigate pilot contamination\cite{Jos09,Jose_etal_contamination_TWC11,Ngo_Marzetta_Larsson_contamination_ICASSP2011}, that is, the undesired effect of obtaining a channel estimate that is contaminated by the channels of other users. All these methods rely on the intuition that users can be (partially) separated in space using their individual covariance matrices during the CSI acquisition process and perfectly orthogonal pilot sequences are no longer needed. In the extreme case of all the individual covariance matrices spanning mutually orthogonal subspaces, unit-length pilot sequences are sufficient to guarantee channel identifiability \cite{pilot_length_minimization_Asilomar2015}. Unfortunately, it is difficult to know how spatial covariance information helps reducing the training overhead required by orthogonal pilots beyond this limit situation.


This work aims at analyzing the fundamental performance of covariance-aided CSI acquisition with non-orthogonal pilots in the massive MIMO regime, i.e., when the individual spatial covariance matrices are rank-deficient and possibly span non-orthogonal subspaces. With this objective, we focus on a MIMO system acquiring the CSI in the uplink under the assumption that the BS is able to perfectly track the individual spatial covariance matrices of all users with negligible additional pilot sequences (as described e.g. in \cite{Dec17}). We study analytically the channel estimation mean-square error (MSE) for the following cases: \textsf{(i)} \emph{conventional CSI acquisition}, which does not assume any spatial covariance knowledge and uses orthogonal pilot sequences; and \textsf{(ii)} \emph{covariance-aided CSI acquisition}, which exploits the individual spatial covariance matrices for channel estimation and possibly uses non-orthogonal pilot sequences. This work is motivated by the difficulty of interpreting the estimation MSE formulas for case \textsf{(ii)} under general covariance matrices and pilot sequences. Specifically, our contribution is as follows:

\begin{itemize}
\item We derive deterministic equivalents for the covariance-aided estimation MSE under different assumptions (either random or deterministic) for the covariance matrices and the set of pilot sequences. This allows to better understand a general situation beyond the extreme cases when orthogonal pilot sequences are used or when the individual spatial covariance matrices span mutually orthogonal subspaces.
\medskip 
\item When the covariance matrices and the pilots sequences are assumed to be drawn from certain random distributions, we link the deterministic equivalent obtained for case \textsf{(ii)}  to the MSE expression of case \textsf{(i)} with orthogonal pilots of an \emph{equivalent length} depending on the received SNRs and the ranks of the covariance matrices of all the users in the system.  This result is used to answer the question of how much the training overhead can be reduced with respect to orthogonal pilots when covariance-aided CSI acquisition is adopted and at least the performance of conventional CSI acquisition needs to be guaranteed.
\medskip 
\item In order to obtain the deterministic equivalents for the MSE in the covariance-aided case, we derive new results on random matrix theory, which are of interest in their own. In particular, we extend the well-known trace-lemma, initially stated in \cite{Bai98}, to block matrices in Proposition \ref{prop:block_trace} using the so-called block-trace operator \cite{Fil98_blktrace}.   
\end{itemize}

The rest of the paper is organized as follows. In Section \ref{sec:preliminaries} we introduce the channel model and describe in detail the adopted training-based CSI acquisition process. Section \ref{sec:mse_analysis} presents the channel estimation MSEs for both the conventional and the covariance-aided CSI acquisition schemes, whereas Section \ref{sec:large_system_analysis} contains the main contribution of this paper, that is, the large-system analysis of the covariance-aided MSE for different covariance matrices and non-orthogonal pilots models. Additionally we apply in Section \ref{sec:large_system_analysis} the obtained deterministic equivalents to approximately solve the so-called CSI pilot length optimization problem in closed-form.  Finally, in Section \ref{sec:numerical} we validate our results via numerical simulations.

\section{Preliminaries} \label{sec:preliminaries}

\noindent In this paper we consider a massive MIMO system, in which a massive MIMO BS with $M$ antennas  estimates the UL channels for  $K+1$ single-antenna users using $K+1$ pilot sequences of length $L$.  The CSI acquisition process can be summarized as follows. First, the $K+1$ users simultaneously transmit their corresponding length-$L$ pilot sequences, which are not necessarily orthogonal, over the same $L$ communication resource elements (e.g., subcarriers in the case of orthogonal frequency-division multiplexing). Then, the BS collects the $M \times L$ observations and estimates the UL channels for the $K+1$ users by means of a linear MMSE (LMMSE) channel estimator.  In the following, we discuss the channel model and we describe the CSI acquisition procedure in more detail.

\subsection{Channel Model}

\noindent Let us assume that the $M\times 1$ narrowband channel connecting the $k$-th single-antenna user with the $M$ BS antennas, $\h_k \in \Compl^M$, can be expressed as
\begin{align} \label{eq:channel_model}
\h_k &= \sqrt{\beta_k} \Sig_k^{1/2} \etav_k,  \qquad  & k=0,\ldots,K 
\end{align}
where $\etav_k \sim \CN(\0,\I_{r_k})$ models the small-scale fading process, $\beta_k>0$ denotes the pathloss, and $\Sig_k$ is the rank-$r_k$ $M \times M$ spatial covariance matrix of user $k$. Here we adopt the widely accepted ``windowed'' wide-sense stationary (WSS) fading channel model (see \cite{Adhikary_JSDM_IT2013,Yin14,Mir18} for details), which assumes that the small scale fading coefficients in $\{\etav_k\}$ are drawn independently and kept fixed during the channel coherence time $T_{\mathsf C}$, whereas the slow time-varying large-scale fading parameters (second-order statistics), $\{\beta_k, \Sig_k\}$ are considered to remain constant over a window  $T_{\mathsf WSS} \gg T_{\mathsf C}$. In consequence, the channel can be approximated as WSS inside this window and we can define
\begin{align} \label{eq:Rk}
 \Exp \{ \h_k \h_k^{\dagger} \} &= \beta_k \Sig_k = \beta_k \U_k \La_k \U_k^{\dagger},  & k=0,\ldots,K
\end{align}
where $\U_k = ( \u_{k,1}, \ldots, \u_{k,r_k})  \in \Compl^{M\times r_k}$ contains the eigenvectors associated with the non-zero eigenvalues of $\Sig_k$, $\lambda_{k,1} \geq \cdots \geq \lambda_{k,r_k}>0$,  and $\La_k = \Diag\big(\lambda_{k,1}, \ldots, \lambda_{k,r_k}\big)$. Additionally, we normalize the covariance matrices $\{\Sig_k\}$ to guarantee that
\begin{align} \label{eq:covariance_normalization}
\Tr( \Sig_k ) & = M, \qquad & k=0,\ldots,K.
\end{align}

\subsection{CSI Acquisition Model}

\noindent The training-based CSI acquisition strategy can be described as follows. We denote by $\P \in \Compl^{L \times (K+1)}$ the matrix gathering the $K+1$ length-$L$ pilot sequences $\mathbf{p}_0, \ldots, \mathbf{p}_K$ assigned to the $K+1$ users:
\begin{align}
\P = \big( \mathbf{p}_0, \ldots,  \mathbf{p}_K \big)
\end{align}
where $\p_k = \big( p_k(1), \ldots, p_k(L)\big)^{\tras}$ satisfies\footnote{Note that this power normalization is more realistic than the assumption that $\|\p\|^2=1$ independently of $L$ commonly adopted in the pilot design literature.}
\begin{align} \label{eq:pilot_normalization}
\frac{1}{L}\sum_{\ell=1}^L |p_k(\ell)|^2 &= 1, \qquad & \qquad k=0,\ldots,K. 
\end{align}
Letting all $K+1$ users simultaneously transmit their respective pilot sequence, the signal received by the BS at the $\ell$-th resource element, $\y(\ell) \in \Compl^{M}$, is
\begin{align} \label{received_signal_l}
\y(\ell) &= \sum_{k=0}^K \sqrt{P_k} \h_k p_k(\ell) + \n(\ell), \qquad &\ell = 1, \ldots, L
\end{align}
where $\{P_k\}$ are the transmit powers and $\n(\ell) = \big(n_1(\ell), \ldots, n_M(\ell)\big)^\tras \in \mathbb{C}^{M}$ denotes the additive white Gaussian noise with i.i.d. circularly symmetric components, $n_m(\ell)\sim \mathcal{CN}(0,\sigma^2)$ for $m = 1, \ldots,M$. Grouping the received signal for the $L$ resource elements dedicated to training in $\Y = \big( \y(1), \ldots, \y(L)\big) \in \Compl^{M\times L}$, the signal model in \eqref{received_signal_l} can be more compactly expressed as
\begin{align} \label{eq:CSI_model_1}
\Y = \Hc \D_{\boldsymbol{P}}^{1/2} \P^{\tras} + \N
\end{align}
with  $\Hc = \big( \h_0, \cdots,  \h_K \big) \in \Compl^{M\times (K+1)}$, $\D_{\boldsymbol{P}} = \Diag\big(P_0, \ldots, P_K\big)\in \Real_{+}^{(K+1)\times (K+1)}$, and $\N =  \big( \n(1), \cdots, \n(L) \big) \in \Compl^{M\times L}$. We can equivalently write 
\begin{align} \label{eq:observation}
\y = \vec(\Y) = \big(\P\D_{\boldsymbol{P}}^{1/2} \otimes \I_M\big)\vec(\Hc) + \vec(\N) = \tilde{\P} \tilde{\D}_{\boldsymbol{P}}^{1/2} \h +\n
\end{align}
where we have defined $\h = \vec(\Hc)= \big(\h_0^{\tras},\ldots, \h_K^{\tras} \big)^{\tras}$, $\tilde{\P} = \P \otimes \I_M$, and  $\tilde{\D}_{\boldsymbol{P}} = \D_{\boldsymbol{P}} \otimes \I_M$. 

Given the observation model in \eqref{eq:observation}, the BS estimates the individual channels from the $K+1$ users, adopting a LMMSE approach, which under the channel model in \eqref{eq:channel_model} is given by \cite[Chap.~12]{Kay93}
\begin{align}  \label{eq:hat_h}
\hat{\h} & = 
 \Sig \tilde{\D}_{\snr} \tilde{\D}_{\boldsymbol{P}}^{-1/2}\tilde{\P}^{\dagger}\big(\tilde{\P} \tilde{\D}_{\snr}^{1/2} \Sig \tilde{\D}_{\snr}^{1/2} \tilde{\P}^{\dagger} + \I_{LM}\big)^{-1} \y  
\end{align}
where we have introduced the received signal-to-noise ratios (SNRs), $\{\snr_k= \beta_k P_k/\sigma^2\}$, $\tilde{\D}_{\snr} = \Diag(\snr_0, \ldots, \snr_K) \otimes \I_M$, and $\Sig = \Diag \big(\Sig_0, \ldots, \Sig_K\big)$.
Observe that the estimator in \eqref{eq:hat_h} requires the knowledge of the second-order statistics of the $K+1$ individual channels. When this information is not available, the estimator in \eqref{eq:hat_h}  is substituted by a mismatched estimator, which has different accuracy depending on how much is assumed to be known from the channel model in \eqref{eq:channel_model}.  
In particular, we distinguish between the following cases:
\begin{itemize}
\item[\textsf{(i)}]\underline{Conventional CSI Acquisition Strategy:}  We assume that the BS either does not have or does not use the individual spatial covariance matrices and, hence, uses $\{\Sig_k = \I_M\}$. It knows the transmit powers $\{P_k\}$, and the received signal-to-noise ratios (SNRs) $\{\snr_k= \beta_k P_k/\sigma^2\}$ including the pathloss information. We also consider that the pilot set gathered in $\P$ is orthogonal, i.e., 
\begin{align} \label{eq:orthogonal_pilots}
\p_k^{\dagger}\p_{k^\prime} &=
\begin{cases}
L , & \qquad k = k^\prime \\ 
0, & \qquad   k\neq k^\prime \\
\end{cases},
& &k=0, \ldots K
\end{align}
which requires the pilot-length to satisfy $L \geq K+1$. Then, the channel estimator in \eqref{eq:hat_h} becomes the mismatched LMMSE estimator $\hat{\h}^{\mathsf{(i)}} = \big((\h^{\mathsf{(i)}}_0)^{\tras},\ldots, (\h^{\mathsf{(i)}}_K)^{\tras} \big)^{\tras}$ with $\h^{\mathsf{(i)}}_k$ given by
\begin{align}
\hat{\h}^{\mathsf{(i)}}_k &= \Big( \frac{\snr_k/ \sqrt{P_k}}{1+L\snr_k} \Big) \tilde{\P}_k^{\dagger} \y, \qquad &k=0,\ldots,K
\label{eq:hat_h_k_i}
\end{align}  
where $\tilde{\P}_k = \p_k \otimes \I_M$. This channel estimation technique coincides with the element-wise MMSE estimator in \cite[Sec.~3.4]{BjornsonHoydisSanguinetti} when the diagonal entries of the individual  spatial covariance matrices $\{\Sig_k\}$ are assumed to be 1.
This case is analyzed in Section \ref{sec:case_i}. 

\bigskip

\item[\textsf{(ii)}]\underline{Covariance-Aided CSI Acquisition Strategy:} We assume that the BS exploits the knowledge of the individual spatial covariance matrices $\{\Sig_k\}$, the transmit powers $\{P_k\}$, and the received SNRs $\{\snr_k\}$, during CSI acquisition, and uses an arbitrary (possibly non-orthogonal) pilot set $\P$, i.e., \eqref{eq:orthogonal_pilots} is not satisfied. The channel estimator in that case is directly obtained from \eqref{eq:hat_h}, that is,  $\hat{\h}^{\mathsf{(ii)}} = \big((\h^{\mathsf{(ii)}}_0)^{\tras},\ldots, (\h^{\mathsf{(ii)}}_K)^{\tras} \big)^{\tras}$ with $\hat{\h}^{\mathsf{(ii)}}_k$ given by 
\begin{align}
\label{eq:hat_h_k_ii}
\hat{\h}_k^{\mathsf{(ii)}} & =  
\big( \snr_k / \sqrt{P_k}\big) \Sig_k\tilde{\P}_k^{\dagger}\big(\tilde{\P}\tilde{\D}_{\snr}^{1/2} \Sig \tilde{\D}_{\snr}^{1/2}\tilde{\P}^{\dagger} + \I_{LM}\big)^{-1} \y , \qquad & k=0,\ldots,K.
\end{align}
This case is investigated in Sections~\ref{sec:case_ii} and \ref{sec:large_system_analysis}.
\end{itemize}

\subsection{CSI identifiability} \label{sec:identifiability}

\noindent Let us now present the identifiability conditions on the system parameters: the pilot length $L$ and the ranks of the individual covariance matrices $\{r_k\}$, which enable to identify the CSI vector $\h$ from the observations $\y$ in \eqref{eq:observation} in the noiseless case ($\n=\bf{0}$). Using an equation counting argument, we see  that CSI identifiability requires $\text{rank}(\tilde{\P}\tilde{\D}_{\P}\Sig^{1/2})=\sum_{k=0}^K r_k$ or, equivalently, that 
\begin{equation}\label{eq:identifiability}
\text{rank}(\tilde{\P}\U)=\sum_{k=0}^K r_k
\end{equation}
where $\U=\text{diag}(\U_1,\dots,\U_K)$, so that the system can be uniquely inverted. Since $\tilde{\P}\U$ has $ML$ rows and $\sum_{k=0}^K r_k$ columns, a necessary condition for CSI identifiability is 
\begin{equation}\label{eq:size_constraint}
\frac{1}{L}\sum_{k=0}^K \frac{r_k}{M} \leq 1.
\end{equation}
In particular, if all covariance matrices are full rank, $\{r_k = M\}$, the CSI is identifiable if and only if $\text{rank}(\P)=K+1$, i.e., it is necessary that $L\geq K+1$. On the contrary, if all the covariance matrices $\{\Sig_k\}$ span orthogonal subspaces, we necessarily have that  $\sum_{k=0}^K r_k \leq M$ and, hence, the CSI is identifiable for $L\geq 1$. Besides these two extreme cases, it is hard to establish identifiability conditions for general pilot sequences and user covariance matrices and this is exactly what complicates the MSE analysis of covariance-aided CSI acquisition. Still, thanks to the large-system analysis in Section \ref{sec:large_system_analysis}, we are able extract meaningful conclusions for the intermediate cases.

\section{MSE Analysis of CSI Acquisition Strategies} \label{sec:mse_analysis}

\noindent In this section we present the channel estimation MSE expressions for both the conventional and the covariance-aided CSI acquisition schemes. In particular, we characterize the performance of each CSI acquisition strategy by the channel estimation mean-square error (MSE) of a given user, denoted as user $0$. This incurs in no loss of generality and allows us to derive useful insights by considering the other users as interferers.

\subsection{Conventional CSI Acquisition} \label{sec:case_i}

\noindent Let us first focus on the conventional CSI acquisition strategy, which 
uses orthogonal pilots and applies the mismatched (covariance-agnostic) LMMSE channel estimator in \eqref{eq:hat_h_k_i}.
The channel estimation error covariance matrix in this case is given by 
\begin{align}
\C_{\e}^{(\mathsf{i})} \big( \{\Sig_k\}, \{\snr_k\} \big) & = \Exp\big\{ (\h - \hat{\h}^{\mathsf{(i)}})(\h - \hat{\h}^{\mathsf{(i)}})^{\dagger} \big\}  = \tilde{\D}_{\betav}^{1/2} \big(\I_{MK} + L \tilde{\D}_{\snr}\big)^{-1} \big(\Sig + L \tilde{\D}_{\snr}\big) \big(\I_{MK} + L\tilde{\D}_{\snr}\big)^{-1} \tilde{\D}_{\betav}^{1/2}
\end{align}
with $\tilde{\D}_{\betav}= \Diag(\beta_0, \ldots, \beta_K) \otimes \I_M$, and the individual error covariance matrix for user $0$ 
follows from the first $M \times M$ block of in the diagonal of $\C_{\e}^{(\mathsf{i})} \big( \{\Sig_k\}, \{\snr_k\} \big)$:
\begin{align} \label{eq:Ce_0_i}
\C_{\e_0}^{(\mathsf{i})} \big(\Sig_0, \snr_0 \big) & = \Exp\big\{ (\h_0 - \hat{\h}_0^{\mathsf{(i)}})(\h_0 - \hat{\h}_0^{\mathsf{(i)}})^{\dagger} \big\} = \frac{\beta_0}{(1+L\snr_0)^2} \big(\Sig_0 + L\snr_0\I_{M} \big).
\end{align}
The channel estimation MSE for user $0$ can be derived from \eqref{eq:Ce_0_i} as presented in the next lemma.
\medskip
\begin{lemma}\label{lem:mse_i_orthogonal} 
Let the pilot set $\P$ contain $K+1$ orthogonal pilot sequences of length $L\geq K+1$, so that \eqref{eq:orthogonal_pilots} holds. Then, the individual MSE of estimating $\h_0$ with the estimator $\hat{\h}^{\mathsf{(i)}}_0$ in \eqref{eq:hat_h_k_i} is given by 
\begin{align}\label{eq:sum_mse_i}
\mse_0^{(\mathsf{i})} (\snr_0) & \triangleq \frac{1}{M} \Tr \big( \C^{\mathsf{(i)}}_{\e_0} (\Sig_0, \snr_0) \big) = \frac{\beta_0}{1+L\snr_0}.
\end{align} 
\end{lemma}

\subsection{Covariance-Aided CSI Acquisition} \label{sec:case_ii}

\noindent Let us now analyze the case in which the spatial covariance matrices are exploited in the CSI acquisition process. In consequence, let us assume now that the BS knows the spatial covariance matrices $\{\Sig_k\}$, the transmit powers $\{P_k\}$, and the received SNRs, $\{\snr_k \}$, and estimates the channel using the LMMSE estimator in \eqref{eq:hat_h_k_ii}. Then, the channel estimation error covariance matrix 
is given by 
\begin{align}
\C_{\e}^{\mathsf{(ii)}}\big(\P, \{\Sig_k\}, \{\snr_k\} \big) 
& = \Exp\big\{ (\h - \hat{\h}^{\mathsf{(ii)}})(\h - \hat{\h}^{\mathsf{(ii)}})^{\dagger} \big\} \\
& = \tilde{\D}_{\betav}^{1/2} \big( \Sig - \Sig \tilde{\D}_{\snr}^{1/2}\tilde{\P}^{\dagger} \big( \tilde{\P}\tilde{\D}_{\snr}^{1/2}\Sig \tilde{\D}_{\snr}^{1/2}\tilde{\P}^{\dagger} + \I_{LM} \big)^{-1}\tilde{\P} \tilde{\D}_{\snr}^{1/2} \Sig \big) \tilde{\D}_{\betav}^{1/2} \label{eq:Ce_ii}
\end{align}
and, hence, the individual error covariance matrix for user $0$ is
\begin{align}
\C_{\e_0}^{\mathsf{(ii)}}\big(\P, \{\Sig_k\}, \{\snr_k\}\big) & =\Exp\big\{ (\h_0 - \hat{\h}_0^{\mathsf{(ii)}})(\h_0 - \hat{\h}_0^{\mathsf{(ii)}})^{\dagger} \big\} 
 = \beta_0 \big( \Sig_0 - \snr_0 \Sig_0 \tilde{\P}_0^{\dagger} \big( \tilde{\P}\tilde{\D}_{\snr}^{1/2}\Sig \tilde{\D}_{\snr}^{1/2}\tilde{\P}^{\dagger} + \I_{LM} \big)^{-1}\tilde{\P}_0 \Sig_0 \big).\label{eq:Ce_ii_k}
\end{align}
Finally, the channel estimation MSE for user $0$ in the covariance-aided case is 
\begin{align} \label{eq:sum_mse_ii_0}
\mse^{(\mathsf{ii})}_0 \big(\P,\{\Sig_k\}, \{\snr_k\}\big) &\triangleq \frac{\beta_0}{M}\Tr\big(\C_{\e_0}^{\mathsf{(ii)}}\big(\P, \{\Sig_k\}, \{\snr_k\}\big) \big)\\
& = \frac{1}{M} \sum_{i=1}^{r_0} \frac{\beta_0\lambda_{0,i}}{1+ \snr_0 \lambda_{0,i}  \label{eq:sum_mse_ii}
\big( \p_0 \otimes \u_{0,i}\big)^{\dagger}  \Big( \sum_{k=1}^K \snr_k \big( \p_k\p_k^{\dagger} \otimes \Sig_k \big)  + \I_{ML} \Big)^{-1} \big( \p_0 \otimes \u_{0,i}\big)} 
\end{align}
where we have used the covariance matrix decomposition in \eqref{eq:Rk} and Lemma~\ref{lem:inv} postponed in Appendix~\ref{app:pre_results}.

For convenience, let us now introduce the estimation signal-to-interference-plus-noise ratio (SINR) as follows. Recall from \eqref{eq:hat_h_k_ii} that $\hat{\h}^{\mathsf{(ii)}}_0 = \E_0^{\dagger}\y $ with $\E_0\triangleq\big( \snr_{0} / \sqrt{P_0}\big) \big(\tilde{\P}\tilde{\D}_{\snr}^{1/2} \Sig \tilde{\D}_{\snr}^{1/2}\tilde{\P}^{\dagger} + \I_{LM}\big)^{-1} \tilde{\P}_{0} \Sig_{0} $ and the received signal $\y$ as given in \eqref{eq:observation}. Thanks to the linearity of the estimator, we can identify the useful signal as the contribution originated from the transmission of $\p_0$ by user $0$, i.e., $\sqrt{P_0}\E_0^{\dagger} \tilde{\P}_0 \h_0$ and denote the rest as interference-plus-noise, i.e., $\E_0^{\dagger}\big(\y-\sqrt{P_0}\E_0^{\dagger} \tilde{\P}_0 \h_0\big)$. Accordingly, we define the estimation SINR measured in the subspace spanned by each eigenvector of $\Sig_0$ as the ratio of the expectation (with respect to the noise) of the two quantities, i.e.,
\begin{align} \label{eq:SINR_0}
\sinr_{0,i}^{(\mathsf{ii})} \big(\P,\{\Sig_k\}, \{\snr_k\}\big) & = \frac{  \u_{0,i}^{\dagger} \big(\snr_0 \E_0^{\dagger} \tilde{\P}_0\Sig_0 \tilde{\P}_0^{\dagger} \E_0\big) \u_{0,i} } { \u_{0,i}^{\dagger} \big( \E_0^{\dagger} \big( \sum_{k=1}^K \snr_k \tilde{\P}_k\Sig_k\tilde{\P}_k^{\dagger} + \I_{ML}\big)  \E_0 \big) \u_{0,i}} \\ \label{eq:SINR}
& = \snr_0 \lambda_{0,i}  \big( \p_0 \otimes \u_{0,i}\big)^{\dagger}  \Big( \sum_{k=1}^K \snr_k \big( \p_k\p_k^{\dagger} \otimes \Sig_k \big)  + \I_{ML} \Big)^{-1} \big( \p_0 \otimes \u_{0,i}\big), 
& \qquad i=1, \ldots, r_0
\end{align}
where we used again Lemma~\ref{lem:inv}. Identifying terms, we can now rewrite the MSE expression in  \eqref{eq:sum_mse_ii} as
\begin{align} \label{eq:sum_mse_ii_sinr}
\mse^{(\mathsf{ii})}_0 \big(\P,\{\Sig_k\}, \{\snr_k\}\big) = \frac{1}{M} \sum_{i=1}^{r_0} \frac{\beta_0\lambda_{0,i}}{1+ \sinr_{0,i}^{(\mathsf{ii})} \big(\P,\{\Sig_k\}, \{\snr_k\}\big)}.
\end{align}

It is interesting to observe (e.g., in \eqref{eq:SINR_0}) that $\sinr_{0,i}^{(\mathsf{ii})} \big(\P,\{\Sig_k\}, \{\snr_k\}\big) \leq L\snr_0 \lambda_{0,i}$, where the upper bound is achieved for $i=1, \ldots, r_0$, when the interference from the $K$ users is completely canceled by LMMSE channel estimator. In that case the MSE in the following lemma results.  
\medskip

\begin{lemma} \label{lem:mse_ii_orthogonal} 
Let one or both following conditions hold:
\begin{itemize}
  \item[\emph{\textbf{(a)}}] (Orthogonal pilot condition). Pilot sequence $\p_0$ is orthogonal to the rest of pilot sequences, i.e., 
  \begin{align} \label{eq:orthogonal_pilots_0}
    \p_0^{\dagger}\p_k & = 0,  \qquad & k = 1, \ldots, K
  \end{align}
  \item[\emph{\textbf{(b)}}] (Orthogonal covariance subspaces). The $r_0$-dimensional subspace spanned by the covariance matrix of user $0$ is orthogonal to the subspace spanned by the covariance matrices of the $K$ interfering users, i.e.,
  \begin{align} \label{eq:orthogonal_covs_0}
    \U_0^{\dagger}\Sig_k & = \0_{r_0},  \qquad & k = 1, \ldots, K.
  \end{align}   
\end{itemize}
Then, $\sinr_{0,i}^{(\mathsf{ii})} \big(\P,\{\Sig_k\}, \{\snr_k\}\big) = L\snr_0 \lambda_{0,i}$ and the individual MSE of estimating $\h_0$ with the estimator $\hat{\h}^{\mathsf{(ii)}}_0$ in \eqref{eq:hat_h_k_ii} is 
\begin{align} \label{eq:mse_ii_orthogonal}
\mse^{(\mathsf{ii})}_0 \big(\Sig_0, \snr_0\big) & \triangleq
\frac{1}{M} \sum_{i=1}^{r_0}\frac{\beta_0 \lambda_{0,i}   }{1+L \snr_0 \lambda_{0,i} } 
\end{align}
where  $\lambda_{0,1}\geq \cdots \geq \lambda_{0,r_0}>0$ are the non-zero eigenvalues of $\Sig_0$ as introduced in \eqref{eq:Rk}.
\end{lemma}
\medskip

Observing that $f(x) = \frac{x}{1+ax}$ is concave in $x>0$ for any $a>0$, we can apply Jensen's inequality to see that 
\begin{align} \label{eq:mse_ii_ub}
\mse^{(\mathsf{ii})}_0 \big(\Sig_0, \snr_0\big)   
\leq \frac{\beta_0 (\Tr(\Sig_0)/M)}{1 +L \snr_0 (\Tr(\Sig_0)/r_0)}  = \frac{\beta_0 }{1 +L \snr_0 (M/r_0) }
\end{align}
with equality when $\lambda_{0,1} = \cdots = \lambda_{0,r_0}$. We can conclude that, under the conditions of Lemma \ref{lem:mse_ii_orthogonal}, covariance-aided CSI acquisition strictly outperforms the conventional strategy with orthogonal pilots in a massive MIMO system, that is,  $\mse^{(\mathsf{ii})}_0 \big(\Sig_0, \snr_0\big) < \mse^{(\mathsf{i})}_0 \big(\snr_0\big)$, whenever the eigenvalues $\lambda_{0,1}, \dots,\lambda_{0,r_0}$ are not all equal and, in particular, when the spatial covariance matrix of user $0$ is not full-rank  ($r_0 < M$). This can be easily interpreted as follows. Both the conventional and the covariance-aided CSI strategies under the conditions in Lemma \ref{lem:mse_ii_orthogonal} effectively remove any interference caused by the pilots of the $K$ remaining users, so that they become purely noise-limited. The covariance-aided channel estimator in \eqref{eq:hat_h_k_ii} additionally removes all the noise outside the subspace spanned by $\Sig_0$ and this reduces the noise power at least by a factor of $M/r_0$, as confirmed by \eqref{eq:mse_ii_ub}.        



\section{Large-System Analysis of Covariance-Aided CSI Acquisition} \label{sec:large_system_analysis}

\noindent Given the difficulty of interpreting the effect of the pilot sequences $\P$ and the spatial covariance matrices $\{\Sig_k\}$ on the channel estimation MSE for the covariance-aided CSI acquisition strategy as given in \eqref{eq:sum_mse_ii}, in this section we adopt a large-scale analysis approach. Indeed, we study $\mse^{(\mathsf{ii})}_0 \big(\P,\{\Sig_k\}, \{\snr_k\}\big)$ asymptotically in $K, L, M$ with given fixed ratios under different assumptions on $\{ \Sig_k\}$ and  $\P \in \Compl^{L \times (K+1)}$ and derive the corresponding deterministic equivalents as summarized in Table~\ref{tab:summary}. Deterministic equivalents provide asymptotically tight deterministic approximations for the MSE, which allow us to decouple the effects of the non-orthogonality of the pilots and the non-orthogonality of the covariance matrices.

\subsection{Deterministic Equivalent for Random Covariance Matrices and Deterministic Pilots} \label{sec:random_covariances}

\noindent Let us assume in this section that the pilot length $L$ and the number of interfering users $K$ are finite and that the individual covariance matrix of user $0$, $\Sig_0$,  and the pilots sequences $\P$ are deterministic. Furthermore, the individual spatial covariance matrices of the interfering users are assumed to be drawn from a random distribution.   

Besides the rank-deficiency, little information is available about the distribution of realistic spatial covariance matrices in massive MIMO, as few experimental works have concentrated specifically on that point (it would require to measure the channel with different array geometries and over different scenarios). Therefore, we adopt the maximum entropy principle \cite{Debbah_maxent}, and observe that among all the distributions over rank-$r_k$ positive semidefinite matrices of a given size and trace, the Wishart distribution with $r_k$ degrees of freedom and column covariance matrix proportional to the identity is the one that has maximum differential entropy \cite[Section~18.2.2.3]{debbahcouillet}. Accordingly, we  assume  the  following (general) random model for the individual spatial covariance matrices.


\begin{table}[t!] 
  {\centering \small
   \begin{tabular}{m{1.0cm} m{3.7cm} m{2.3cm} m{1.7cm} m{7.1cm}}
  \hline \addlinespace[.2cm] 
   &\textbf{Deterministic Equivalent \newline to the estimation MSE} & \textbf{Covariance\newline Matrices}  & \textbf{Pilots} & \textbf{Asymptotic Regime}  \\
  \addlinespace[.2cm] \hline \addlinespace[.1cm] 
  Thm.~\ref{thm:deteq1} & $\xi_{0}^{(\mathsf{ii})}\big(\P, \Sig_0, \{\snr_k\} \big)$ & $\Sig_0$ deterministic \newline $\{\Sig_k\}$ random  & deterministic & $\sbullet M, r_k\rightarrow \infty$ with ratios $\tau_{M,k}=r_k/M$, where\newline  $\phantom{\sbullet}$$0<\lim\inf_{M,r_k} \tau_{M,k}$ \newline $\sbullet K,L$ finite \\
  \addlinespace[.1cm]  \hline \addlinespace[.1cm] 
  Thm.~\ref{thm:deteq3} & $\xi_{0}^{(\mathsf{ii})}\big(\{\Sig_k\}, \{\snr_k\}; \Gam_L\big)$ & $\Sig_0$,  $\{\Sig_k\}$  \newline deterministic & random & $\sbullet M, r_k$ finite \newline $\sbullet K,L \rightarrow \infty$ with ratio $\alpha_L = L/K$, where \newline $\phantom{\sbullet}$$0<\lim\inf_{K,L}\alpha_L$ \\
   \addlinespace[.1cm]  \hline \addlinespace[.1cm] 
  Thm.~\ref{thm:deteq3_uni} & $\xi_{0}^{(\mathsf{ii})}\big(\{\Sig_k\}, \{\snr_k\}; \Gam_L\big)$ & $\Sig_0$,  $\{\Sig_k\}$  \newline deterministic & random & $K,L,M,r_k\rightarrow \infty$ with ratios: \newline $\sbullet K,L \rightarrow \infty$ with ratio $\alpha_L = L/K$,  where \newline$\phantom{\sbullet}$$0<\lim\inf_{K,L}\alpha_L$, and \newline $\sbullet$there exists $0<\delta<1$ such that \newline $\phantom{\sbullet}$$\lim\sup_{M,L} ML^{-\delta}<\infty$\\
  \addlinespace[.1cm]  \hline \addlinespace[.1cm] 
  Thm.~\ref{thm:deteq4} & $\bar{\xi}_{0}^{(\mathsf{ii})}\big(\Sig_0, \{\snr_k\} ; \gamma_L \big)$ & $\Sig_0$ deterministic \newline $\{\Sig_k\}$ random & random & $K,L,M,r_k\rightarrow \infty$ with ratios: \newline $\sbullet \tau_{M,k}=r_k/M$, where \newline$\phantom{\sbullet}$$0 < \lim \inf_{M,r_k} \tau_{M,k} \leq \lim \sup_{M,r_k} \tau_{M,k}< 1$  \newline $\sbullet \alpha_L = L/K$, where $0<\lim\inf_{K,L} \alpha_L$, and
  \newline $\sbullet$there exists $\frac{2}{3}\leq\delta<1$ such that \newline $\phantom{\sbullet}$$0<\lim\inf_{M,L} ML^{-\delta}\leq \lim\sup_{M,L} ML^{-\delta}<\infty$\\
   \addlinespace[.1cm]  \hline  
  \end{tabular}}
  \caption[]{Summary of asymptotic results for $\mse^{(\mathsf{ii})}_0 \big(\P,\{\Sig_k\}, \{\snr_k\}\big)$.\footnotemark}
  \label{tab:summary}
  \end{table}
  
  \footnotetext{Note that the assumptions on the asymptotic regime provided in this table are standard with the exception of the existence of some $\delta<1$ such that $\lim\sup_{M,L} ML^{-\delta}<\infty$, which prevents from considering the case of $M/L$ being bounded. However, we conjecture that this assumption could be relaxed without changing our results.}

\begin{assumption}[\textbf{Random covariance model}]
\label{ass:covmodel}
The individual covariance matrices of the $K$ interfering users are assumed to be random according to:
\begin{align} \label{ass:sig_k}
\Sig_{k} & = \frac{1}{r_k}\sum_{i=1}^{r_k} \x_{k,i}\x_{k,i}^\dagger, & \qquad k=1,\ldots, K
\end{align} 
where $r_k$ is the rank of $\Sig_{k}$ with probability $1$ and the entries of $\x_{k,i}$ are i.i.d.~with zero-mean, unit variance, and have finite eighth order moment. Observe that the covariance matrix normalization in \eqref{eq:covariance_normalization} is now satisfied in expectation. 
\end{assumption}

\noindent Note that Assumption~\ref{ass:covmodel} is in fact more general than the maximum entropy covariance matrix assumption. The (entropy maximizing) Wishart distribution is obtained by adding the Gaussianity assumption to \ref{ass:covmodel}.


\noindent Then, the deterministic equivalent of the MSE of user $0$ as $M, \{r_k\}\rightarrow \infty$ given in the next theorem follows.
\medskip
\begin{theorem}
\label{thm:deteq1}
Assume that the individual covariance matrices of the interfering users, $\{\Sig_k\}$, follow Assumption~{\rm\ref{ass:covmodel}}. Let $\lambda_{0,1}\geq \cdots \geq \lambda_{0,r_0}>0$ denote the non-zero eigenvalues of $\Sig_0$ and define
\begin{equation} \label{eq:mse_ii_eq}
\xi_{0}^{(\mathsf{ii})}\big(\P, \Sig_0, \{\snr_k\} \big)  \triangleq
\frac{1}{M}\sum_{i=1}^{r_0}\frac{\beta_0 \lambda_{0,i}}{1+\snr_0 \lambda_{0,i}  \p_0^\dagger\S_{L}^{-1}\p_0}
\end{equation}
where
\begin{equation}\label{eq:SM}
\S_{L} = \sum_{k=1}^K \frac{\snr_k }{1+\iota_{M,k}}\p_{k}\p_{k}^\dagger+\I_{L}
\end{equation}
and the constants $\iota_{M,1},\dots,\iota_{M,K}$ are given by the unique nonnegative solutions to the following fixed point equations:
\begin{align} \label{eq:c_k_M}
\iota_{M,k} & = \frac{\snr_k}{\tau_{M,k}} \Big(\p_{k}^\dagger\S_{L}^{-1}\p_{k} 
-\frac{1}{M}\sum_{i=1}^{r_0} \frac{\snr_0 \lambda_{0,i}  |\p_{k}^\dagger\S_{L}^{-1}\p_0|^2}{1+\snr_0  \lambda_{0,i} \p_0^\dagger\S_{L}^{-1}\p_0}\Big),& k = 1,\ldots,K.
\end{align}
Then, as $M$ and $r_k$ grow large with ratios $\tau_{M,k}\triangleq r_k/M$ such that $0< \lim\inf_{M,r_k} \tau_{M,k}$, we have that  
\begin{equation} \label{eq:convergence_deteq1}
\mse_{0}^{(\mathsf{ii})}\big(\P, \{\Sig_k\}, \{\snr_k\}\big) - \xi_{0}^{(\mathsf{ii})}\big(\P, \Sig_0, \{\snr_k\}\big) \xrightarrow[M\rightarrow \infty]{\mathrm{a.s.}} 0.
\end{equation}
\end{theorem}
\vspace*{-8pt}
\noindent \emph{Proof.} The proof is mainly based on an application of \cite[Thm.~1]{Wag12}. See Appendix~\ref{app:deteq1}.
\medskip

The result in Theorem \ref{thm:deteq1} is useful to see the effect of using non-orthogonal pilots on the CSI acquisition accuracy. Note that the constant $\iota_{M,k}$ measures the level of interference created by the pilot sequence of user $k$ when estimating the channel of user $0$ (the larger $\iota_{M,k}$, the lower the interference). In particular, if the pilots are orthogonal (see condition \emph{\textbf{(a)}} in Lemma $\ref{lem:mse_ii_orthogonal}$), it holds that $\p_0^\dagger\S_L^{-1}\p_0=\p_0^\dagger\p_0 =  L$ and the deterministic equivalent in \eqref{eq:mse_ii_eq} becomes the MSE given in Lemma $\ref{lem:mse_ii_orthogonal}$ for the interference-free case.

\subsection{Deterministic Equivalent for Deterministic Covariance Matrices and Random Pilots} \label{sec:random_pilots}

\noindent  Let us now assume that all the covariance matrices $\{\Sig_k\}$ are deterministic and consider the following random model for the pilot sequences.
\medskip
\begin{assumption}[\textbf{Random pilot model}]
\label{ass:pilotmodel}
The random length-$L$  pilot sequences are of the form
\begin{align} \label{eq:random_pilots}
 \p_k &= (p_k(1),\dots, p_k(L))^T, & k =0, \ldots, K
\end{align}
with $\p_0, \ldots, \p_K \in \Compl^{L}$ being independent random vectors with i.i.d.~entries of zero-mean, unit variance, and have finite eighth order moment. Observe that the normalization in \eqref{eq:pilot_normalization} is now satisfied in expectation.
\end{assumption}
\noindent Furthermore, we assume that the pilot length $L$ is sufficient to estimate the subspace spanned by the covariance matrices of all users $\{\Sig_k\}$ as formalized next.
\medskip
\begin{assumption}[\textbf{Pilot length}]\label{ass:enough_pilots}
There exists $\nu_M<1$ such that
\begin{equation}\label{eq:enough_pilots}
\lim\sup_{K,L}\Big\|\frac{1}{L}\sum_{k=0}^K \U_k\U_k^\dagger\Big\|\leq \nu_M \text{\ \ a.s.}
\end{equation}
with $\U_k \in \Compl^{M\times r_k}$ containing the $r_k$ eigenvectors associated with the non-zero eigenvalues of $\Sig_k$.
\end{assumption}

\noindent Indeed, Assumption~\ref{ass:enough_pilots} is related to the system identifiability discussed in Section \ref{sec:identifiability}. If \eqref{eq:enough_pilots} holds, one immediate consequence is that
\begin{equation}
\lim\sup_{K,L}\frac{1}{L}\sum_{k=0}^K r_k \leq M\Big\|\frac{1}{L}\sum_{k=0}^K \U_k\U_k^\dagger\Big\|\leq \nu_MM \leq M
\label{eq:enough_pilots2}
\end{equation}
which is a necessary condition for asymptotic identifiability (see \eqref{eq:size_constraint}). 

For convenience, let us introduce the block version of the trace operator for block matrices (see \cite{Fil98_blktrace} for details) before presenting the deterministic equivalent of the MSE error of user $0$ as $L,K \rightarrow \infty$ in the next theorem.
\medskip
\begin{definition} \label{def:blk_trace}
Consider a matrix $\B_{M,L}\in\mathbb{C}^{ML\times ML}$ composed of blocks $(\B_M^{(i,j)})_{1\leq i,j\leq L}$ of size $M\times M$, i.e.
\begin{equation}
\B_{M,L} = \left(\begin{array}{ccc}
\B_M^{(1,1)} & \dots & \B_M^{(1,L)}\\
\vdots & \ddots & \vdots\\
\B_M^{(L,1)} & \dots & \B_M^{(L,L)}\\
\end{array}\right).
\end{equation}
The block-trace of $\B_{M,L}$ is defined as\footnote{Note that this definition is highly dependent on the size of the blocks $M$. We omit, however, the reference to $M$ in the notation for the sake of simplicity.} 
\begin{equation}\label{eq:block_trace}
\mathrm{blktr}[\B_{M,L}] = \sum_{i=1}^L \B_M^{(i,i)}\in\mathbb{C}^{M \times M}.
\end{equation}
\end{definition}
\vspace*{-1pt}
\begin{theorem}\label{thm:deteq3} 
Assume that all $K+1$ pilot sequences are random pilots satisfying Assumption~{\rm\ref{ass:pilotmodel}} and the pilot-length is such that  Assumption {\rm \ref{ass:enough_pilots}} holds. Let $\lambda_{0,1}\geq \cdots \geq \lambda_{0,r_0}>0$ denote the non-zero eigenvalues of $\Sig_0$ and $\u_{0,1},\dots,\u_{0,r_0}$ the corresponding eigenvectors and define
\begin{align} \label{eq:xi_0_deteq3}
\xi_{0}^{(\mathsf{ii})}\big(\{\Sig_k\}, \{\snr_k\}; \Gam_L\big) 
\triangleq
\frac{1}{M}\sum_{i=1}^{r_0}\frac{\beta_0 \lambda_{0,i}}{1+ L\snr_0 \lambda_{0,i} \u_{0,i}^{\dagger}(\Gam_L + \I_M)^{-1} \u_{0,i} }
\end{align}
where $\Gam_L$ is the unique $M\times M$ positive definite matrix solution to the following fixed point equation
\begin{align} \label{eq:gam_fp}
\Gam_L &= \sum_{k=1}^K  \snr_k\Sig_k\big(\I_M+L \snr_k(\Gam_L+\I_M)^{-1}\Sig_k\big)^{-1}.
\intertext{Define}
\label{eq:defA}
\A_L &\triangleq \frac{1}{L}\mathrm{blktr}\Big[\Big(\sum_{k=1}^K \snr_k\tilde{\P}_{k} \Sig_k\tilde{\P}_{k} ^\dagger + \I_{ML}\Big)^{-1}\Big]
\end{align}
where the block-trace operator is introduced in Definition {\rm \ref{def:blk_trace}}. 
Then, as $L$ and $K$ grow large with ratio $\alpha_L \triangleq L/K$ such that $\lim\inf_{K,L}\alpha_L>0$ and whenever there exists $\epsilon_M>0$ such that
\begin{align} \label{eq:lambda_max_Gam_L}
\lim\inf_{K,L} \lambda_{\min}(\A_L)\geq\epsilon_M \text{\ \ a.s.}
\end{align}
we have that
\begin{equation} \label{eq:convergence_deteq3}
L\left|\mse_{0}^{(\mathsf{ii})}\big(\P, \{\Sig_k \}, \{\snr_k\}\big) - \xi_{0}^{(\mathsf{ii})}\big(\{\Sig_k\}, \{\snr_k\}; \Gam_L\big) \right|
  \xrightarrow[K,L\rightarrow\infty]{\mathrm{a.s.}} 0 .
\end{equation} 
\end{theorem}
\vspace*{-8pt}
\noindent \emph{Proof.} The idea of the proof is to generalize the result of Bai and Silverstein \cite{Sil95} to block-matrices in order to provide a deterministic equivalent for $\A_L$. The main step consists in using the equivalent of a rank-1 perturbation for block matrices and proving that the fixed point mapping underlying \eqref{eq:gam_fp} is satisfied by $\A_L$ and that it is indeed a contraction. As a key ingredient, we generalize the trace lemma to block-matrices with convenient random matrix concentration inequalities in Proposition~\ref{prop:block_trace} of Appendix \ref{app:pre_results}. See the detailed proof of Theorem \ref{thm:deteq3} in Appendix \ref{app:deteq3}.
\medskip

The result in Theorem \ref{thm:deteq3} is useful to understand the effect of the relative orthogonality between the subspace spanned by the covariance matrix of user 0, $\Sig_0 = \sum_{i=1}^{r_0} \lambda_{0,i} \u_{0,i} \u_{0,i}^{\dagger}$, and the subspace spanned by the covariance matrices of the $K$ interfering  users, as captured by the estimation SINR term $L\snr_0\lambda_{0,i}\u_{0,i}^{\dagger}(\Gam_L + \I_M)^{-1} \u_{0,i}$ in \eqref{eq:xi_0_deteq3}. In particular, when the subspace spanned by the covariance matrix of user $0$ is orthogonal to the subspace spanned by the other users, i.e., $\u_{0,i}^{\dagger} \Big(\sum_{k=1}^K  \U_k\U_k^{\dagger} \Big)\u_{0,i}=0$ for $i=1,\ldots,r_0$ (see condition \emph{\textbf{(b)}} in Lemma $\ref{lem:mse_ii_orthogonal}$), it follows from \eqref{eq:gam_fp} that the deterministic equivalent in \eqref{eq:xi_0_deteq3} becomes the MSE given in Lemma $\ref{lem:mse_ii_orthogonal}$. The reason is that the projection into the subspace of user $0$ in the channel estimator in \eqref{eq:hat_h_k_ii} already cancels all interference from the other users and orthogonal pilots are no longer needed. Moreover, using the fixed point equation in \eqref{eq:gam_fp}, the SINR can be upper bounded using the inequality
\begin{align}
 \u_{0,i}^{\dagger}(\Gam_L + \I_M)^{-1} \u_{0,i}  \leq  \u_{0,i}^{\dagger} \Big(   \sum_{k=1}^K  \snr_k\Sig_k\big(\I_M+L \snr_k\Sig_k\big)^{-1} + \I_M \Big)^{-1} \u_{0,i} 
& \leq \u_{0,i}^{\dagger} \Big(   \frac{1}{L}\sum_{k=1}^K  \U_k\U_k^{\dagger} + \I_M \Big)^{-1} \u_{0,i}  \label{eq:interf_bound_1}\\
& \leq 1 - \Big(\frac{1-\nu_M}{L}\Big) \u_{0,i}^{\dagger} \Big(\sum_{k=1}^K  \U_k\U_k^{\dagger} \Big)\u_{0,i} \label{eq:interf_bound_2}
\end{align}
where in \eqref{eq:interf_bound_1} we use that $(\Gam_L+\I_M)^{-1} \preceq \I_M$ and that $(\snr_k^{-1}\La^{-1}_k+ L\I_{r_k})^{-1} \preceq \frac{1}{L}\I_{r_k}$,\footnote{Let $\A,\B$ be $N \times N$ Hermitian matrices. We say that $\A \preceq \B$, if $\B -\A$ is positive semidefinite.} and \eqref{eq:interf_bound_2} follows from using the Taylor expansion of $\big(\frac{1}{L}\sum_{k=1}^K  \U_k\U_k^{\dagger} + \I_M \big)^{-1}$ under the assumption in \eqref{eq:enough_pilots}. Therefore, the larger the projection of $\u_{0,i}$ on the interferer covariance matrix, the lower the SINR.

In order to extend the results of Theorem \ref{thm:deteq3} for the case $M, \{r_k\}\rightarrow \infty$, we need to ensure that the convergence in \eqref{eq:convergence_deteq3} holds uniformly in $M$. This is done in the following theorem. 
\medskip
\begin{theorem}\label{thm:deteq3_uni} 
Assume that all $K+1$ pilot sequences are random pilots satisfying Assumption~{\rm\ref{ass:pilotmodel}} and the pilot-length is such that Assumption {\rm \ref{ass:enough_pilots}} holds uniformly in $M$. Assume that pilot $\p_0$ is uniformly bounded, i.e., there exists $\chi_{0}>0$ such that for any $\ell$ it holds $|p_{0}(\ell)|^2<\chi_{0}$ a.s.\footnote{This assumption is always satisfied in practice, since the maximum transmit power is always limited.}  If, furthermore, $\epsilon_M$ in \eqref{eq:lambda_max_Gam_L} is independent from $M$ and there exists $0<\delta<1$ such that $\lim\sup_{L,M} ML^{-\delta}<\infty \text{\ \ a.s.}$, the convergence of $\mse_{0}^{(\mathsf{ii})}\big(\P, \{\Sig_k \}, \{\snr_k\}\big)$ to the deterministic equivalent $\xi_{0}^{(\mathsf{ii})}\big(\{\Sig_k\}, \{\snr_k\}; \Gam_L\big)$ introduced in Theorem {\rm \ref{thm:deteq3}} is uniform in $M$.
\end{theorem}
\vspace*{-8pt}
\noindent \emph{Proof.} See Appendix \ref{app:deteq3_uni}.
\medskip

Finally, recall that the deterministic equivalent in Theorems \ref{thm:deteq3} (and \ref{thm:deteq3_uni}) only holds under some technical condition given in \eqref{eq:lambda_max_Gam_L} on matrix $\A_L$ defined in \eqref{eq:defA}. For completeness we provide in the following proposition two alternative sufficient conditions on the system parameters for guaranteeing the required result.

\medskip
\begin{proposition} \label{prop:Gam_L}
Let $\A_L$ be defined as in \eqref{eq:defA}, and let the pilot sequence $\p_0$ satisfy Assumption~{\rm\ref{ass:pilotmodel}}. Then, there exists $\epsilon_M>0$ such that $\lim\inf_{K,L} \lambda_{\min}(\A_L)\geq \epsilon_M$ a.s., whenever one of the following conditions hold: 
\begin{itemize}
\item[\emph{\textbf{(a)}}] (Summable received SNRs condition). The covariance matrices $\{\Sig_k\}$ have uniformly bounded spectral norm and
\begin{align} \label{eq:sum_snr}
\sum_{k=0}^K \snr_k < \infty.
\end{align}
\item[\emph{\textbf{(b)}}] (Strong subspace identifiability condition).
There exists a constant $c_M>0$ such that
\begin{equation} \label{eq:Gam_L_cond_b}
\frac{1}{K+1}\U^\dagger\tilde{\P}^\dagger\tilde{\P}\U \succeq c_M \I_{\sum_{k=0}^K r_k}\text{\ \ a.s.}
\end{equation}
where $\U = \Diag \big(\U_0,\U_1,\dots,\U_K\big) \in \Compl^{(K+1)M \times \sum_{k=0}^K r_k}$ gathers the eigenvectors of the individual covariance matrices of the $K+1$ users as defined in \eqref{eq:Rk}.
\end{itemize}
Moreover, if \textbf{(a)} holds or if $c_M$ in \textbf{(b)} is independent of $M$, then $\epsilon_M$ is also independent of $M$.
\end{proposition}
\vspace*{-8pt}
\noindent \emph{Proof.} See Appendix \ref{app:A_L}.
\medskip

Observe that the conditions in \emph{\textbf{(a)}} imply that the power of the interference perceived by user $0$ from the other users (which is proportional to $\{\snr_k\}$) decreases fast enough, so that the interference level is controlled independently of the particular covariance eigenspaces. Alternatively, condition \emph{\textbf{(b)}} ensures, without imposing any restriction on the power of the interfering users, that random pilots are good enough for CSI acquisition given the relative orthogonality between the covariance eigenspaces of all users in the system. In fact, \eqref{eq:Gam_L_cond_b} is slightly stronger than the channel identifiability condition, which consists in assuming that $\U^\dagger\tilde{\P}^\dagger\tilde{\P}\U \succ \bf{0}$ (see Section~\ref{sec:identifiability}). Note that if $\lim\sup \frac{L}{K}>1$, \cite[Thm.~1.1]{Bai98} guarantees the existence of some constant $\alpha>0$ such that $\frac{1}{K+1}\tilde{\P}^\dagger\tilde{\P}\succeq \alpha \I_{ML}$ a.s., so that condition \emph{\textbf{(b)}}
in \eqref{eq:Gam_L_cond_b} is satisfied in that case.


\subsection{Deterministic Equivalent for Random Covariance Matrices and Random Pilots} 
\label{sec:random_covariances_random_pilots}
\noindent  Let us now focus on the case in which the covariance matrices $\{\Sig_k\}$ follow the random model in Assumption {\rm\ref{ass:covmodel}} and the pilot sequences follow the random model in Assumption {\rm\ref{ass:pilotmodel}}. To this end, we particularize the results in Section \ref{sec:random_pilots} for random covariance matrices. In the following proposition we give a deterministic equivalent in Frobenius norm for the fixed point matrix $\Gam_L$ in Theorem~\ref{thm:deteq3}  when $K,L,M \rightarrow \infty$.
\medskip
\begin{proposition} \label{prop:Gam_L2}
Assume that the individual covariance matrices of the interfering users $\{\Sig_k\}$ satisfy Assumption~{\rm\ref{ass:covmodel}}. Let $\Gam_L$ be the solution of the fixed point equation in \eqref{eq:gam_fp} and $\gamma_L$ be the unique solution to the following the fixed point equation
\begin{align} \label{eq:gamma}
\gamma_L= \frac{1}{L}\sum_{k=1}^K \frac{2\tau_{M,k}(1+\gamma_L)}{1+\tau_{M,k}+\frac{\tau_{M,k}(1+\gamma_L)}{L\snr_k}+\sqrt{\big(1+\tau_{M,k}+\frac{\tau_{M,k}(1+\gamma_L)}{L\snr_k}\big)^2-4\tau_{M,k}}}.
\end{align}
Further assume that Assumption {\rm \ref{ass:enough_pilots}} holds uniformly in $M$ and define $\alpha_L \triangleq L/K$ and $\tau_{M,k} \triangleq r_k/M$. Then, as $K$ and $L$ grow large with ratio $\alpha_L$, and $M$ and $r_k$ grow large with ratios $\tau_{M,k}$ satisfying $0< \lim\inf_{M,r_k} \tau_{M,k}\leq  \lim\sup_{M,r_k} \tau_{M,k}< 1$ for $k = 0, \ldots, K$, it holds that
\begin{align} \label{eq:gamma_L_bounds}
1 - \frac{\frac{1}{K}\sum_{k=1}^K \tau_{M,k}}{\alpha_L}\leq (1+\gamma_L)^{-1}\leq 1.
\end{align}
Then, when  $\lim\sup_k \snr_k<\infty$ and $\lim\inf_{L,M} ML^{-2/3}>0$, we have that
\begin{equation}
\frac{1}{M}\Big\|(\Gam_L+\I_M)^{-1}-(1+\gamma_L)^{-1}\I_M\Big\|_{F}^2\xrightarrow[K,L,M\rightarrow\infty]{\mathrm{a.s.}} 0.
\end{equation}
\end{proposition}
\vspace*{-8pt}
\noindent \emph{Proof.} See Appendix \ref{app:Gam_L2}.

\medskip

We are now in the position to present in the next theorem the deterministic equivalent for the MSE of user $0$ when the individual covariance matrices of the $K$ interfering users and all the $K+1$ pilot sequences are random and as $K$, $L$ and $M$ grow large with some fixed ratios.

\begin{theorem}\label{thm:deteq4}
Assume that the individual covariance matrices of the interfering users $\{\Sig_k\}$  satisfy Assumption~{\rm\ref{ass:covmodel}} and that all $K+1$ length-$L$  pilot sequences are random pilots  satisfying Assumption~{\rm\ref{ass:pilotmodel}}. Let $\lambda_{0,1}\geq \cdots \geq \lambda_{0,r_0}>0$ denote the non-zero eigenvalues of $\Sig_0$ and define
\begin{align} \label{eq:xi_0_deteq4}
\bar{\xi}_{0}^{(\mathsf{ii})}\big(\Sig_0, \{\snr_k\} ; \gamma_L \big) \triangleq \frac{1}{M}\sum_{i=1}^{r_0} \frac{\beta_0 \lambda_{i,0}  }{1+ (1+\gamma_L)^{-1} L  \snr_0\lambda_{i,0}}
\end{align}
with $\gamma_L$ being the unique solution to the fixed point equation in \eqref{eq:gamma}. 
Then, under the conditions of Theorem {\rm \ref{thm:deteq3_uni}} and in the asymptotic regime of Proposition {\rm \ref{prop:Gam_L2}}, we have that
\begin{equation}
L\Big|\mse_{0}^{(\mathsf{ii})}\big(\P, \{\Sig_k \}, \{\snr_k\} \big)  - \bar{\xi}_{0}^{(\mathsf{ii})}\big(\Sig_0, \{\snr_k\}; \gamma_L \big)\Big|
  \xrightarrow[K,L,M\rightarrow\infty]{\mathrm{a.s.}} 0.
\end{equation}
\end{theorem}
\vspace*{-8pt}
\noindent \emph{Proof.} See Appendix \ref{app:deteq4}.

\medskip


Observe that the deterministic equivalent for the MSE of user $0$ in \eqref{eq:xi_0_deteq4} has a very similar expression to the MSE given in Lemma \ref{lem:mse_ii_orthogonal}. More precisely, the LMMSE estimator with non-orthogonal  pilots of length $L$ in the large-system regime becomes equivalent to the LMMSE estimator with orthogonal pilots of length $(1+\gamma_L)^{-1} L$ $($see Lemma~$\ref{lem:mse_ii_orthogonal})$, where $\gamma_L$ is the solution of \eqref{eq:gamma}. Thus, the effect of pilot contamination can be understood as an effective reduction of the estimation SINRs (see definition in \eqref{eq:SINR}) from $L\snr_0\lambda_{0,i}$ to $(1+\gamma_L)^{-1} L\snr_0\lambda_{0,i}$, where  $(1+\gamma_L)^{-1}$ satisfies the bounds in \eqref{eq:gamma_L_bounds}, or, following the discussion in \cite[Sec.~3.2]{BjornsonHoydisSanguinetti} as an effective reduction of the pilot processing gain. Accordingly, we can define the equivalent pilot processing gain/length  of random non-orthogonal pilots as $(1+\gamma_L)^{-1} L$.    

It is very interesting to investigate under which conditions on the system parameters, more exactly, on the received SNRs $\{\snr_k\}$ and the covariance ranks $\{r_k\}$ of the interfering users, these limiting values are attained. 

\medskip

\begin{corollary}\label{cor:deteq4}
Define
\begin{equation} \label{eq:de}
\bar{\xi}_{0}^{(\mathsf{ii})}\big(\Sig_0, \snr_0; \gamma_\infty \big) \triangleq \frac{1}{M}\sum_{i=1}^{r_0} \frac{\beta_0 \lambda_{i,0} }{1+(1+\gamma_\infty)^{-1}  L  \snr_0\lambda_{i,0}}.
\end{equation}  
Then, under the conditions of Theorem $\ref{thm:deteq4}$, it holds that
\begin{equation}
L\Big|\mse_{0}^{(\mathsf{ii})}\big(\P, \{\Sig_k \}, \{\snr_k\}\big) -  \bar{\xi}_{0}^{(\mathsf{ii})}\big(\Sig_0, \snr_0; \gamma_\infty \big) \Big|
\xrightarrow[K,L,M\rightarrow\infty]{\mathrm{a.s.}} 0
\end{equation} 
where constant $\gamma_{\infty}$ takes the following values:
\begin{itemize}
\item[\emph{\textbf{(a)}}] $\gamma_{\infty} = 0$, whenever
\begin{equation} \label{eq:beta_far_ues}
\frac{1}{K} \sum_{k=1}^K \Big(\frac{1}{1 - \frac{\bar{\tau}_M}{\alpha_L} +\frac{1}{ L\snr_k}}\Big)\xrightarrow[K,L\rightarrow\infty]{} 0;
\end{equation}
\item[\emph{\textbf{(b)}}] $\gamma_{\infty} = \frac{\bar{\tau}_M}{\alpha_L-\bar{\tau}_M}$ with $\bar{\tau}_M =\frac{1}{K}\sum_{k=1}^K \tau_{M,k}$, whenever
\begin{equation} \label{eq:beta_close_ues}
\frac{1}{K}\sum_{k=1}^K \Big(\frac{1}{L \snr_k}\Big)^{1/2}\xrightarrow[K,L\rightarrow\infty]{} 0.
\end{equation}
\end{itemize}
\end{corollary}
\vspace*{-8pt}
\noindent \emph{Proof.} See Appendix \ref{app:deteq4bis}.
\medskip

The conditions in Corollary \ref{cor:deteq4} can be interpreted as sufficient conditions for the system to be either \emph{\textbf{(a)}} interference-free or \emph{\textbf{(b)}} pilot-contaminated. Recall that the SNR of user $k$ is defined as $\snr_k = \beta_k P_k / \sigma^2$, where $\beta_k>0$ denotes the pathloss. As an example, let us consider that the cell size increases with the number of users and that the users are uniformly distributed over the cell. Then, under a suitable ordering of the users according to the pathloss, there exists some constants $C_{\beta}>0$ and $e_{\beta}>1$ such that $\beta_k \leq C_{\beta}k^{-e_{\beta}}$, and the conditions in Corollary \ref{cor:deteq4}.\emph{\textbf{(a)}} are satisfied. Hence, we are in the noise-limited scenario and the deterministic equivalent in \eqref{eq:xi_0_deteq4} takes the form of the MSE expression in Lemma \ref{lem:mse_ii_orthogonal} for the interference-free case.  To the contrary, when the cell size is fixed and the users are uniformly distributed in the cell, there exists some constant $\bar{\beta}>0$ such that $\beta_k > \bar{\beta}$ for $k=1, \ldots, K$. Thus, the conditions in Corollary \ref{cor:deteq4}.\emph{\textbf{(b)}} are fulfilled and the LMMSE estimator does not completely remove interference. In that case, the estimation SINRs gets reduced by a factor of  $(1-\tfrac{\bar{\tau}_M}{\alpha_L})^{-1}$ due to the effect of pilot contamination.

\subsection{Covariance-Aided CSI Acquisition for Training Overhead Reduction} \label{sec:training_overhead}

\noindent Let us now illustrate how the previous large-system analysis can be used to quantify the benefits of exploiting the knowledge of the user covariance matrices during CSI acquisition in order to reduce the training overhead beyond the extreme cases of mutually orthogonal pilots and/or covariance matrices. This can be more formally stated as follows.

\smallskip
\begin{problem*}[Pilot length optimization] \label{pr:pilot_length}
 Given a massive MIMO system with $M$ BS antennas serving $K+1$ users, we seek the minimum pilot length  $L^{\star}$ which guarantees for the covariance-aided CSI acquisition in  
 \eqref{eq:hat_h_k_ii} (case {\rm\textsf{(ii)}}) at least the same average channel estimation performance obtained by the conventional CSI acquisition scheme in \eqref{eq:hat_h_k_i}  with orthogonal pilots (case {\rm\textsf{(i)}}) of length $K+1$, that is, the minimum pilot length  $L^{\star}$ such that
\begin{align} \label{eq:design_condition}
\mse_0^{\mathsf{(ii)}} \big(\P, \{ \Sig_k\}, \{\snr_k\} \big) \leq \mse_0^{\mathsf{(i)}}(\snr_0)
\end{align}
where $\mse_0^{\mathsf{(i)}}(\snr_0)$ and $\mse_0^{\mathsf{(ii)}} \big(\P, \{ \Sig_k\}, \{\snr_k\} \big)$ are defined in \eqref{eq:sum_mse_i} and \eqref{eq:sum_mse_ii}, respectively.
\end{problem*}
 \smallskip

 The previous pilot length optimization problem is difficult to solve based on the MSE expression depending of the exact covariance matrices and pilots given in \eqref{eq:sum_mse_ii}. However, it can be approximated in closed form using the deterministic equivalent in Corollary \ref{cor:deteq4}.{\bf(\emph{b})} as follows. We upper-bound the deterministic equivalent using Jensen's inequality as in \eqref{eq:mse_ii_ub}:
\begin{align}
\bar{\xi}_{0}^{(\mathsf{ii})}\Big(\Sig_0, \snr_0; \frac{\bar{\tau}_M}{\alpha_L-\bar{\tau}_M} \Big) \leq  \frac{\beta_0 }{1 + \frac{1}{\tau_{M,0}} \big(1- \frac{\bar{\tau}_M}{\alpha_L}\big) L \snr_0  }
\end{align}
so that the performance guarantee condition in \eqref{eq:design_condition} can be simply approximated in the large-system limit as
\begin{align}
  \frac{1}{\tau_{M,0}} \big(1- \frac{\bar{\tau}_M}{\alpha_L}\big) L\snr_0  \geq (K+1)\snr_0
\end{align}
which implies that $L^{\star}$ can be approximated by $L^{\mathsf{(ii)}}$ given by
\begin{align} \label{eq:estimated_length}
L^{\mathsf{(ii)}} = \big\lceil (K+1) \tau_{M,0} + K \bar{\tau}_M   \big\rceil .
\end{align}
The length in \eqref{eq:estimated_length} can be interpreted as follows.
The first term,  $(K+1)\tau_{M,0}$, corresponds to the noise reduction obtained from the projection into the $r_0$-dimensional subspace spanned by the covariance matrix of user $0$ performed by the covariance-aided channel estimator in \eqref{eq:hat_h_k_ii}.
The second term, $K \bar{\tau}_M$, accounts for the loss due the interference of the other users as it is explicit in \eqref{eq:sum_mse_ii}.
The result suggests moreover that non-orthogonal pilots are useful only if
\begin{align} \label{eq:rank_cod_length}
  \tau_{M,0}+\bar{\tau}_M<1.
\end{align}
Indeed, if it is not the case, the approximated minimum length $L^{\mathsf{(ii)}}$ is necessarily greater than $K+1$, meaning that the use of orthogonal pilots will result in better channel estimation accuracy even if the individual covariance matrices are exploited. This can be explained by the fact that orthogonal pilots naturally completely remove interference between users while the interference induced by non-orthogonal pilots cannot be not compensated by the covariance matrix knowledge if \eqref{eq:rank_cod_length} does not hold.

To the contrary, whenever \eqref{eq:rank_cod_length} is satisfied, covariance-aided CSI acquisition allows to reduce the pilot length with respect to conventional CSI acquisition with orthogonal pilots of length $K+1$ by a factor $\Delta^{\star} = (K+1)/\Exp\{L^{\star}\}$. 
This ratio can be approximated by $\Delta^{\mathsf{(ii)}} = (K+1)/L^{\mathsf{(ii)}}$ satisfying 
\begin{equation} \label{eq:limitDelta}
\big|\Delta^{\mathsf{(ii)}}-\bar{\Delta}\big| \xrightarrow[K,L,M\rightarrow\infty]{} 0\text{\ \ \ \ \ \ with \ \ \ \ \ \ } \bar{\Delta} =  \frac{1}{\tau_{M,0} +\bar{\tau}_M} .
\end{equation}

The accuracy of the approximation of $L^{\star}$ by $L^{\mathsf{(ii)}}$ and that of $\Delta^{\star}$ by $\Delta^{\mathsf{(ii)}}$  is numerically investigated in Section \ref{sec:numerical}.

\section{Numerical Evaluation} \label{sec:numerical}

\noindent In this section we illustrate numerically the accuracy of the deterministic equivalents presented in Theorems~\ref{thm:deteq1}-\ref{thm:deteq4} and Corollary~\ref{cor:deteq4}. For doing so, we consider the following the random pilots and covariance matrix models. 
\smallskip
\begin{pilot_model}
Given the number of users $K+1$, we generate the pilot sequences $\p_k = \big( p_k(1), \ldots, p_k(L)\big)^{\tras}$ of length $L = \lfloor \alpha_L K\rfloor$ for $k=0, \ldots, K$ as
\begin{align} \label{eq:sim_pilot_model}
p_{k}(\ell) & = e^{j\psi_{k,\ell}}, & \qquad  \ell = 1, \ldots, L
\end{align}
where $\psi_{k,\ell}$ are i.i.d. random variables uniformly distributed in  $[0,2\pi)$ (this model satisfies Assumption~{\rm\ref{ass:pilotmodel}}) and we generate the covariance matrices using one of the following models.
\end{pilot_model}
\smallskip

\begin{cov_model}[Maximum Entropy]
Under the maximum entropy principle, the individual spatial covariance matrices are modeled as 
\begin{align} \label{eq:maxentr_cov_model}
\Sig_{k} & = \frac{1}{ r_k } \X_{k}\X_{k}^\dagger, & \qquad k=0,\ldots, K
\end{align} 
where $r_k = \lfloor \tau_{k,M} M \rfloor$ denotes the rank, the entries of $\X_{k}\in\Compl^{M \times r_k}$ are i.i.d.~complex Gaussian zero-mean, unit variance random variables. This model satisfies Assumption~{\rm\ref{ass:covmodel}}. 
\end{cov_model}

\smallskip

Alternatively, we also consider a one-ring correlation model \cite{Adhikary_JSDM_IT2013},\cite[Sec.~2.6]{BjornsonHoydisSanguinetti}, which assumes that a user located at azimuth angle $\phi$ is surrounded by a cluster of scatterers creating multipath components with angles of arrival uniformly distributed in $[\phi-\sqrt{3}\sigma_{\phi}, \phi+\sqrt{3}\sigma_{\phi}]$, where $\sigma_{\phi}$ denotes the angular spread. In particular we generate the random spatial covariance matrices as follows. 

\smallskip

\begin{cov_model}[One-Ring with UCA]
 Assume that the BS is equipped with uniform circular array (UCA) with $M$ antenna elements with half-wavelength spacing {\rm \cite{UCA2005}}. Under the one-ring model, the individual spatial covariance matrices are obtained as 
  \begin{align} \label{eq:uca_cov_model}
    \big[ \Sig_{k}\big]_{n,m} = \frac{\beta_k}{2\sqrt{3}\sigma_{\phi}}\int_{-\sqrt{3}\sigma_{\phi}}^{\sqrt{3}\sigma_{\phi}} e^{-j\frac{M}{2} 
    \big(\cos(\phi_k+\varphi-\theta_n)-\cos(\phi_k+\varphi-\theta_m)\big)} d\varphi, & \qquad n,m=1,\ldots, M, \qquad k=0,\ldots, K
    \end{align}
  where $\theta_m = (m-1) 2\pi/M$, the azimuth angle $\phi_k$ is uniformly distributed in  $[0,2\pi)$, the angular spread $\sigma_{\phi}$ is set to $10^\circ$.
\end{cov_model}
\smallskip
Observe that this model does not satisfy Assumptions~{\rm\ref{ass:covmodel}} since the ranks $\{r_k\}$ are random and cannot be explicitly controlled. We illustrate the difference by plotting the respective average normalized rank versus $M$ for both models in Figure~\ref{fig:rank_covariances}.

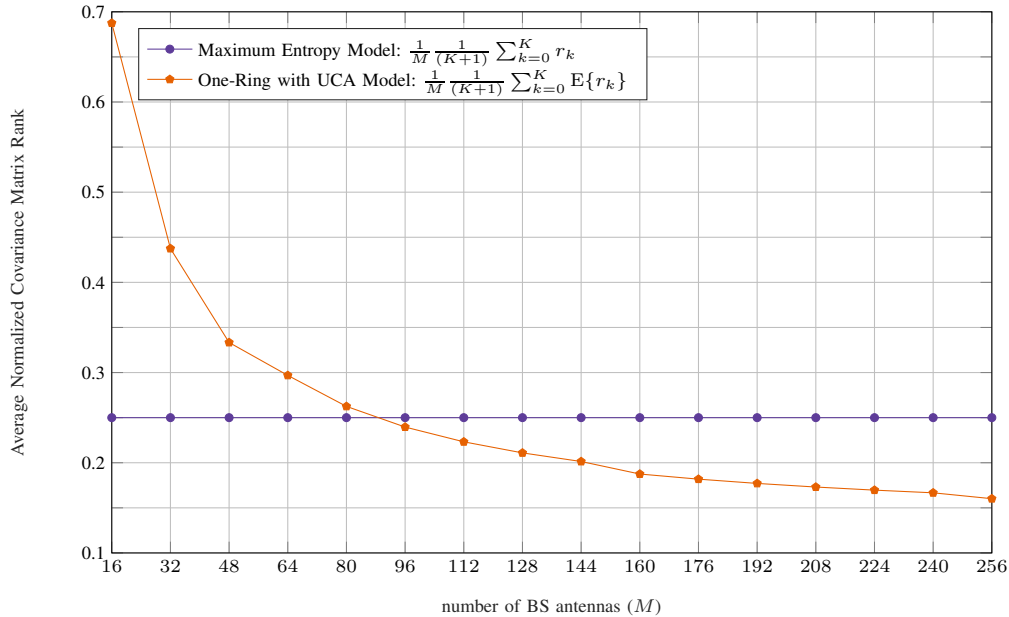
\begin{figure}[t!]
\centering
{\scriptsize 
%
%
\definecolor{mycolor1}{rgb}{0.36860,0.23530,0.60000}%
\definecolor{mycolor2}{rgb}{0.90200,0.38040,0.00390}%
\begin{tikzpicture}

\begin{axis}[%
width=0.65\textwidth,
height=0.4\textwidth,
at={(0\textwidth,0\textwidth)},
scale only axis,
xmin=16,
xmax=256,
xtick={16,32,48,64,80,96,112,128,144,160,176,192,208,224,240,256},
xlabel style={font=\color{white!15!black}},
xlabel={number of BS antennas ($M$)},
ymin=0.1,
ymax=0.7,
ytick={0.1,0.15,0.2,0.25,0.3,0.35,0.4,0.45,0.5,0.55,0.6,0.65,0.7},
yticklabels={{0.1},{},{0.2},{},{0.3},{},{0.4},{},{0.5},{},{0.6},{},{0.7}},
ylabel style={font=\color{white!15!black}},
ylabel={Average Normalized Covariance Matrix Rank},
axis background/.style={fill=white},
xmajorgrids,
ymajorgrids,
legend style={at={(0.03,0.97)}, anchor=north west, legend cell align=left, align=left, draw=white!15!black}
]
\addplot [color=mycolor1, mark size=1.5pt, mark=*, mark options={solid, fill=mycolor1, mycolor1}]
  table[row sep=crcr]{%
16	0.25\\
32	0.25\\
48	0.25\\
64	0.25\\
80	0.25\\
96	0.25\\
112	0.25\\
128	0.25\\
144	0.25\\
160	0.25\\
176	0.25\\
192	0.25\\
208	0.25\\
224	0.25\\
240	0.25\\
256	0.25\\
};
\addlegendentry{Maximum Entropy Model: $\frac{1}{M}\frac{1}{(K+1)}\sum_{k=0}^K r_k$ }

\addplot [color=mycolor2, mark size=1.6pt, mark=pentagon*, mark options={solid, fill=mycolor2, mycolor2}]
  table[row sep=crcr]{%
16	0.6875\\
32	0.4375\\
48	0.333333333333333\\
64	0.296875\\
80	0.2625\\
96	0.239583333333333\\
112	0.223214285714286\\
128	0.2109375\\
144	0.201388888888889\\
160	0.1875\\
176	0.181818181818182\\
192	0.177083333333333\\
208	0.173076923076923\\
224	0.169642857142857\\
240	0.166666666666667\\
256	0.16015625\\
};
\addlegendentry{One-Ring with UCA  Model: $\frac{1}{M} \frac{1}{(K+1)}\sum_{k=0}^K \Exp\{ r_k \}$ }

\addplot [color=black, forget plot]
  table[row sep=crcr]{%
32	1\\
48	1\\
64	1\\
80	1\\
96	1\\
112	1\\
128	1\\
144	1\\
160	1\\
176	1\\
192	1\\
208	1\\
224	1\\
240	1\\
256	1\\
};
\end{axis}
\end{tikzpicture}%
} \\
\caption{ \small Averaged normalized rank for the maximum entropy covariance model with $r_k = \lfloor M/4 \rfloor$ and the one-ring covariance models with $10^\circ$ angular spread.}
\label{fig:rank_covariances}
\vspace*{-5mm}
\end{figure}

First, in Figure \ref{fig:mse_max_entropy} we plot the exact MSE for the covariance-aided CSI acquisition strategy (computed using \eqref{eq:sum_mse_ii}) and the approximations obtained from the deterministic equivalents in Theorems~\ref{thm:deteq1},~\ref{thm:deteq3},~and~\ref{thm:deteq4} averaged over 100 realizations of the random pilots and covariance matrices. For each value of the BS antenna number $M$, we set the number of users to $K+1$, with $K=\lfloor M/4 \rfloor-1$, we generate the individual covariance matrices according to the maximum entropy model in \eqref{eq:maxentr_cov_model} with $r_k = \lfloor M/4 \rfloor$ and pathloss $\beta_k=1$, and we allocate pilots of length $L = \lfloor K/4 \rfloor$ following the model in \eqref{eq:sim_pilot_model}.  All the obtained deterministic equivalents provide a very good accuracy  in approximating  the actual MSE and this accuracy increases with the number of antennas as expected. This observation is further confirmed in Figure \ref{fig:mse_error_max_entropy}, where we plot the normalized approximation error incurred by the different deterministic equivalents in order to measure the convergence of the approximation in the large-system limit and defined as 
\begin{equation}
L\Exp\big\{\big|\mse_0^{(\mathsf{ii})} (\P,\{\Sig_k\},\{\snr_k\})-\xi_{0}^{(\mathsf{ii})}(\cdot)\big|\big\}
\end{equation}
where $\xi_{0}^{(\mathsf{ii})}(\cdot)$ denotes the considered deterministic equivalent.

In this simulation setup, the effect of using non-orthogonal pilots is more important in the CSI estimation MSE than the relative orthogonality of the user covariance matrices. Indeed, we observe a higher accuracy achieved by the deterministic equivalent from Theorem~\ref{thm:deteq1}, where the pilots are assumed to be deterministic and, hence, the actual pilot set is used to compute $\xi_{0}^{(\mathsf{ii})}\big(\P, \Sig_0, \{\snr_k\} \big)$. Furthermore, when the pilots are assumed to be random as done in Theorems~\ref{thm:deteq3}~and~\ref{thm:deteq4}, there is no appreciable accuracy improvement from using the exact covariance matrices (as in Theorem~\ref{thm:deteq3}) with respect to modeling them as random (as in Theorem~\ref{thm:deteq4}).  In both figures, we have omitted the deterministic equivalent in Corollary~\ref{cor:deteq4}, since under case {\bf(\emph{b})}  $\gamma_\infty = \frac{\bar{\tau}_M}{\alpha_L-\bar{\tau}_M}$, it provides the same result as using the fixed point equation for $\gamma_L$ in Proposition~\ref{prop:Gam_L2}. 

Let us now consider the one-ring correlation model with a BS equipped with a UCA as defined in Covariance Matrix Model~2 with unit pathloss for all users. For each covariance matrix, we compute its rank by using the Matlab function, defining then the rank as the number of eigenvalues whose ratio with the strongest eigenvalue is above the numerical precision (see Figure \ref{fig:rank_covariances}). We repeat the previous simulations and plot the results in Figure~\ref{fig:channel_perf_uca_10}. Conversely to the maximum entropy model, the one-ring UCA covariance matrix model does not satisfy Assumption \ref{ass:covmodel} and, hence, Theorems~\ref{thm:deteq1}~and~\ref{thm:deteq4} do not hold in this case. Still, we can see in Figure~\ref{fig:channel_perf_uca_10} that the corresponding deterministic equivalents provide a fairly good approximation of the actual MSE but, as expected, Theorem~\ref{thm:deteq3} results in a higher accuracy, since it explicitly takes into account all individual covariance matrices. 

Finally, in order to quantify the benefits of exploiting the knowledge of the user covariance matrices during CSI acquisition, we focus on the CSI training overhead reduction problem stated in Section \ref{sec:training_overhead}. For each value of the BS antenna number $M$, we set the number of users to $K+1$, with $K=\lfloor M/4 \rfloor-1$, and we compute (using exhaustive search) the minimum pilot length $L^{\star}$ required by covariance-aided CSI acquisition using random non-orthogonal pilots following the model in \eqref{eq:sim_pilot_model}, which guarantees the same MSE as the conventional CSI acquisition strategy with orthogonal pilots of length $L^{(\mathsf{i})}=K+1$. In Figures \ref{fig:pilot_perf_max_entropy}.(a) and \ref{fig:pilot_perf_uca_10}.(a) we compare $L^{\star}$ with the approximated minimum length $L^{\mathsf{(ii)}}=\big\lceil (K+1) \tau_{M,0} + K \bar{\tau}_M \big\rceil$  for both the maximum entropy model with $r_k = \lfloor M/4 \rfloor$ and for the one-ring UCA model with an angular spread of $10^\circ$, respectively, averaged over 100 random realizations of the pilots and the covariance matrices. In both cases $L^{\mathsf{(ii)}}$ gives a very accurate approximation of $\Exp\big\{ L^{\star}\big\}$ which improves with increasing number of BS antennas. This is further confirmed in Figures \ref{fig:pilot_perf_max_entropy}.(b) and \ref{fig:pilot_perf_uca_10}.(b), where we plot the average pilot reduction with respect to orthogonal pilots, defined as $\Delta =  L^{\mathsf{(i)}}/\Exp\big\{ L \big\} = (K+1)/\Exp\big\{ L \big\}$, where $L$ either denotes $ L^{\star}$ or $L^{\mathsf{(ii)}}$. In the case of maximum entropy covariance matrices, for which we can explicitly control the rank so that $\tau_{M,k} = 1/4$, we also include the large-system pilot-length reduction in \eqref{eq:limitDelta},  $\bar{\Delta} = \frac{1}{\tau_{M,0} +\bar{\tau}_M}= 2$. This confirms the benefits of using covariance-aided CSI acquisition for significantly reducing the training overhead in massive MIMO systems.

\section{Conclusions}

\noindent We have applied a large-system analysis to characterize the performance of covariance-aided multi-user CSI estimation in the uplink of a massive MIMO system. Deterministic equivalents of the achieved estimation MSE were obtained under several assumptions related to the stochastic nature of the spatial covariance matrices and/or the pilot sequences.
When the covariance matrices and the pilots sequences are assumed to be drawn from some i.i.d. random distributions, our results indicate that the performance of covariance-aided CSI acquisition can be interpreted as that of a system using orthogonal pilot sequences of certain equivalent pilot length, for which a closed-form expression enables an intuitive interpretation of the achieved MSE.  Numerical results demonstrate that the covariance-based strategy allows to significantly reduce the training overhead with respect to conventional CSI acquisition. Finally, we contributed to random matrix analysis by extending the trace-lemma from \cite{Bai98} to block matrices.

\clearpage
\begin{figure}
\centering
\subfloat[MSE and deterministic equivalents]{{\scriptsize 
%
%
\definecolor{mycolor1}{rgb}{0.36860,0.23530,0.60000}%
\definecolor{mycolor2}{rgb}{0.90200,0.38040,0.00390}%
\definecolor{mycolor3}{rgb}{0.99220,0.72160,0.38820}%
\begin{tikzpicture}

\begin{axis}[%
width=0.65\textwidth,
height=0.466\textwidth,
at={(0\textwidth,0\textwidth)},
scale only axis,
xmin=16,
xmax=256,
xtick={16,32,48,64,80,96,112,128,144,160,176,192,208,224,240,256},
xticklabels={{16},{32},{48},{64},{80},{96},{112},{128},{144},{160},{176},{192},{208},{224},{240},256},
xlabel style={font=\color{white!15!black}},
xlabel={number of BS antennas ($M$)},
ymode=log,
ymin=0.0001,
ymax=0.01,
yminorticks=true,
ylabel style={font=\color{white!15!black}},
ylabel={MSE},
axis background/.style={fill=white},
xmajorgrids,
ymajorgrids,
yminorgrids,
legend style={legend cell align=left, align=left, draw=white!15!black}
]
\addplot [color=white, line width=0.0pt, mark size=1.5pt, mark=*, mark options={solid, fill=white, black}]
  table[row sep=crcr]{%
16	0.00769744535012391\\
32	0.00257752567518054\\
48	0.00156204247238608\\
64	0.00111923651256605\\
80	0.000870906779501302\\
96	0.000713587645575558\\
112	0.000604384276103901\\
128	0.000524080529821591\\
144	0.000462960542303109\\
160	0.000414701212262275\\
176	0.000374852144376643\\
192	0.000342696232168508\\
208	0.000315317873769934\\
224	0.000291547204399278\\
240	0.00027178559633031\\
256	0.000254188654859204\\
};
\addlegendentry{$\Exp\big\{\mse_0^{(\mathsf{ii})} (\P,\{\Sig_k\}, \{\snr_k\})\big\}$}

\addplot [color=mycolor1]
  table[row sep=crcr]{%
16	0.00654465792147914\\
32	0.00246483421177591\\
48	0.00152050958433111\\
64	0.00109777237891869\\
80	0.000858935186213502\\
96	0.000705623125399121\\
112	0.00059887604470694\\
128	0.000519877309386732\\
144	0.000459752019288268\\
160	0.000412084302412833\\
176	0.000372743353802665\\
192	0.000340937125650642\\
208	0.000313856889988744\\
224	0.000290295006544456\\
240	0.000270697259716612\\
256	0.000253242003131544\\
};
\addlegendentry{Thm.~1: $\Exp\big\{\xi_{0}^{(\mathsf{ii})}(\P, \Sig_0, \{\snr_k\})\big\}$}

\addplot [color=mycolor2]
  table[row sep=crcr]{%
16	0.00655402110405728\\
32	0.00245582877498416\\
48	0.00151771547127928\\
64	0.00109725914634971\\
80	0.000858137404808993\\
96	0.000705056487983958\\
112	0.00059819425628976\\
128	0.000519636546021396\\
144	0.0004591959109514\\
160	0.000411456418464184\\
176	0.000372640726538432\\
192	0.000340537233027674\\
208	0.000313502929615228\\
224	0.000290481454200408\\
240	0.000270592497171297\\
256	0.000253248911611012\\
};
\addlegendentry{Thm.~2: $\Exp\big\{\xi_{0}^{(\mathsf{ii})}(\{\Sig_k\}, \{\snr_k\}; \Gam_L)\big\}$}

\addplot [color=mycolor3]
  table[row sep=crcr]{%
16	0.00624237311467101\\
32	0.00242043095979283\\
48	0.00150127208381352\\
64	0.00108804907211711\\
80	0.000853198377192894\\
96	0.000701736092800897\\
112	0.000595940883637248\\
128	0.000517867265329213\\
144	0.00045787947671559\\
160	0.000410346881229778\\
176	0.000371755042238937\\
192	0.000339798035444458\\
208	0.000312900078800639\\
224	0.000289948168833928\\
240	0.0002701335761087\\
256	0.000252853723651495\\
};
\addlegendentry{Thm.~4: $\Exp\big\{\bar{\xi}_{0}^{(\mathsf{ii})}(\Sig_0, \{\snr_k\} ; \gamma_L )\big\}$}

\end{axis}
\end{tikzpicture}%
}\label{fig:mse_max_entropy}}\\
\vspace*{5mm}
\subfloat[Normalized approximation error]{{\scriptsize 
%
%
\definecolor{mycolor1}{rgb}{0.36860,0.23530,0.60000}%
\definecolor{mycolor2}{rgb}{0.90200,0.38040,0.00390}%
\definecolor{mycolor3}{rgb}{0.99220,0.72160,0.38820}%
\begin{tikzpicture}

\begin{axis}[%
width=0.65\textwidth,
height=0.466\textwidth,
at={(0\textwidth,0\textwidth)},
scale only axis,
xmin=16,
xmax=256,
xtick={16,32,48,64,80,96,112,128,144,160,176,192,208,224,240,256},
xticklabels={{16},{32},{48},{64},{80},{96},{112},{128},{144},{160},{176},{192},{208},{224},{240},{256}},
xlabel style={font=\color{white!15!black}},
xlabel={number of BS antennas ($M$)},
ymode=log,
ymin=4e-05,
ymax=0.003,
yminorticks=true,
ylabel style={font=\color{white!15!black}},
ylabel={Normalized Approximation Error},
axis background/.style={fill=white},
xmajorgrids,
ymajorgrids,
yminorgrids,
legend style={legend cell align=left, align=left, draw=white!15!black}
]
\addplot [color=mycolor1]
  table[row sep=crcr]{%
16	0.0023222195363525\\
32	0.000563457317023135\\
48	0.000332263104439758\\
64	0.000236105470120908\\
80	0.000167602306029202\\
96	0.000135396842999432\\
112	0.000110164627939219\\
128	9.66740700017615e-05\\
144	8.34215983858572e-05\\
160	7.58903856338039e-05\\
176	6.7481298367322e-05\\
192	6.15687281253289e-05\\
208	5.55173836851964e-05\\
224	5.13401120476841e-05\\
240	4.78868110027286e-05\\
256	4.44926311999874e-05\\
};
\addlegendentry{Thm.~1}

\addplot [color=mycolor2]
  table[row sep=crcr]{%
16	0.00228684849213328\\
32	0.000609482059863556\\
48	0.00035511604413078\\
64	0.000243154354789002\\
80	0.000181952979585238\\
96	0.000147648188340335\\
112	0.000129312190045205\\
128	0.00010905910185601\\
144	0.000104998158055729\\
160	9.93538638447722e-05\\
176	8.01243022214223e-05\\
192	8.27698626483345e-05\\
208	7.96132153669635e-05\\
224	5.91348071212616e-05\\
240	5.99008540487078e-05\\
256	5.38933074260936e-05\\
};
\addlegendentry{Thm.~2}

\addplot [color=mycolor3]
  table[row sep=crcr]{%
16	0.00291047522568832\\
32	0.000785473576938545\\
48	0.00048616310858048\\
64	0.000343637680207253\\
80	0.000248202747305933\\
96	0.000201694169886084\\
112	0.000169882648811182\\
128	0.000146576187701347\\
144	0.000134433110722944\\
160	0.00012784103559717\\
176	0.000103303319854274\\
192	0.000104848460418519\\
208	9.84964062318456e-05\\
224	7.32035417048732e-05\\
240	7.64623495131695e-05\\
256	6.7077887217259e-05\\
};
\addlegendentry{Thm.~4}

\end{axis}
\end{tikzpicture}%
}\label{fig:mse_error_max_entropy}}
\caption{ \small CSI MSE metrics as a function of $M$ with $K=\lfloor M/4 \rfloor -1$, $\alpha_L=3/4$, $\beta_k =1$, and $\snr_k = 15$ dB averaged over 100 realizations of the random pilots and the covariance matrices following the maximum entropy model with $r_k = M/4$.}
\label{fig:channel_perf_max_entropy}
\end{figure}
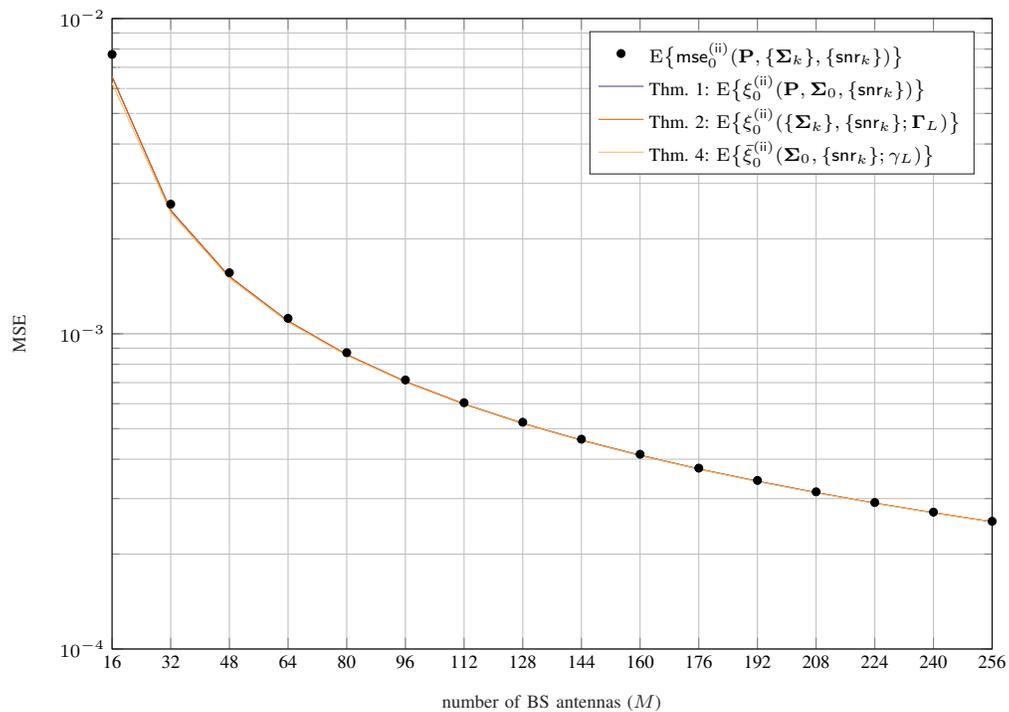
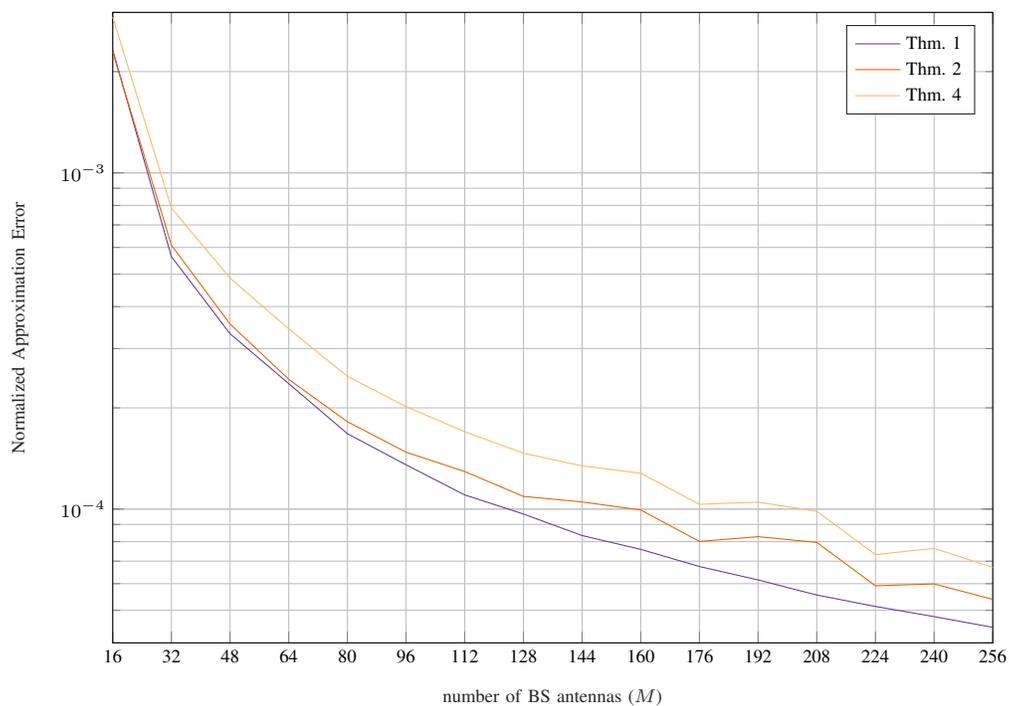

%
\clearpage

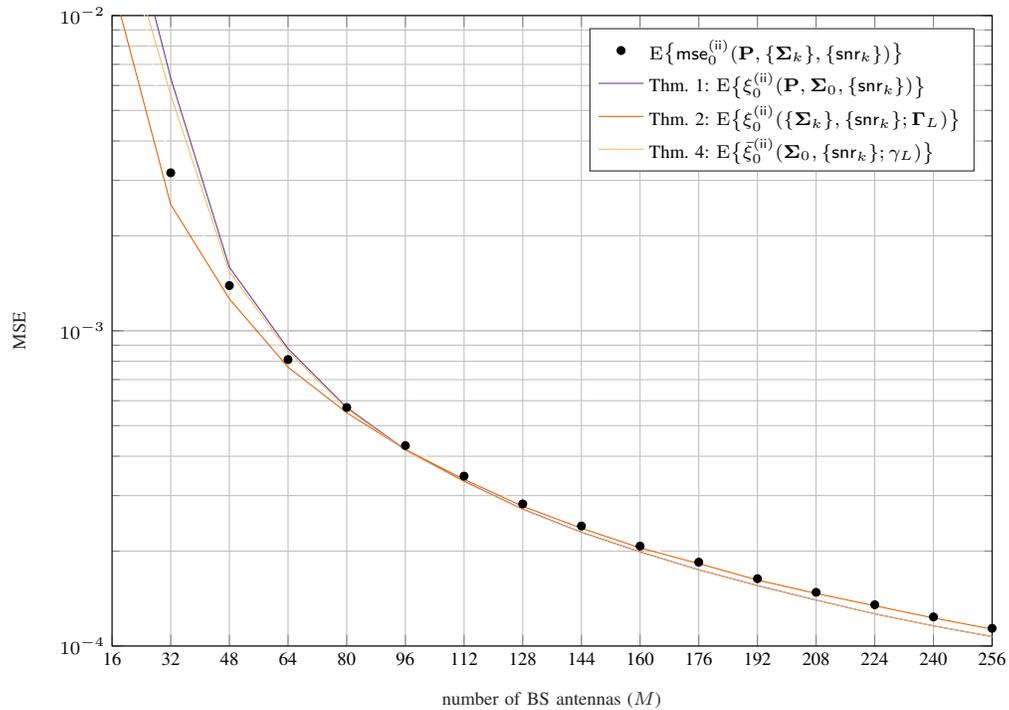
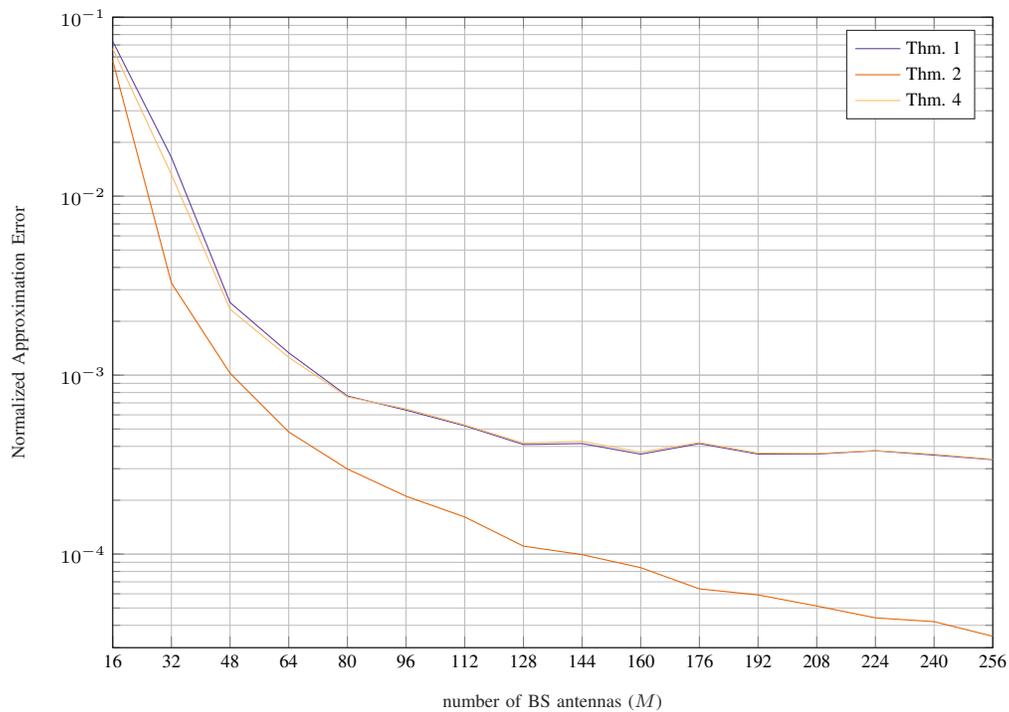
\begin{figure}
\centering
\subfloat[MSE and deterministic equivalents]{{\scriptsize 
%
%
\definecolor{mycolor1}{rgb}{0.36860,0.23530,0.60000}%
\definecolor{mycolor2}{rgb}{0.90200,0.38040,0.00390}%
\definecolor{mycolor3}{rgb}{0.99220,0.72160,0.38820}%
\begin{tikzpicture}

\begin{axis}[%
width=0.65\textwidth,
height=0.466\textwidth,
at={(0\textwidth,0\textwidth)},
scale only axis,
xmin=16,
xmax=256,
xtick={16,32,48,64,80,96,112,128,144,160,176,192,208,224,240,256},
xticklabels={{16},{32},{48},{64},{80},{96},{112},{128},{144},{160},{176},{192},{208},{224},{240},{256}},
xlabel style={font=\color{white!15!black}},
xlabel={number of BS antennas ($M$)},
ymode=log,
ymin=0.0001,
ymax=0.01,
yminorticks=true,
ylabel style={font=\color{white!15!black}},
ylabel={MSE},
axis background/.style={fill=white},
xmajorgrids,
ymajorgrids,
yminorgrids,
legend style={legend cell align=left, align=left, draw=white!15!black}
]
\addplot [color=white, line width=0.0pt, mark size=1.5pt, mark=*, mark options={solid, fill=white, black}]
  table[row sep=crcr]{%
16	0.0413303105468588\\
32	0.00317134363407415\\
48	0.00139106359224089\\
64	0.000810385638182852\\
80	0.000570961605447893\\
96	0.000432210579371609\\
112	0.000345904892492119\\
128	0.000282246904569246\\
144	0.000240229621224276\\
160	0.000207528026099835\\
176	0.000184492567154699\\
192	0.000163526977587906\\
208	0.000148020936383583\\
224	0.000135160104223514\\
240	0.000123651854519914\\
256	0.000113770511041507\\
};
\addlegendentry{$\Exp\big\{\mse_0^{(\mathsf{ii})} (\P,\{\Sig_k\}, \{\snr_k\})\big\}$}

\addplot [color=mycolor1]
  table[row sep=crcr]{%
16	0.0337526150895571\\
32	0.0063370243368715\\
48	0.00159134682171455\\
64	0.000877745627617655\\
80	0.000570926498587459\\
96	0.000419528692086751\\
112	0.00033324767785002\\
128	0.000271646164475984\\
144	0.000229517471754642\\
160	0.000198853864781833\\
176	0.000174472056456381\\
192	0.000155487948540779\\
208	0.0001398247958145\\
224	0.000126451399436909\\
240	0.000115864179267092\\
256	0.000107089745774163\\
};
\addlegendentry{Thm.~1: $\Exp\big\{\xi_{0}^{(\mathsf{ii})}(\P, \Sig_0, \{\snr_k\})\big\}$}

\addplot [color=mycolor2]
  table[row sep=crcr]{%
16	0.0127622408407351\\
32	0.00251597898156607\\
48	0.00126296935556774\\
64	0.000766632959359\\
80	0.000549613917955604\\
96	0.000419820057759122\\
112	0.000337850031806992\\
128	0.000277428899198298\\
144	0.000236467480915959\\
160	0.000204769995179059\\
176	0.000182577955126303\\
192	0.000161936217877098\\
208	0.000146772925623304\\
224	0.00013417771263995\\
240	0.000122735158134497\\
256	0.000113100832796456\\
};
\addlegendentry{Thm.~2: $\Exp\big\{\xi_{0}^{(\mathsf{ii})}(\{\Sig_k\}, \{\snr_k\}; \Gam_L)\big\}$}

\addplot [color=mycolor3]
  table[row sep=crcr]{%
16	0.0245181184484801\\
32	0.00559658324277874\\
48	0.00153657318838369\\
64	0.000862895353016997\\
80	0.000565990296480626\\
96	0.000417229912422374\\
112	0.000331832847285346\\
128	0.000270897262784834\\
144	0.000228900683586774\\
160	0.000198317522153423\\
176	0.000174231824570598\\
192	0.000155192371456666\\
208	0.000139605320633664\\
224	0.00012641323439324\\
240	0.000115765506401658\\
256	0.000107033996220265\\
};
\addlegendentry{Thm.~4: $\Exp\big\{\bar{\xi}_{0}^{(\mathsf{ii})}(\Sig_0, \{\snr_k\} ; \gamma_L )\big\}$}

\end{axis}
\end{tikzpicture}%
}\label{fig:mse_uca_10}}\\
\vspace*{5mm}
\subfloat[Normalized approximation error]{{\scriptsize 
%
%
\definecolor{mycolor1}{rgb}{0.36860,0.23530,0.60000}%
\definecolor{mycolor2}{rgb}{0.90200,0.38040,0.00390}%
\definecolor{mycolor3}{rgb}{0.99220,0.72160,0.38820}%
\begin{tikzpicture}

\begin{axis}[%
width=0.65\textwidth,
height=0.466\textwidth,
at={(0\textwidth,0\textwidth)},
scale only axis,
xmin=16,
xmax=256,
xtick={16,32,48,64,80,96,112,128,144,160,176,192,208,224,240,256},
xticklabels={{16},{32},{48},{64},{80},{96},{112},{128},{144},{160},{176},{192},{208},{224},{240},{256}},
xlabel style={font=\color{white!15!black}},
xlabel={number of BS antennas ($M$)},
ymode=log,
ymin=3e-05,
ymax=0.1,
yminorticks=true,
ylabel style={font=\color{white!15!black}},
ylabel={Normalized Approximation Error},
axis background/.style={fill=white},
xmajorgrids,
ymajorgrids,
yminorgrids,
legend style={legend cell align=left, align=left, draw=white!15!black}
]
\addplot [color=mycolor1]
  table[row sep=crcr]{%
16	0.0733905384762734\\
32	0.0164883395546703\\
48	0.00254517974581039\\
64	0.00133325278987756\\
80	0.000764797522250228\\
96	0.000637322511681636\\
112	0.000521140473023835\\
128	0.000409694867973287\\
144	0.000414606960406931\\
160	0.000361950797104811\\
176	0.00041433321479274\\
192	0.000362158108239217\\
208	0.000362259108737908\\
224	0.000377961133315567\\
240	0.000357678896645036\\
256	0.00033722584984074\\
};
\addlegendentry{Thm.~1}

\addplot [color=mycolor2]
  table[row sep=crcr]{%
16	0.0571361394122474\\
32	0.00327682326254037\\
48	0.00102475389338521\\
64	0.000481555533338129\\
80	0.000298867624892042\\
96	0.000210638867412277\\
112	0.000161264166546221\\
128	0.000110814123531808\\
144	9.93055580674058e-05\\
160	8.39924171249765e-05\\
176	6.39335457404488e-05\\
192	5.91543601732952e-05\\
208	5.12776395209026e-05\\
224	4.40397406894047e-05\\
240	4.19402203838007e-05\\
256	3.47432406273148e-05\\
};
\addlegendentry{Thm.~2}

\addplot [color=mycolor3]
  table[row sep=crcr]{%
16	0.0668209140583361\\
32	0.0132200492835451\\
48	0.00233533087159379\\
64	0.00125849061905865\\
80	0.000755690939883784\\
96	0.000648267541872237\\
112	0.000527010299918649\\
128	0.000417243519762521\\
144	0.000426846215264344\\
160	0.000370936997274353\\
176	0.000420308560985494\\
192	0.000366935513314976\\
208	0.000365308296922071\\
224	0.000379223533651923\\
240	0.000361587190986576\\
256	0.00033907359515971\\
};
\addlegendentry{Thm.~4}

\end{axis}
\end{tikzpicture}%
}\label{fig:mse_error_uca_10}}
\caption{\small CSI MSE metrics as a function of $M$ with $K=\lfloor M/4 \rfloor -1$, $\alpha_L=3/4$, $\beta_k =1$, and $\snr_k = 15$ dB averaged over 100 realizations of the random pilots and the covariance matrices following the one-ring UCA model with an angular spread of $10^\circ$.}
\label{fig:channel_perf_uca_10}
\end{figure}

%
\clearpage

\begin{figure}
\centering
\subfloat[Minimum pilot length required by the conventional and the covariance-aided CSI acquisition strategies]{{\scriptsize 
%
%
\definecolor{mycolor1}{rgb}{0.36860,0.23530,0.60000}%
\definecolor{mycolor2}{rgb}{0.90200,0.38040,0.00390}%
\begin{tikzpicture}

\begin{axis}[%
width=0.65\textwidth,
height=0.466\textwidth,
at={(0\textwidth,0\textwidth)},
scale only axis,
xmin=16,
xmax=256,
xtick={16,32,48,64,80,96,112,128,144,160,176,192,208,224,240,256},
xlabel style={font=\color{white!15!black}},
xlabel={number of BS antennas},
ymin=1,
ymax=65,
ytick={1,10,20,30,40,50,60,65},
ylabel style={font=\color{white!15!black}},
ylabel={Average Pilot Length},
axis background/.style={fill=white},
xmajorgrids,
ymajorgrids,
legend pos=north west,
legend style={anchor=north west, legend cell align=left, align=left, draw=white!15!black}
]

\addplot [color=white, line width=0.0pt, mark size=1.5pt, mark=*, mark options={solid, fill=mycolor1, mycolor1}]
  table[row sep=crcr]{%
16	2.3425\\
32	4.4096\\
48	6.4363\\
64	8.4282\\
80	10.4417\\
96	12.4462\\
112	14.4247\\
128	16.4364\\
144	18.4184\\
160	20.4574\\
176	22.426\\
192	24.4175\\
208	26.425\\
224	28.4325\\
240	30.4895\\
256	32.396\\
};
\addlegendentry{$\Exp\big\{ L^{\star} \big\}$}

\addplot [color=mycolor2, mark size=1.6pt, mark=pentagon*, mark options={solid, fill=mycolor2, mycolor2}]
  table[row sep=crcr]{%
16	3\\
};
\addlegendentry{$L^{(\mathsf{ii})}$}

\addplot [color=white, mark size=1.3pt, mark=square*, mark options={solid, fill=black, black}, forget plot]
  table[row sep=crcr]{%
16	4\\
32	8\\
48	12\\
64	16\\
80	20\\
96	24\\
112	28\\
128	32\\
144	36\\
160	40\\
176	44\\
192	48\\
208	52\\
224	56\\
240	60\\
256	64\\
};

\addplot [color=black, mark size=1.3pt, mark=square*, mark options={solid, fill=black, black}]
  table[row sep=crcr]{%
16	4\\
};
\addlegendentry{$L^{(\mathsf{i})}$}

\addplot [color=white, mark size=1.6pt, mark=pentagon*, mark options={solid, fill=mycolor2, mycolor2}, forget plot]
  table[row sep=crcr]{%
16	3\\
32	5\\
48	7\\
64	9\\
80	11\\
96	13\\
112	15\\
128	17\\
144	19\\
160	21\\
176	23\\
192	25\\
208	27\\
224	29\\
240	31\\
256	33\\
};

\addplot [color=black, forget plot]
  table[row sep=crcr]{%
16	4\\
32	4\\
};
\addplot [color=black, forget plot]
  table[row sep=crcr]{%
32	8\\
48	8\\
};
\addplot [color=black, forget plot]
  table[row sep=crcr]{%
48	12\\
64	12\\
};
\addplot [color=black, forget plot]
  table[row sep=crcr]{%
64	16\\
80	16\\
};
\addplot [color=black, forget plot]
  table[row sep=crcr]{%
80	20\\
96	20\\
};
\addplot [color=black, forget plot]
  table[row sep=crcr]{%
96	24\\
112	24\\
};
\addplot [color=black, forget plot]
  table[row sep=crcr]{%
112	28\\
128	28\\
};
\addplot [color=black, forget plot]
  table[row sep=crcr]{%
128	32\\
144	32\\
};
\addplot [color=black, forget plot]
  table[row sep=crcr]{%
144	36\\
160	36\\
};
\addplot [color=black, forget plot]
  table[row sep=crcr]{%
160	40\\
176	40\\
};
\addplot [color=black, forget plot]
  table[row sep=crcr]{%
176	44\\
192	44\\
};
\addplot [color=black, forget plot]
  table[row sep=crcr]{%
192	48\\
208	48\\
};
\addplot [color=black, forget plot]
  table[row sep=crcr]{%
208	52\\
224	52\\
};
\addplot [color=black, forget plot]
  table[row sep=crcr]{%
224	56\\
240	56\\
};
\addplot [color=black, forget plot]
  table[row sep=crcr]{%
240	60\\
256	60\\
};
\addplot [color=mycolor2, forget plot]
  table[row sep=crcr]{%
16	3\\
32	3\\
};
\addplot [color=mycolor2, forget plot]
  table[row sep=crcr]{%
32	5\\
48	5\\
};
\addplot [color=mycolor2, forget plot]
  table[row sep=crcr]{%
48	7\\
64	7\\
};
\addplot [color=mycolor2, forget plot]
  table[row sep=crcr]{%
64	9\\
80	9\\
};
\addplot [color=mycolor2, forget plot]
  table[row sep=crcr]{%
80	11\\
96	11\\
};
\addplot [color=mycolor2, forget plot]
  table[row sep=crcr]{%
96	13\\
112	13\\
};
\addplot [color=mycolor2, forget plot]
  table[row sep=crcr]{%
112	15\\
128	15\\
};
\addplot [color=mycolor2, forget plot]
  table[row sep=crcr]{%
128	17\\
144	17\\
};
\addplot [color=mycolor2, forget plot]
  table[row sep=crcr]{%
144	19\\
160	19\\
};
\addplot [color=mycolor2, forget plot]
  table[row sep=crcr]{%
160	21\\
176	21\\
};
\addplot [color=mycolor2, forget plot]
  table[row sep=crcr]{%
176	23\\
192	23\\
};
\addplot [color=mycolor2, forget plot]
  table[row sep=crcr]{%
192	25\\
208	25\\
};
\addplot [color=mycolor2, forget plot]
  table[row sep=crcr]{%
208	27\\
224	27\\
};
\addplot [color=mycolor2, forget plot]
  table[row sep=crcr]{%
224	29\\
240	29\\
};
\addplot [color=mycolor2, forget plot]
  table[row sep=crcr]{%
240	31\\
256	31\\
};

\end{axis}
\end{tikzpicture}%
}\label{fig:pilot_length_max_entropy}}\\
\subfloat[Pilot length reduction with respect to conventional CSI acquisition]{{\scriptsize 
%
%
\definecolor{mycolor1}{rgb}{0.36860,0.23530,0.60000}%
\definecolor{mycolor2}{rgb}{0.90200,0.38040,0.00390}%
\begin{tikzpicture}

\begin{axis}[%
width=0.65\textwidth,
height=0.466\textwidth,
at={(0\textwidth,0\textwidth)},
scale only axis,
xmin=16,
xmax=256,
xtick={16,32,48,64,80,96,112,128,144,160,176,192,208,224,240,256},
xlabel style={font=\color{white!15!black}},
xlabel={number of BS antennas},
ymin=0.9,
ymax=2.1,
ytick={  1, 1.1, 1.2, 1.3, 1.4, 1.5, 1.6, 1.7, 1.8, 1.9,   2},
ylabel style={font=\color{white!15!black}},
ylabel={Average Pilot Length Reduction},
axis background/.style={fill=white},
xmajorgrids,
ymajorgrids,
legend style={at={(0.838,0.13)}, anchor=south west, legend cell align=left, align=left, draw=white!15!black}
]
\addplot [color=mycolor1, mark size=1.5pt, mark=*, mark options={solid, fill=mycolor1, mycolor1}]
  table[row sep=crcr]{%
16	1.70757737459979\\
32	1.81422351233672\\
48	1.86442521324363\\
64	1.89838874255476\\
80	1.91539691812636\\
96	1.92829940062027\\
112	1.94111489320401\\
128	1.94689834757003\\
144	1.95456717195848\\
160	1.95528268499418\\
176	1.96200838312673\\
192	1.96580321490734\\
208	1.9678334910123\\
224	1.96957706849556\\
240	1.96789058528346\\
256	1.97555253735029\\
};
\addlegendentry{$\Delta^{\star}$}

\addplot [color=mycolor2, mark size=1.6pt, mark=pentagon*, mark options={solid, fill=mycolor2, mycolor2}]
  table[row sep=crcr]{%
16	1.33333333333333\\
32	1.6\\
48	1.71428571428571\\
64	1.77777777777778\\
80	1.81818181818182\\
96	1.84615384615385\\
112	1.86666666666667\\
128	1.88235294117647\\
144	1.89473684210526\\
160	1.9047619047619\\
176	1.91304347826087\\
192	1.92\\
208	1.92592592592593\\
224	1.93103448275862\\
240	1.93548387096774\\
256	1.93939393939394\\
};
\addlegendentry{$\Delta^{(\mathsf{ii})}$}

\addplot [color=black, forget plot]
  table[row sep=crcr]{%
16	1\\
24	1\\
32	1\\
40	1\\
48	1\\
56	1\\
64	1\\
72	1\\
80	1\\
88	1\\
96	1\\
104	1\\
112	1\\
120	1\\
128	1\\
136	1\\
144	1\\
152	1\\
160	1\\
168	1\\
176	1\\
184	1\\
192	1\\
200	1\\
208	1\\
216	1\\
224	1\\
232	1\\
240	1\\
248	1\\
256	1\\
};
\addplot [color=black, forget plot]
  table[row sep=crcr]{%
16	2\\
24	2\\
32	2\\
40	2\\
48	2\\
56	2\\
64	2\\
72	2\\
80	2\\
88	2\\
96	2\\
104	2\\
112	2\\
120	2\\
128	2\\
136	2\\
144	2\\
152	2\\
160	2\\
168	2\\
176	2\\
184	2\\
192	2\\
200	2\\
208	2\\
216	2\\
224	2\\
232	2\\
240	2\\
248	2\\
256	2\\
};
\end{axis}

\begin{axis}[%
width=0.707\textwidth,
height=0.53\textwidth,
at={(-0.042\textwidth,-0.053\textwidth)},
scale only axis,
xmin=0,
xmax=1,
ymin=0,
ymax=1,
axis line style={draw=none},
ticks=none,
axis x line*=bottom,
axis y line*=left
]
\node[below right, align=left, draw=none]
at (rel axis cs:0.08,0.95) {covariance-aided CSI acquisition large-system limit, $\bar{\Delta}$};
\node[below right, align=left, draw=none]
at (rel axis cs:0.08,0.218) {conventional CSI acquisition};
\end{axis}
\end{tikzpicture}%
}\label{fig:pilot_reduction_max_entropy}}
\caption{\small Pilot overhead metrics as a function of $M$ with $K=\lfloor M/4 \rfloor -1$, $\beta_k =1$, and $\snr_k = 15$ dB averaged over 100 realizations of the random pilots and the covariance matrices following the maximum entropy model with $r_k = \lfloor M/4 \rfloor$.}
\label{fig:pilot_perf_max_entropy}
\end{figure}
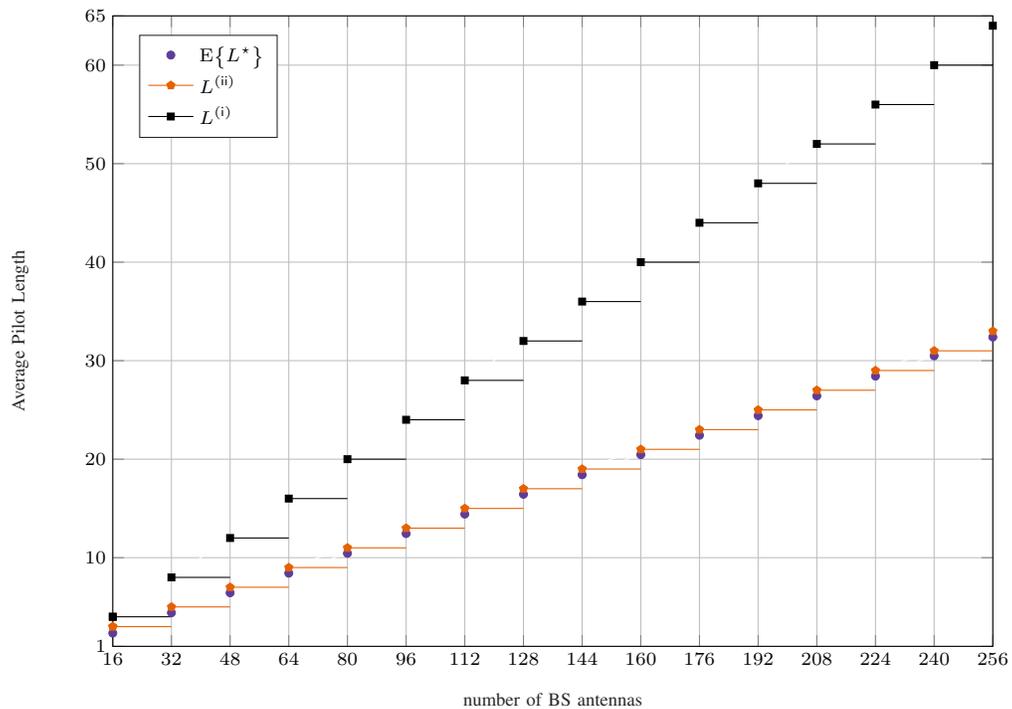
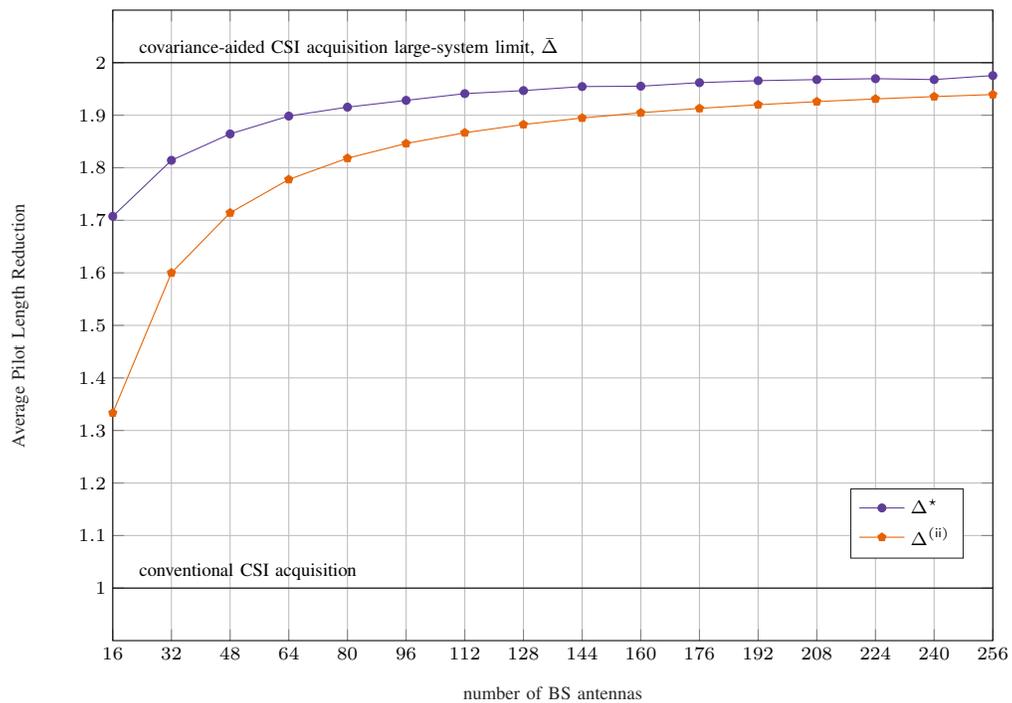

%
\clearpage

\begin{figure}
\centering
\subfloat[Minimum pilot length required by the conventional and the covariance-aided CSI acquisition strategies]{{\scriptsize 
%
%
\definecolor{mycolor1}{rgb}{0.36860,0.23530,0.60000}%
\definecolor{mycolor2}{rgb}{0.90200,0.38040,0.00390}%
\begin{tikzpicture}

\begin{axis}[%
width=0.65\textwidth,
height=0.466\textwidth,
at={(0\textwidth,0\textwidth)},
scale only axis,
xmin=16,
xmax=256,
xtick={16,32,48,64,80,96,112,128,144,160,176,192,208,224,240,256},
xlabel style={font=\color{white!15!black}},
xlabel={number of BS antennas},
ymin=1,
ymax=65,
ytick={1,10,20,30,40,50,60,65},
ylabel style={font=\color{white!15!black}},
ylabel={Average Pilot Length},
axis background/.style={fill=white},
xmajorgrids,
ymajorgrids,
legend pos=north west,
legend style={anchor=north west, legend cell align=left, align=left, draw=white!15!black}
]

\addplot [color=white, line width=0.0pt, mark size=1.5pt, mark=*, mark options={solid, fill=mycolor1, mycolor1}]
  table[row sep=crcr]{%
16	2.7986\\
32	4.6225\\
48	6.2507\\
64	7.4312\\
80	8.86699999999999\\
96	10.1488\\
112	11.4497\\
128	12.5249\\
144	13.7291\\
160	14.8866\\
176	16.2975\\
192	17.197\\
208	18.457\\
224	19.759\\
240	20.8505\\
256	22.0545\\
};
\addlegendentry{$\Exp\big\{ L^{\star} \big\}$}

\addplot [color=mycolor2, mark size=1.6pt, mark=pentagon*, mark options={solid, fill=mycolor2, mycolor2}]
  table[row sep=crcr]{%
16	4\\
};
\addlegendentry{$\Exp\big\{L^{(\mathsf{ii})} \big\}$}

\addplot [color=white, mark size=1.3pt, mark=square*, mark options={solid, fill=black, black}, forget plot]
  table[row sep=crcr]{%
16	4\\
32	8\\
48	12\\
64	16\\
80	20\\
96	24\\
112	28\\
128	32\\
144	36\\
160	40\\
176	44\\
192	48\\
208	52\\
224	56\\
240	60\\
256	64\\
};

\addplot [color=black, mark size=1.3pt, mark=square*, mark options={solid, fill=black, black}]
  table[row sep=crcr]{%
16	4\\
};
\addlegendentry{$L^{(\mathsf{i})}$}

\addplot [color=white, mark size=1.6pt, mark=pentagon*, mark options={solid, fill=mycolor2, mycolor2}, forget plot]
  table[row sep=crcr]{%
16	4\\
32	8\\
48	10.82\\
64	11.84\\
80	12.39\\
96	13.06\\
112	14.24\\
128	15.03\\
144	16.02\\
160	17\\
176	18.01\\
192	18.99\\
208	19.78\\
224	20.4\\
240	21.31\\
256	22.47\\
};

\addplot [color=black, forget plot]
  table[row sep=crcr]{%
16	4\\
32	4\\
};
\addplot [color=black, forget plot]
  table[row sep=crcr]{%
32	8\\
48	8\\
};
\addplot [color=black, forget plot]
  table[row sep=crcr]{%
48	12\\
64	12\\
};
\addplot [color=black, forget plot]
  table[row sep=crcr]{%
64	16\\
80	16\\
};
\addplot [color=black, forget plot]
  table[row sep=crcr]{%
80	20\\
96	20\\
};
\addplot [color=black, forget plot]
  table[row sep=crcr]{%
96	24\\
112	24\\
};
\addplot [color=black, forget plot]
  table[row sep=crcr]{%
112	28\\
128	28\\
};
\addplot [color=black, forget plot]
  table[row sep=crcr]{%
128	32\\
144	32\\
};
\addplot [color=black, forget plot]
  table[row sep=crcr]{%
144	36\\
160	36\\
};
\addplot [color=black, forget plot]
  table[row sep=crcr]{%
160	40\\
176	40\\
};
\addplot [color=black, forget plot]
  table[row sep=crcr]{%
176	44\\
192	44\\
};
\addplot [color=black, forget plot]
  table[row sep=crcr]{%
192	48\\
208	48\\
};
\addplot [color=black, forget plot]
  table[row sep=crcr]{%
208	52\\
224	52\\
};
\addplot [color=black, forget plot]
  table[row sep=crcr]{%
224	56\\
240	56\\
};
\addplot [color=black, forget plot]
  table[row sep=crcr]{%
240	60\\
256	60\\
};
\addplot [color=mycolor2, forget plot]
  table[row sep=crcr]{%
16	4\\
32	4\\
};
\addplot [color=mycolor2, forget plot]
  table[row sep=crcr]{%
32	8\\
48	8\\
};
\addplot [color=mycolor2, forget plot]
  table[row sep=crcr]{%
48	10.82\\
64	10.82\\
};
\addplot [color=mycolor2, forget plot]
  table[row sep=crcr]{%
64	11.84\\
80	11.84\\
};
\addplot [color=mycolor2, forget plot]
  table[row sep=crcr]{%
80	12.39\\
96	12.39\\
};
\addplot [color=mycolor2, forget plot]
  table[row sep=crcr]{%
96	13.06\\
112	13.06\\
};
\addplot [color=mycolor2, forget plot]
  table[row sep=crcr]{%
112	14.24\\
128	14.24\\
};
\addplot [color=mycolor2, forget plot]
  table[row sep=crcr]{%
128	15.03\\
144	15.03\\
};
\addplot [color=mycolor2, forget plot]
  table[row sep=crcr]{%
144	16.02\\
160	16.02\\
};
\addplot [color=mycolor2, forget plot]
  table[row sep=crcr]{%
160	17\\
176	17\\
};
\addplot [color=mycolor2, forget plot]
  table[row sep=crcr]{%
176	18.01\\
192	18.01\\
};
\addplot [color=mycolor2, forget plot]
  table[row sep=crcr]{%
192	18.99\\
208	18.99\\
};
\addplot [color=mycolor2, forget plot]
  table[row sep=crcr]{%
208	19.78\\
224	19.78\\
};
\addplot [color=mycolor2, forget plot]
  table[row sep=crcr]{%
224	20.4\\
240	20.4\\
};
\addplot [color=mycolor2, forget plot]
  table[row sep=crcr]{%
240	21.31\\
256	21.31\\
};

\end{axis}
\end{tikzpicture}%
}\label{fig:pilot_length_uca_10}}\\
\subfloat[Pilot length reduction with respect to conventional CSI acquisition]{{\scriptsize 
%
%
\definecolor{mycolor1}{rgb}{0.36860,0.23530,0.60000}%
\definecolor{mycolor2}{rgb}{0.90200,0.38040,0.00390}%
\begin{tikzpicture}

\begin{axis}[%
width=0.65\textwidth,
height=0.466\textwidth,
at={(0\textwidth,0\textwidth)},
scale only axis,
xmin=16,
xmax=256,
xtick={16,32,48,64,80,96,112,128,144,160,176,192,208,224,240,256},
xlabel style={font=\color{white!15!black}},
xlabel={number of BS antennas},
ymin=0.9,
ymax=3,
ytick={1,1.1,1.2,1.3,1.4,1.5,1.6,1.7,1.8,1.9,2,2.1,2.2,2.3,2.4,2.5,2.6,2.7,2.8,2.9,3},
yticklabels={{1},{},{1.2},{},{1.4},{},{1.6},{},{1.8},{},{2},{},{2.2},{},{2.4},{},{2.6},{},{2.8},{},{3}},
ylabel style={font=\color{white!15!black}},
ylabel={Average Pilot Length Reduction},
axis background/.style={fill=white},
xmajorgrids,
ymajorgrids,
legend style={at={(0.838,0.13)}, anchor=south west, legend cell align=left, align=left, draw=white!15!black}
]
\addplot [color=mycolor1, mark size=1.5pt, mark=*, mark options={solid, fill=mycolor1, mycolor1}]
  table[row sep=crcr]{%
16	1.42928607160723\\
32	1.73066522444565\\
48	1.91978498408178\\
64	2.15308429325008\\
80	2.25555430246983\\
96	2.36481160334227\\
112	2.44547892084509\\
128	2.55491061804885\\
144	2.62216751280127\\
160	2.68698023726035\\
176	2.69980058291149\\
192	2.79118450892598\\
208	2.81735926748659\\
224	2.83415152588694\\
240	2.87762883384092\\
256	2.90190210614614\\
};
\addlegendentry{$\Delta^{\star}$}

\addplot [color=mycolor2, mark size=1.5pt, mark=*, mark options={solid, fill=mycolor2, mycolor2}]
  table[row sep=crcr]{%
16	1\\
32	1\\
48	1.1090573012939\\
64	1.35135135135135\\
80	1.61420500403551\\
96	1.83767228177642\\
112	1.96629213483146\\
128	2.1290751829674\\
144	2.24719101123596\\
160	2.35294117647059\\
176	2.44308717379234\\
192	2.52764612954186\\
208	2.62891809908999\\
224	2.74509803921569\\
240	2.81557954012201\\
256	2.84824210057855\\
};
\addlegendentry{$\Delta^{(\mathsf{ii})}$}

\addplot [color=black, forget plot]
  table[row sep=crcr]{%
32	1\\
48	1\\
64	1\\
80	1\\
96	1\\
112	1\\
128	1\\
144	1\\
160	1\\
176	1\\
192	1\\
208	1\\
224	1\\
240	1\\
256	1\\
};
\end{axis}

\begin{axis}[%
width=0.707\textwidth,
height=0.53\textwidth,
at={(-0.042\textwidth,-0.053\textwidth)},
scale only axis,
xmin=0,
xmax=1,
ymin=0,
ymax=1,
axis line style={draw=none},
ticks=none,
axis x line*=bottom,
axis y line*=left
]
\node[below right, align=left, draw=none]
at (rel axis cs:0.14,0.18) {conventional CSI acquisition};
\end{axis}
\end{tikzpicture}%
}\label{fig:pilot_reduction_uca_10}}
\caption{\small Pilot overhead metrics as a function of $M$ with $K=\lfloor M/4 \rfloor -1$, $\beta_k =1$, and $\snr_k = 15$ dB averaged over 100 realizations of the random pilots and the covariance matrices following the one-ring UCA model with an angular spread of $10^\circ$.}
\label{fig:pilot_perf_uca_10}
\end{figure}
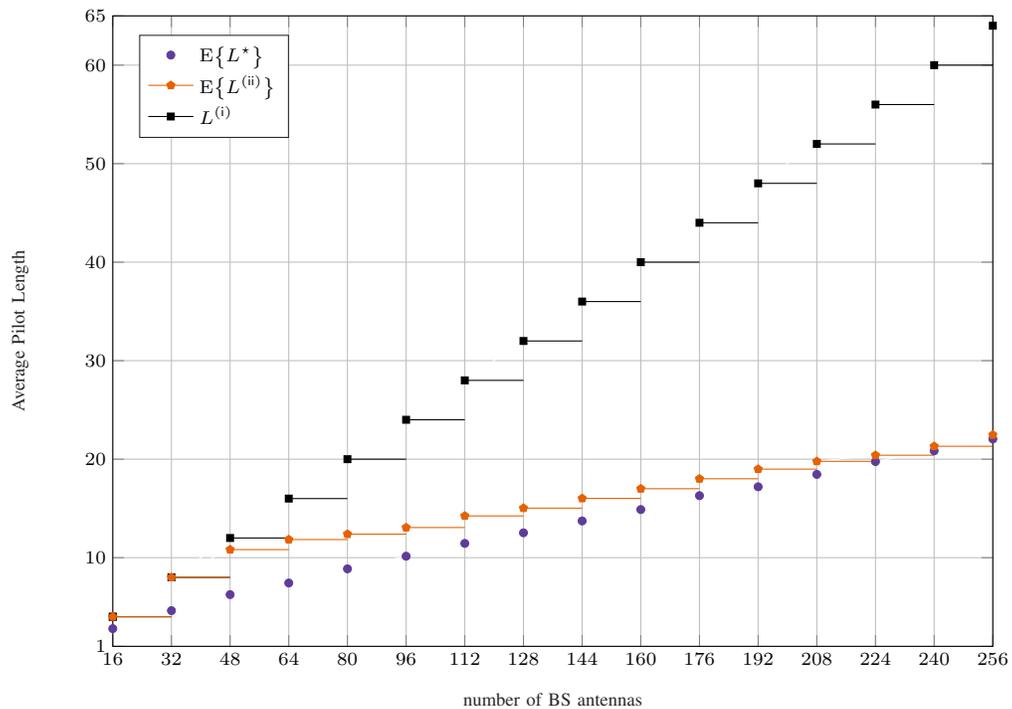
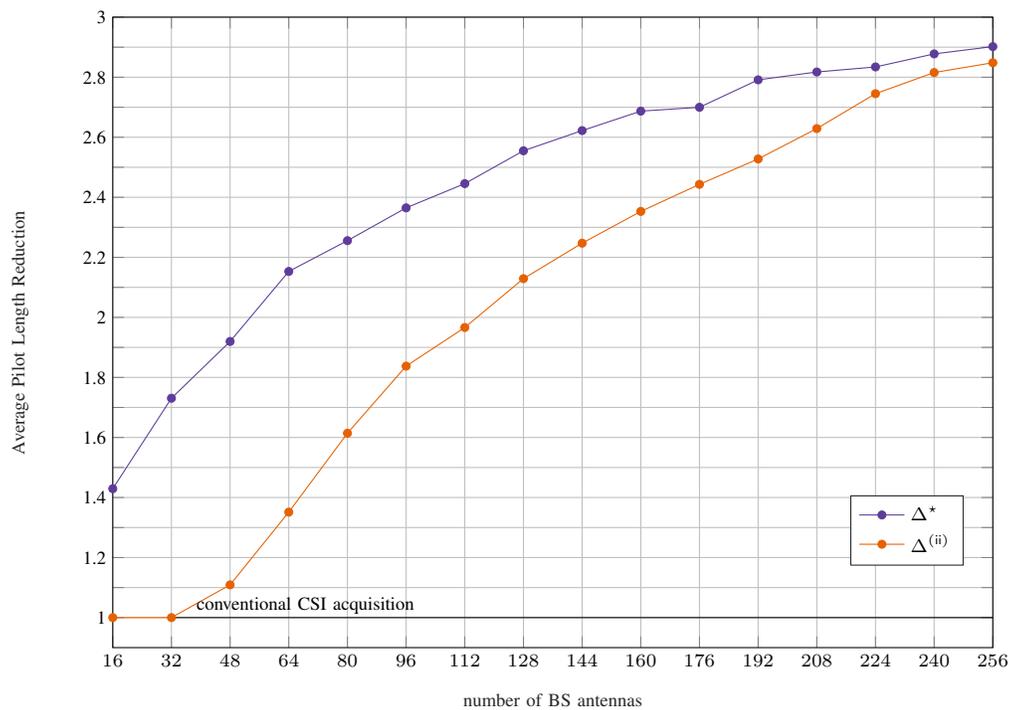

%
\clearpage

\begin{appendices}

\section{Preliminary Results} \label{app:pre_results}

\begin{lemma}[Woodbury identity] \label{lem:inv}
Let $\A, \U, \C, \V$ be respectively $N \times N$, $N\times M$, $M\times M$ and $M\times N$ complex matrices such that $\A, \C, \A+\U\C\V$ are invertible. Then
\begin{equation}
(\A+\U\C\V)^{-1} = \A^{-1} - \A^{-1}\U(\C^{-1}+\V\A^{-1}\U)^{-1}\V\A^{-1}.
\end{equation}
Consider also $\x \in \Compl^{N}$, $c\in\Compl$ for which $\A+c\x\x^{\dagger}$ is invertible. Then,
\begin{align}
\big( \A+c\x\x^{\dagger}\big)^{-1} &= \A^{-1} - \frac{c\A^{-1}\x\x^{\dagger}\A^{-1}}{1+c \x^{\dagger}\A^{-1}\x }\\
\intertext{and}
\x^{\dagger}\big( \A+c\x\x^{\dagger}\big)^{-1} &= \frac{\x^{\dagger}\A^{-1}}{1+c \x^{\dagger}\A^{-1}\x }.
\end{align}
\end{lemma}
\smallskip
\begin{lemma}[{Resolvent identity}]\label{lem:resolvent}
Let $\A$ and $\B$ be two invertible complex matrices of size $N \times N$. Then,
\begin{align}
\A^{-1}-\B^{-1}=-\A^{-1}(\A-\B)\B^{-1}.
\end{align}
\end{lemma}
\smallskip
\begin{lemma} [{\hspace{-.1mm}\cite{Bai98}}] \label{lem:trace}
Let $\A_1,\A_2,\ldots,$ with $\A_N \in \Compl^{N \times N}$ be a series of random matrices with uniformly bounded spectral norm on $N$. Let $\x_1,\x_2,\ldots,$ with $\x_N \in \Compl^N$, be random vectors of i.i.d. entries with zero mean, unit variance, and finite eighth order moment, independent of $A_N$. Then,
\begin{align}
\frac{1}{N}\x_N^{\dagger} \A_N \x_N - \frac{1}{N} \Tr (\A_N) \xrightarrow[N\rightarrow \infty]{\mathrm{a.s.}} 0.
\end{align}
\end{lemma}
\smallskip
\begin{lemma}[{\hspace{-.1mm}\cite[Lem.~2.7]{Bai98}}] \label{lem:proof_trace}
Under the conditions of Lemma  $\ref{lem:trace}$, if for all $N$ and $k$, $\Exp\big\{|x_{M,k}|^m\big\}\leq \nu_m$, then for all $q\geq 1$,
\begin{align}
\Exp \Big\{ \big| \frac{1}{N}\x_N^{\dagger} \A_N \x_N - \frac{1}{N} \Tr (\A_N)\Big|^q\Big\} & \leq 
\frac{C_q}{N^q} \big(\nu_4^{q/2} + \nu_{2q}\big)\Tr\big(\A_N \A_N^{\dagger}\big)^{q/2} 
\end{align}
for some constant $C_q$ depending only on $q$.
\end{lemma}
\smallskip
\begin{lemma}[{\hspace{-.1mm}\cite[Thm.~3.7]{debbahcouillet}}] \label{lem:trace2}
Let $\A_1,\A_2,\ldots,$ with $\A_N \in \Compl^{N \times N}$ be a series of random matrices with uniformly bounded spectral norm on $N$. Let $\x_1,\x_2,\ldots,$ and $\y_1,\y_2,\ldots,$  , $\x_N \in \Compl^N$ and $\y_N \in \Compl^N$, two series of random vectors with i.i.d.~entries such that
 have  zero mean, unit variance, and finite fourth order moment, independent of $\A_N$. Then,
\begin{align}
\frac{1}{N}\x_N^{\dagger} \A_N \y_N   \xrightarrow[N\rightarrow \infty]{\mathrm{a.s.}} 0.
\end{align}
\end{lemma}
\smallskip
\begin{lemma} \label{lem:tr_AB}
Let $\A,\B$ be two $N\times N$ Hermitian matrices such that $\A$ is positive semidefinite. Then,
\begin{align}
|\Tr(\A\B)|\leq \|\B\|\Tr(\A).
\end{align}
\end{lemma}
\begin{proof}
Let  $\lambda_1,\dots,\lambda_N$  denote the eigenvalues of $\B$ and $\v_1,\dots,\v_N$ the corresponding eigenvectors. Then,
\begin{align}
|\Tr(\A\B)|
& = \Big|\Tr\Big(\sum_{i=1}^N \lambda_i\v_i\v_i^\dagger\A\Big)\Big| = \Big|\sum_{i=1}^N \lambda_i\v_i^\dagger\A\v_i\Big| \leq \sum_{i=1}^N |\lambda_i||\v_i^\dagger\A\v_i|
 \leq \|\B\| \sum_{i=1}^N \v_i^\dagger\A\v_i = \|\B\| \Tr(\A)
\end{align}
where the last inequality holds since $\A$ is positive semidefinite.
\end{proof}
\smallskip
\begin{lemma} \label{lem:sum_ABC}
Let $\A_k,\B_k,\C_k$ be three series of matrices of size $N\times M$, $N\times N$ and $N\times M$, respectively. Then, for $N,M\geq 1$
\begin{align}
\Big\|\sum_{k=1}^K \A_k^\dagger\B_k\C_k\Big\|\leq \max_{1\leq k\leq K}\|\B_k\| \Big\|\sum_{k=1}^K \A_k^\dagger\A_k\Big\|^{1/2} \Big\|\sum_{k=1}^K \C_k^\dagger\C_k\Big\|^{1/2}
\end{align}
and
\begin{align}
\Big\|\sum_{k=1}^K \A_k^\dagger\C_k\Big\|_F^2 \leq \sum_{k=1}^K\|\A_k\|_F^2 \Big\|\sum_{k=1}^K \C_k^\dagger\C_k\Big\|.
\end{align}
\end{lemma}
\begin{proof}
Define matrices $\A=(\A_1^T,\A_2^T,\dots,\A_K^T)^T$, $\C=(\C_1^T,\C_2^T,\dots,\C_K^T)^T$ of size $KN \times M$, and the block diagonal matrix $\D_{\B} = \text{diag}(\B_1,\dots,\B_K)$ of size $KN \times KN$. Then, using Cauchy-Schwarz inequality, it holds
\begin{align}
\Big\|\sum_{k=1}^K \A_k^\dagger\B_k\C_k\Big\| & = \|\A^\dagger\D_{\B}\C\|  = \max_{\u,\v\in\mathbb{C}^M,\|\u\|=\|\v\|=1}|\u^\dagger\A^\dagger\D_{\B}\C\v|\\
 &\leq \max_{\u,\v\in\mathbb{C}^M,\|\u\|=\|\v\|=1} \|\D_{\B}\| \|\A\u\|_2 \|\C\v\|_2\\
 &= \max_{1\leq k\leq K}\|\B_k\|  \max_{\u,\in\mathbb{C}^M,\|\u\|=1}\Big|\u^\dagger\Big(\sum_{k=1}^K \A_k^\dagger\A_k\Big)\u\Big|^{1/2} \max_{\v,\in\mathbb{C}^M,\|\v\|=1}\Big|\v^\dagger\Big(\sum_{k=1}^K \C_k^\dagger\C_k\Big)\v\Big|^{1/2} .
\end{align}
Additionally, we prove the second part using Lemma~\ref{lem:tr_AB} 
\begin{align}
\Big\|\sum_{k=1}^K \A_k^\dagger\C_k\Big\|_F^2 & = \Tr\big(\A^\dagger\C\C^\dagger\A\big) \leq \Tr\big(\A\A^\dagger\big) \big\|\C\C^\dagger\big\| = \Tr\big(\A^\dagger\A\big)\big\|\C^\dagger\C\big\|.
\end{align}
\end{proof}
\smallskip

Let us now provide a blockified version generalizing the convergence of the trace lemma (see Lemma \ref{lem:trace}) for block-matrices with a convergence in spectral norm on blocks.
\smallskip
\begin{proposition} \label{prop:block_trace}
Let $\A^{(i,j)}_1,\A^{(i,j)}_2, \ldots,$ with $\A^{(i,j)}_M \in \Compl^{M \times M}$ and  $\A^{(i,j)}_M = \big(\A^{(j,i)}_M\big)^\dagger$, for $i,j=1,\ldots,L$, be a series of matrices. Let $\A_{M,1},\A_{M,2},\ldots,$ with  $\A_{M,L} \in \Compl^{ML\times ML}$ be a series of Hermitian matrices with uniformly bounded spectral norm gathering the blocks $\A^{(i,j)}_M$ as
\begin{align}
[\A_{M,L}]_{i,j} &= \A^{(i,j)}_M, &\qquad  i,j=1,\ldots,L.
\end{align}
Let $\x_1,\x_2,\ldots,$ with $\x_L \in \Compl^L$, be random vectors of i.i.d. entries with zero mean, variance $1$, and finite eighth order moment, independent of $\A_N$, and let $\X_{M,L} = \x_{L} \otimes \I_M\in\mathbb{C}^{ML\times M}$.
Then, considering the block-trace operator defined in \eqref{def:blk_trace}, 
\begin{align}\label{eq:blocktrace1}
\frac{1}{L} \big\|  \X_{M,L}^{\dagger} \A_{M,L} \X_{M,L} - \mathrm{blktr} [\A_{M,L}] \big\| \xrightarrow[L\rightarrow \infty]{\mathrm{a.s.}} 0.
\end{align}

Assume furthermore that the entries of $\x_L = \big(x_1, \ldots, x_L\big)^{\tras}$ are bounded almost surely, i.e., there exists $\chi>0$ such that $|x_i|\leq \chi$ a.s. for $i = 1, \ldots,L$. Then, if $\lim\sup_{L,M} ML^{-\delta}<\infty$ for some $0<\delta<1$, it holds that 
\begin{align}\label{eq:blocktrace2}
\frac{1}{L} \big\|  \X_{M,L}^{\dagger} \A_{M,L} \X_{M,L} - \mathrm{blktr} [\A_{M,L}] \big\| \xrightarrow[L,M\rightarrow \infty]{\mathrm{a.s.}} 0.
\end{align}
\end{proposition}

\begin{proof}
We divide the proof in two parts using different concentration inequalities. First, we prove \emph{\textbf{(i)}} the non-uniform convergence result in \eqref{eq:blocktrace1} and, then, \emph{\textbf{(ii)}} the uniform convergence in $M$ so that \eqref{eq:blocktrace2} holds.

\begin{proof}[Proof of \textbf{(i)}] Similarly to the proof of the trace lemma in \cite{Bai98}, we need to find an integer $q$ such that
\begin{align}
\frac{1}{L^q}\Exp \Big\{ \Big\|  \X_{M,L}^{\dagger} \A_{M,L} \X_{M,L} -  \mathrm{blktr} [\A_{M,L}] \Big\|^q \Big\} \leq f_{M,L}
\end{align}
where $f_{M,L}$ is a constant independent of $\X_{M,L}$, such that $\sum_{L} f_{M,L} < \infty$. Let us first introduce the $M\times M$ matrix $\De_M=\frac{1}{L} \big(\X_{M,L}^{\dagger} \A_{M,L} \X_{M,L} -  \mathrm{blktr} [\A_{M,L}]\big) $ with elements given by
\begin{align}
\big[\De_M\big]_{n,m} &=  \frac{1}{L}\sum_{i = 1}^L \big[\A^{(i,i)}_M\big]_{n,m} \big( |x_i|^2 - 1\big) 
+ \frac{1}{L} \sum_{i=1}^L \sum_{j=1, j\neq i}^L \big[\A^{(i,j)}_M\big]_{n,m} x_i^{\ast} x_j, & n,m = 1, \ldots M
\end{align}
which can be rewritten as 
\begin{align}
\big[\De_M\big]_{n,m} = \frac{1}{L}\x_L^{\dagger} \bar{\A}_L^{(n,m)} \x_L - \frac{1}{L} \Tr(\bar{\A}_L^{(n,m)})
\end{align}
by introducing the $L \times L$ matrices $\bar{\A}_L^{(n,m)}$ with elements given by $\big[\bar{\A}_L^{(n,m)}]_{i,j} = [\A_M^{(i,j)}]_{n,m}$. Then, from Lemma \ref{lem:proof_trace}, we know that for $q\geq 1$,
\begin{align}\label{eq:bound_delta_elements}
\Exp\big\{ \big|\big[\De_M\big]_{n,m} \big|^q \big\} & \leq \frac{C_q}{L^q} \Big(\sum_{i,j =1}^L \big|[\A_M^{(i,j)}]_{n,m}\big|^2\Big)^{q/2}
\end{align}
with $C_q$ being a constant depending only on $q$. 
Furthermore, using that $\big\| \De_M \big\|\leq \big\| \De_M \big\|_F$ , we can bound $\Exp\big\{ \big\| \De_M \big\|^q  \big\}$ for any $q\geq 1$ as
\begin{align} \label{eq:bound_delta}
\Exp\big\{ \big\| \De_M \big\|^q  \big\} & \leq \Exp  \Big\{ \Big( \sum_{m,n=1}^M \big|[\De_M]_{n,m}\big|^2 \Big)^{q/2} \Big\} 
\leq \Exp  \Big\{ M^q\max_{m,n} \big|[\De_M]_{n,m}\big|^q \Big\}\leq M^q\sum_{m,n=1}^M\Exp\Big\{\big|[\De_M]_{n,m}\big|^q \Big\}.
\end{align}
Then,  substituting \eqref{eq:bound_delta_elements} back in \eqref{eq:bound_delta}, it follows for $q\geq 2$ that
\begin{align}
\Exp\big\{ \big\| \De_M \big\|^q  \big\} &\leq 
 \frac{C_qM^q}{L^q}\sum_{m,n=1}^M \Big( \sum_{i,j =1}^L \big|[\A_M^{(i,j)}]_{n,m}\big|^2 \Big)^{q/2}
 \leq  \frac{C_qM^q}{L^q} \Big( \sum_{m,n=1}^M\sum_{i,j =1}^L \big|[\A_M^{(i,j)}]_{n,m}\big|^2 \Big)^{q/2}.
 \end{align}
Using that $\|\A_{M,L}\|_F \leq \sqrt{ML} \|\A_{M,L}\|$, we finally have that
\begin{align}\label{eq:block_trace_main_ineq}
\Exp\big\{ \big\| \De_M \big\|^q  \big\} &\leq 
\frac{M^{3q/2}}{L^{q/2}}C_q  \|\A_{M,L}\|^{q}
 \end{align}
and by considering the case $q=4$ we conclude the first part of the proof.
\end{proof}


\begin{proof}[Proof of \textbf{(ii)}] 
We decompose the convergence result into two parts: 
\begin{equation} \label{eq:unif_converg_2terms}
\frac{1}{L}\Big\| \X_{M,L}^\dagger\A_{M,L}\X_{M,L} - \mathrm{blktr} [\A_{M,L}] \Big\|  \leq \Big\|\frac{1}{L}\sum_{i = 1}^L  \big( |x_i|^2 - 1\big)\A^{(i,i)}_M\Big\| + \Big\|\frac{1}{L}\sum_{i  \neq j}^L  x_ix_j^* \A^{(i,i)}_M\Big\|
\end{equation}
and show first that 
\begin{equation} \label{eq:conv_x_L}
\Big\|\frac{1}{L}\sum_{i = 1}^L  \big( |x_i|^2 - 1\big)\A^{(i,i)}_M\Big\|  \xrightarrow[L,M\rightarrow \infty]{\mathrm{a.s.}} 0.
\end{equation}
Let $y_i=|x_i|^2-1$ for $i =1, \ldots, L$, which satisfies $|y_i|\leq \chi^2+1$ a.s., given the assumption that $|x_i|\leq \chi$ a.s. Therefore, for each $i$ we have that $\|y_i\A^{(i,i)}_M\|<(\chi^2+1)\|\A_{M,L}\|\triangleq D$ a.s.
Then, from the matrix Bernstein inequality (see \cite[Th.~6.1]{{Trop12}}), it holds 
\begin{equation}\label{eq:concentration1}
\Pr\Big(\Big\|\sum_{i = 1}^L  y_i\A^{(i,i)}_M\Big\|\geq L\varepsilon\Big)\leq M\exp\Big(\frac{-L\varepsilon^2}{2(D^2+D\varepsilon/3)}\Big).
\end{equation}
Since $ ML^{-\delta}$ is finite by assumption and $L^{\delta}\exp(-CL)$ is summable, we can apply Borel-Cantelli lemma \cite[Thm.~4.3]{Bil95} and conclude that \eqref{eq:conv_x_L} holds.

Now we focus on the second term in \eqref{eq:unif_converg_2terms} and show that 
\begin{equation} \label{eq:conv_x_L2}
\Big\|\frac{1}{L}\sum_{i \neq j}^L  x_ix_j^*\A^{(i,j)}_M\Big\|  \xrightarrow[L,M\rightarrow \infty]{\mathrm{a.s.}} 0
\end{equation}
by applying the results in \cite{Trop11} to the Hermitian matrix process
\begin{equation}
\Y_L \triangleq \sum_{i \neq j}  x_ix_j^*\A^{(i,j)}_M = \sum_{\ell=1}^L \X_\ell \text{\ \ with\ \ } \X_\ell = \sum_{i=1}^{\ell-1} x_ix_\ell^*\A_M^{(i,\ell)} + \sum_{i=1}^{\ell-1} x_\ell x_i^*\A_M^{(\ell,i)}.
\end{equation}
For doing so, we first check the conditions of  \cite[Th.~1.2]{{Trop11}}.
The matrix process $\{\Y_L\}_{L=1,2\dots}$ is indeed a martingale since $\mathbb{E}_{L-1}\{ \Y_L \}= \Y_{L-1}$,  where $\mathbb{E}_{L-1}$ denotes the expectation with respect to $x_L$ given $x_1,\dots, x_{L-1}$ (see \cite{Trop11}).
Furthermore, the sequence $\{\X_\ell\}_{\ell=1,2\dots}$ is uniformly bounded as follows
\begin{align}
\|\X_\ell\| &\leq 2 \Big\| \sum_{i=1}^{\ell-1} x_ix_k^*\A_M^{(i,\ell)}\Big\| 
	\leq 2\chi \Big\| \sum_{i=1}^{\ell-1} x_i\A_M^{(i,\ell)}\Big\|
	\leq \max_{\u,\v\in\mathbb{C}^M,\|\u\|=\|\v\|=1} 2\chi \Big| \u^\dagger\Big(\sum_{i=1}^{\ell-1} x_i\A_M^{(i,\ell)}\Big)\v\Big|\\
 	&\leq \max_{\u,\v\in\mathbb{C}^M,\|\u\|=\|\v\|=1} 2\chi \big(  x_1\u^\dagger, \cdots,  x_{\ell-1}\u^\dagger, \0_M,  \cdots, \0_M \big)\A_{M,L} \big(\0_M, \cdots, \0_M, \v^{\dagger}, \0_M  \cdots, \0_M  \big)^{\dagger} \\
 &\leq 2\chi^2\sqrt{\ell-1}  \|\A_{M,L}\| \max_{\u,\in\mathbb{C}^M,\|\u\|=1}\|\u\| \max_{\v\in\mathbb{C}^M,\|\v\|=1}\|\v\| \leq 2\chi^2\sqrt{L} \|\A_{M,L}\|.
\end{align}
Moreover, we have that 
\begin{align}
\X_\ell^2 = x_\ell^2\sum_{i,j=1}^{\ell-1} x_i^*x_j^*\A_M^{(\ell,i)}\A_M^{(\ell,j)} + (x_\ell^*)^2\sum_{i,j=1}^{\ell-1} x_ix_j\A_M^{(i,\ell)}\A_M^{(j,\ell)} + |x_\ell|^2\sum_{i,j=1}^{\ell-1} \big(x_ix_j^*\A_M^{(i,\ell)}\A_M^{(\ell,j)} +x_i^*x_j\A_M^{(\ell,i)}\A_M^{(j,\ell)}\big)
\end{align}
and we define 
\begin{multline}
\W_L \triangleq \sum_{\ell=1}^L \mathbb{E}_{\ell-1} \big\{\X_\ell^2\big\} = \mathbb{E}\big\{(x_1^*)^2\big\}\sum_{\ell=1}^L \sum_{i,j=1}^{\ell-1}  x_i^*x_j^*\A_M^{(\ell,i)}\A_M^{(\ell,j)} + \mathbb{E}\big\{x_1^2\big\}\sum_{k=1}^L \sum_{i,j=1}^{\ell-1}  x_ix_j\A_M^{(i,\ell)}\A_M^{(j,\ell)}\\
+ \mathbb{E}\big\{|x_1|^2\big\} \sum_{\ell=1}^L \sum_{i,j=1}^{\ell-1}  \big(x_ix_j^*\A_M^{(i,\ell)}\A_M^{(\ell,j)} +x_i^*x_j\A_M^{(\ell,i)}\A_M^{(j,\ell)}\big)
\end{multline}
satisfying the inequality
\begin{equation}
\label{eq:inequality_Wk}
\|\W_L\| = \Big\|\sum_{\ell=1}^L \mathbb{E}_{\ell-1} \big\{ \X_\ell^2 \big\} \Big\| \leq 4L\chi^2\max_{1\leq \ell \leq L}\Big\|\sum_{i=1}^{\ell-1}  x_i\A_M^{(i,\ell)}\Big\|^2.
\end{equation}

Finally, we are ready to apply matrix Freedman's inequality for $\varepsilon>0$ \cite[Th.~1.2]{{Trop11}} and obtain for some $0< \alpha<1$ that
\begin{align}
\Pr\big(\|\Y_L\|\geq L\varepsilon\big) &= \Pr\Big(\|\Y_L\|\geq L\varepsilon , \|\W_L\|\leq 4L^{1+\alpha}\chi^2\Big) + \Pr\Big(\|\Y_L\|\geq L\varepsilon, \|\W_L\|\geq  4L^{1+\alpha}\chi^2\Big) \\
&\leq M\exp\Big(\frac{-L^2\epsilon^2/2}{4L^{1+\alpha}\chi^2 + 2\epsilon/3\chi^2L^{1+\alpha} \|\A_{M,L}\|}\Big) + \Pr\Big(\|\W_L\|\geq  4L^{1+\alpha}\chi^2\Big) \\
&\leq M\exp\Big(\frac{-L^{1-\alpha}\epsilon^2/2}{4\chi^2 + 2\epsilon/3\chi^2 \|\A_{M,L}\|}\Big) + \Pr\Big( \max_{1\leq \ell\leq L}\Big\|\sum_{i=1}^{\ell-1}  x_i\A_M^{(i,\ell)}\Big\|\geq  L^{\alpha/2}\Big) \\
&\leq M\exp\Big(\frac{-L^{1-\alpha}\epsilon^2/2}{4\chi^2 + 2\epsilon/3\chi^2 \|\A_{M,L}\|}\Big) + L\max_{1\leq \ell \leq L}\Pr\Big(\Big\|\sum_{i=1}^{\ell-1}  x_i\A_M^{(i,\ell)}\Big\|\geq  L^{\alpha/2}\Big)\label{eq:bernstein_in}
\end{align}
where we have used \eqref{eq:inequality_Wk}. 

Next we want to apply the matrix Bernstein inequality in \cite[Th.~1.6]{Trop12} to bound the second term in \eqref{eq:bernstein_in}. Let us fix $1\leq \ell \leq L$. Then, for any $1\leq i\leq \ell-1$, we have that 
$\mathbb{E}\big\{ x_i\A_M^{(i,\ell)}\big\}=0$, $\| x_i\A_M^{(i,\ell)}\|\leq \chi\|\A_{M,L}\|$, and
\begin{align} 
\Big\|\sum_{i=1}^{\ell-1}\mathbb{E}\big\{ |x_i|^2\big\} &\A_M^{(i,\ell)}\big[\A_M^{(i,\ell)}\big]^\dagger\Big\| \nonumber \\
&\leq \chi^2 \Big\|\sum_{i=1}^{L}\A_M^{(i,\ell)}\A_M^{(\ell,i)}\Big\|
\leq \chi^2 \Tr\Big(\sum_{i=1}^{L}\A_M^{(i,\ell)}\A_M^{(\ell,i)}\Big)
\leq \chi^2 \Tr\Big(\big[\A_{M,L}^2\big]^{(\ell,\ell)}\Big)
\leq \chi^2 M\Big\|\A_{M,L}\Big\|^2
\end{align}
where $\big[\A \big]^{(\ell,\ell)}$ denotes the $(\ell,\ell)$-th $M\times M$ block of matrix $\A$. Similarly, it can be shown that 
\begin{equation}
\Big\|\sum_{i=1}^{\ell-1}\mathbb{E}\big\{ |x_i|^2\big\} \big[\A_M^\dagger\big]^{(i,\ell)}\A_M^{(i,\ell)}\Big\|  \leq \chi^2 M\Big\|\A_{M,L}\Big\|^2
\end{equation}
 and we can finally apply \cite[Th.~1.6]{Trop12}:
\begin{equation}\label{eq:bernstein_in2}
\Pr\Big(\Big\|\sum_{i=1}^{\ell-1}  x_i\A_M^{(i,\ell)}\Big\|\geq  L^{\alpha/2}\Big) \leq 2M\exp\Big(\frac{-L^{\alpha}/2}{\chi^2 M \|\A_{M,L}\|^2 +\chi\|\A_{M,L}\|L^{\alpha/2}/3 }\Big).
\end{equation}
Since there exists $1>\delta>0$ such that $\lim\sup_{M,L} ML^{-\delta}<\infty$, we take $\delta>\gamma>0$ and $1>\alpha>0$ so that $\alpha=\delta+\gamma$ and by introducing the constants $C_1 = \frac{\epsilon^2/2}{4\chi^2 + 2\epsilon/3\chi^2 \|\A_{M,L}\|}$ and $C_2=\frac{1/2}{\chi^2 \|\A_{M,L}\|^2\lim\sup_{M,L} ML^{-\delta} +\chi\|\A_{M,L}\|/3 }$, and substituting \eqref{eq:bernstein_in2} back in \eqref{eq:bernstein_in}, it yields
\begin{equation}\label{eq:concentration2}
\Pr\big(\|\Y_L\|\geq L\varepsilon\big) \leq M\exp\big(-C_1L^{1-\alpha}\big) + 2 ML\exp\big(-C_2L^{\gamma}\big).
\end{equation}
Since $ ML^{-\delta}$ is finite and $L^{\delta}\exp(-CL^{1-\alpha})$ and $L^{\delta}\exp(-CL^\gamma)$ are summable, we can apply Borel-Cantelli lemma \cite[Thm.~4.3]{Bil95} and conclude that \eqref{eq:conv_x_L2} holds.
\end{proof}
\phantom\qedhere
\end{proof}

\section{Proofs}
\subsection{Proof of Theorem {\rm \ref{thm:deteq1}}} \label{app:deteq1}
\begin{proof}
Under the user covariance matrix model in \eqref{ass:sig_k}, we can rewrite the MSE of user 0 using \eqref{eq:Ce_ii_k} together with \eqref{eq:sum_mse_ii_0} as 
\begin{equation}
\mse_0^{(\mathsf{ii})}\big(\P, \{\Sig_k\}, \{\snr_k\}\big)  =
\frac{\beta_0}{M}\Tr\big(\Sig_0 \big) - \beta_0\frac{\snr_0}{M} \Tr\Big( \tilde{\P}_0\Sig_0^2\tilde{\P}_0^\dagger\big( \snr_0\tilde{\P}_0\Sig_0\tilde{\P}_0^\dagger+\A_M^\dagger\A_M+\I_{ML}\big)^{-1}\Big)
\end{equation}
with 
\begin{equation}
\A_M^\dagger \triangleq \frac{1}{\sqrt{M}} \Big[\Big(\frac{\snr_1}{\tau_{M,1}}\Big)^{1/2} \tilde{\P}_1\x_{1,1},\ldots, \Big(\frac{\snr_1}{\tau_{M,1}}\Big)^{1/2} \tilde{\P}_1\x_{1,r_1},
\ldots,\Big(\frac{\snr_K}{\tau_{M,K}}\Big)^{1/2}\tilde{\P}_K\x_{K,r_k}\Big]\in\mathbb{C}^{ML\times \sum_k r_k}
\end{equation}
so that
\begin{align}
\A_M^\dagger\A_M =  \sum_{k=1}^{K} \Big(\frac{\snr_k}{r_k}\Big) \sum_{i=1}^{r_k} \x_{k,i} \tilde{\P}_k\tilde{\P}_k^{\dagger}\x_{k,i}^{\dagger}.
\end{align}
Then, since for $\tau_{M,k} = r_k/M$ such that  $0<\lim \sup_{M,r_k} \tau_{M,k} <\infty$, it holds that  $\{ \frac{\snr_k}{\tau_{M,k}} \tilde{\P}_k\tilde{\P}_k^{\dagger}\}$, $ \snr_0 \tilde{\P}_0\Sig_0^2\tilde{\P}_0^\dagger$, and $\snr_0 \tilde{\P}_0\Sig_0\tilde{\P}_0^\dagger$ have uniformly bounded spectral norm with respect to $M$, we can directly apply \cite[Thm.~1]{Wag12} and obtain the convergence result in \eqref{eq:convergence_deteq1} with
\begin{equation}
\xi_{0}^{(\mathsf{ii})}\big(\P, \Sig_0, \{\snr_k\}\big)  =
\frac{\beta_0}{M} \Tr\big(\Sig_0 \big) - \beta_0\frac{\snr_0}{M} \Tr\Big(\tilde{\P}_0\Sig_0^2\tilde{\P}_0^\dagger\Big( \snr_0 \tilde{\P}_0\Sig_0\tilde{\P}_0^\dagger+ \frac{1}{M} \sum_{k=1}^K \frac{\snr_k}{\tau_{M,k}} \frac{r_k \tilde{\P}_{k}\tilde{\P}_{k}^\dagger}{1+\iota_{M,k}}+\I_{ML}\Big)^{-1}\Big)
\end{equation}
and the constants $\iota_{1,M},\dots,\iota_{M,k}$ given by the following fixed point equations
\begin{align} \label{eq:fp_eq_proof}
\iota_{M,k} &=\frac{1}{M}   \Tr\Big( \frac{\snr_k}{\tau_{M,k}} \tilde{\P}_{k}\tilde{\P}_{k}^\dagger\Big( \snr_0 \tilde{\P}_0\Sig_0\tilde{\P}_0^\dagger+
\frac{1}{M} \sum_{j=1}^K \frac{\snr_j}{\tau_{M,j}} \frac{r_j  \tilde{\P}_{j}\tilde{\P}_{j}^\dagger}{1+\iota_{j,M}}+\I_{ML}\Big)^{-1}\Big) & k = 1,\ldots, K.
\end{align}
Finally, for $\S_L$ as defined in \eqref{eq:SM}, we simplify  $\xi_{0}^{(\mathsf{ii})}\big(\P, \Sig_0\big)$ and the fixed point equations by applying Lemma~\ref{lem:inv}:
\begin{align}
\xi_{0}^{(\mathsf{ii})}\big(\P, \Sig_0, \{\snr_k\}\big)  &=
\frac{\beta_0 }{M} \Tr\Big(\Sig_0 - 
\snr_0\Sig_0\tilde{\P}_0^\dagger\big( \snr_0 \tilde{\P}_0\Sig_0\tilde{\P}_0^\dagger+ \S_{L}^{-1}\otimes\I_M\big)\tilde{\P}_0\Sig_0\Big)\\
&= \frac{\beta_0}{M} \Tr\big(\Sig_0 \big(\I_M+\snr_0(\p_0^\dagger\S_{L}^{-1}\p_0) \Sig_0\big)^{-1}\big)
= \frac{1}{M} \sum_{i=1}^M \frac{\beta_0\lambda_{0,i}}{1+\snr_0\lambda_{0,i}(\p_0^\dagger\S_{L}^{-1}\p_0)}
\end{align}
and
\begin{align}
\iota_{M,k} &=\frac{1}{M} \frac{\snr_k}{\tau_{M,k}} \Big(M  \big(\p_{k}^\dagger\S_{L}^{-1}\p_{k}\big) 
- \Tr\Big( \big(\frac{1}{\snr_0}\Sig_0^{-1}+\tilde{\P}_0^\dagger\big(\S_{L}^{-1}\otimes\I_M\big)\tilde{\P}_0\big)^{-1}
\big| \p_{k}^\dagger\S_{L}^{-1}\p_0\big|^2\Big)\\
&=\frac{\snr_k}{\tau_{M,k}}  \Big( \p_{k}^\dagger\S_{L}^{-1}\p_{k} -\frac{1}{M}\sum_{i=1}^M \frac{\lambda_{0,i} \snr_0|\p_{k}^\dagger\S_{L}^{-1}\p_0|^2}{1+\lambda_{0,i} \snr_0\p_0^\dagger\S_{L}^{-1}\p_0}\Big).
\end{align}
This completes the proof.
\end{proof}

\subsection{Proof of Theorem {\rm\ref{thm:deteq3}}} \label{app:deteq3}
\begin{proof}
The idea behind the proof of Theorem \ref{thm:deteq3} is to blockify the result of Bai and Silverstein \cite{Sil95}. To this end, we first rewrite $\mse_0^{(\mathsf{ii})}\big(\P, \{\Sig_k\}\big)$ in \eqref{eq:sum_mse_ii}  using the Woodbury identity in Lemma~\ref{lem:inv} and 
\begin{equation} \label{eq:mse_thm3}
\mse_0^{(\mathsf{ii})}\big(\P, \{\Sig_k\}, \{\snr_k\}\big) =  
\frac{\beta_0}{M} \Tr\big(\Sig_0 \big(\I_M+\tilde{\P}_0^\dagger(\tilde{\P}_{(0)} \Ps_{(0)} \tilde{\P}_{(0)}^\dagger + \I_{ML})^{-1}\tilde{\P}_0\Ps_0\big)^{-1}\big)
\end{equation}
where $\tilde{\P}_{(0)} = \big(\tilde{\P}_1, \ldots, \tilde{\P}_K\big)$ with $\tilde{\P}_k = \p_k \otimes \I_M$ and $\Ps_{(0)} = \Diag \big( \Ps_1, \ldots, \Ps_K\big)$ with $\Ps_k = \snr_k \Sig_k$ gathers all interfering users, i.e.,  $k=1,\dots,K$ (without user 0). Then, for  $\xi_{0}^{(\mathsf{ii})}\big(\{\Sig_k\}, \{\snr_k\}; \Gam_L\big)$  in \eqref{eq:xi_0_deteq3} rewritten as
\begin{align}
\xi_{0}^{(\mathsf{ii})}\big(\{\Sig_k\}, \{\snr_k\}; \Gam_L\big)= \frac{\beta_0}{M} \Tr \big(\Sig_0 \big(\I_M + L \snr_0 (\Gam_L + \I_M)^{-1}\Sig_0 \big)^{-1}\big)
\end{align}
and for
\begin{align}
\A_0\triangleq\frac{1}{L}\tilde{\P}_0^\dagger(\tilde{\P}_{(0)}\Ps_{(0)}\tilde{\P}_{(0)}^\dagger +\I_{ML})^{-1}\tilde{\P}_0
\end{align}
it holds that
\begin{align}
L \big| \mse_{0}^{(\mathsf{ii})}\big(\P, \{\Sig_k\}, \{\snr_k\}\big) - \xi_{0}^{(\mathsf{ii})}&\big(\{\Sig_k\}, \{\snr_k\}; \Gam_L\big) \big|  \nonumber \\
& =  L \frac{\beta_0}{M}
\big|\Tr\big(\Sig_0 \big((L\A_0 \Ps_0+\I_M)^{-1} - (\I_M + L(\Gam_L + \I_M)^{-1}\Ps_0)^{-1}\big) \big) \big| \\
& \leq L^2  \xi_{0}^{(\mathsf{ii})}\big(\{\Sig_k\}, \{\snr_k\}; \Gam_L\big) \big\|\Ps_0 \big(L\A_0\Ps_0 + \I_M\big)^{-1} \big \|
\big \| \A_0 - (\Gam_L + \I_M)^{-1} \big \|
\end{align}
where the inequality comes from Lemma \ref{lem:tr_AB}. Observing that 
\begin{align} \label{eq:bound_xi_ii}
\xi_{0}^{(\mathsf{ii})}\big(\{\Sig_k\}, \{\snr_k\}; \Gam_L\big) &\leq  \frac{\beta_0}{\snr_0}  \big\|  \big(\I_M + L (\Gam_L + \I_M)^{-1}\Ps_0 \big)^{-1} \Ps_0 \big\| 
\leq   \frac{\beta_0}{\snr_0} \big\|  \big(\Ps_0^{\#}+ L (\Gam_L + \I_M)^{-1}\big)^{-1} \big\|  \\
& \leq  \frac{\beta_0}{\snr_0L}  (\|\Gam_L\|+1)
\intertext{with $\Ps_0^{\#}$ denoting the pseudo-inverse of $\Ps_0$, and}
\big\|\Ps_0 \big(\I_M+L\A_0\Ps_0\big)^{-1} \big \|
& = 
\big\| \big( L(\Gam_L + \I_M)^{-1}\Ps_0 + L\big( ( \A_0  - (\Gam_L + \I_M)^{-1}) \Ps_0 \big) +  \I_M\big)^{-1} \Ps_0\big\| \\
& \leq  \frac{1}{L} \big\|\big( (\Gam_L + \I_M)^{-1}\ - \| \A_0 - (\Gam_L + \I_M)^{-1} \| \I_M \big)^{-1} +\frac{1}{L}\Ps_0^{\#} \big\| \\
& \leq   \frac{1}{L} \big( (\|\Gam_L\|+1)^{-1} -  \| \A_0 - (\Gam_L + \I_M)^{-1}  \|\big)^{-1} 
\end{align}
we can apply Lemma \ref{lem:resolvent} to obtain 
\begin{equation}
L \big| \mse_{0}^{(\mathsf{ii})}\big(\P, \{\Sig_k\}, \{\snr_k\}\big) - \xi_{0}^{(\mathsf{ii})}\big(\{\Sig_k\}, \{\snr_k\}; \Gam_L\big) \big| 
 \leq   \frac{\beta_0}{\snr_0} \frac{  (\|\Gam_L\|+1) \big\| \A_0 - (\Gam_L + \I_M)^{-1} \big\| }{(\|\Gam_L\|+1)^{-1}- \| \A_0 - (\Gam_L + \I_M)^{-1} \|}.
\end{equation}
In consequence, proving the theorem reduces to showing that \emph{\textbf{(i)}}
\begin{align} \label{eq:statement_i}
\big \| \A_0-  (\Gam_L + \I_M)^{-1}  \big\| \xrightarrow[K,L\rightarrow \infty]{\mathrm{a.s.}} 0
\end{align}
if $\lim \sup_{K,L}\lambda_{\max}(\Gam_L)<\infty$, which is the case under Assumption \ref{ass:enough_pilots} in \eqref{eq:enough_pilots}. Indeed, let us define $\S_L \triangleq (\Gam_L+\I_M)^{-1}$ and use the fixed point equation for $\Gam_L$ in 
\eqref{eq:gam_fp} to state that
\begin{align}
\big\|\S_L^{-1}\big\|
 &=\Big\|\sum_{k=1}^K \Ps_k^{1/2}\big(L\Ps_k^{1/2}\S_L\Ps_k^{1/2} + \I_M\big)^{-1}\Ps_k^{1/2} + \I_M\Big\|
 \leq\Big\|\sum_{k=1}^K \Ps_k^{1/2}\Big(\frac{L}{\|\S_L^{-1}\|}\Ps_k + \I_M\Big)^{-1}\Ps_k^{1/2} + \I_M\Big\|\\
 &\leq\Big\|\sum_{k=1}^K \U_k\Big(\frac{L}{\|\S_L^{-1}\|}\I_{r_k} + \frac{1}{\snr_k}\La_k^{-1}\Big)^{-1}\U_k^\dagger + \I_M\Big\|
 \leq \big\|\S_L^{-1} \big \| \Big\|\frac{1}{L}\sum_{k=1}^K \U_k\U_k^\dagger \Big\|+1
\end{align}
where $\U_k \in \Compl^{M \times r_k}$ contains the eigenvectors associated with the $r_k$ non-zero eigenvalues $\La_k = \Diag\big(\lambda_{k,1}, \ldots, \lambda_{k,r_k}\big)$ of $\Sig_k$. Hence, we can conclude that
\begin{equation}\label{eq:sum_uk}
\big\|\S_L^{-1}\big\|\leq \frac{1}{1-\big\|\frac{1}{L}\sum_{k=1}^K \U_k\U_k^\dagger \big\|}<\infty
\end{equation}
which implies  $\lim \sup_{K,L}\lambda_{\max}(\Gam_L)<\infty$.


\begin{proof}[Proof of \textbf{(i)}] Let us first introduce the following definitions 
\begin{align}
\A_L &\triangleq \frac{1}{L}\mathrm{blktr}\big[ \big(\tilde{\P}_{(0)}\Ps_{(0)}\tilde{\P}_{(0)}^\dagger + \I_{ML}\big)^{-1}\big] = \frac{1}{L}\mathrm{blktr}\big[\tilde{\A}_L\big] \\
\T_L(\A_L) &\triangleq   \Big(\sum_{k=1}^K \Ps_k \big( L\A_L \Ps_k + \I_M \big)^{-1} + \I_M\Big)^{-1}\label{eq:T}
\end{align}
with $\tilde{\A}_L = (\tilde{\P}_{(0)}\Ps_{(0)}\tilde{\P}_{(0)}^\dagger + \I_{ML})^{-1}$. 
Then, we can establish that
\begin{align}
\big \| \A_0-  (\Gam_L + \I_M)^{-1}  \big \| & \leq \big \| \A_0-  \A_L \big \| + \big \| \A_L -  (\Gam_L + \I_M)^{-1}  \big \|. 
\end{align}
Furthermore, from the blockified version of the trace lemma given in Proposition \ref{prop:block_trace} we know  that
\begin{equation}
\big\| \A_0 - \A_L \big\| 
\xrightarrow[K,L\rightarrow \infty]{\mathrm{a.s.}} 0.
\end{equation}
Thus, in order to prove \eqref{eq:statement_i}, it only remains to show that 
\begin{equation} \label{eq:A_L_minus_Gam_L}
\big \| \A_L -  (\Gam_L + \I_M)^{-1} \big \| \xrightarrow[K,L\rightarrow \infty]{\mathrm{a.s.}} 0.
\end{equation}
For doing so, we use Lemma \ref{lem:resolvent} to observe that  
\begin{equation}
\big \| \A_L -  (\Gam_L + \I_M)^{-1} \big) \big \|\leq \big \|\A_L^{1/2}\big \| \big \| \A_L^{-1/2}(\A_L - \S_L)\S_L^{-1/2} \big) \big \| \big\|\S_L^{1/2}\big\|\leq  \big \| \A_L^{-1/2}(\A_L - \S_L)\S_L^{-1/2}  \big \|
\end{equation}
where
\begin{align}
 \big \| \A_L^{-1/2}(\A_L -  \S_L)\S_L^{-1/2}  \big \| &\leq \big \| \A_L^{-1/2}(\A_L - \T_L(\A_L))\S_L^{1/2} \big \| + \big\| \A_L^{-1/2}(\T_L(\A_L) - \S_L)\S_L^{-1/2}  \big \|\\
 &\leq \big \|\A_L^{-1/2}\big \|\big \| \A_L - \T_L(\A_L)\big \| \big \|\S_L^{-1/2}\big \| + \big\| \A_L^{-1/2}(\T_L(\A_L) - \S_L)\S_L^{-1/2}  \big \|.
\end{align}
Then, using  the fact that $\|\A_L^{-1}\|\leq \varepsilon_M^{-1}$ and  $\|\S_L^{-1}\|<\infty$, we can show \eqref{eq:A_L_minus_Gam_L} by proving that \emph{\textbf{(ii)}}
\begin{equation} \label{eq:A-B}
\big\| \A_L - \T_L(\A_L)\big\| \xrightarrow[K,L\rightarrow \infty]{\mathrm{a.s.}} 0
\end{equation}
and that \emph{\textbf{(iii)}} there exists $0<\eta_M<1$ independent from $L,K$ such that
\begin{align} \label{eq:statement_iii}
\big \|  \A_L^{-1/2}(\T_L(\A_L) - \S_L) \S_L^{-1/2}\big  \| \leq\eta_M \big \|   \A_L^{-1/2}(\A_L - \S_L) \S_L^{-1/2}\big \|.
\end{align}
\end{proof}
\begin{proof}[Proof of \textbf{(ii)}]  Let us first rewrite  $\A_L - \T_L(\A_L)$ as
\begin{align}
 \A_L - \T_L(\A_L)
& =\frac{1}{L} \mathrm{blktr}\big [ \tilde{\A}_L  - \T_L(\A_L) \otimes \I_L  \big] \\
& =\frac{1}{L}\mathrm{blktr}\Big[ \big(\T_L(\A_L) \otimes \I_L\big)\Big(\tilde{\P}_{(0)}\Ps_{(0)}\tilde{\P}_{(0)}^\dagger 
- \Big(\sum_{k=1}^K \Ps_k \big( L\A_L \Ps_k + \I_M \big)^{-1}\Big)\otimes \I_L\Big)\tilde{\A}_L \Big]\\
&  =\frac{1}{L} \T_L(\A_L) \sum_{k=1}^K \mathrm{blktr}\Big[\tilde{\P}_k\Ps_k\tilde{\P}_k^\dagger\tilde{\A}_L  
-\Big(\big(\Ps_k \big( L\A_L \Ps_k + \I_M \big)^{-1}\big) \otimes \I_L\Big)\tilde{\A}_L \Big] \label{eq:A0-B0}
\end{align}
where the second equality follows from Lemma \ref{lem:resolvent}. Using that
\begin{align}
\mathrm{blktr}\big[ \tilde{\P}_k\Ps_k\tilde{\P}_k^\dagger\tilde{\A}_L \big]
& = \Ps_k \tilde{\P}_k^\dagger \tilde{\A}_L \tilde{\P}_k
=   \Ps_k  \big( \bar{\A}_{(k)} \Ps_k +\tfrac{1}{L} \I_M\big)^{-1}  \bar{\A}_{(k)}   
\end{align}
with $\bar{\A}_{(k)}  \triangleq \frac{1}{L}\tilde{\P}_k^\dagger \big( \tilde{\P}_{(0)}\Ps_{(0)}\tilde{\P}_{(0)}^\dagger 
- \tilde{\P}_k\Ps_k\tilde{\P}_k^\dagger + \I_{ML}\big)^{-1}\tilde{\P}_k$ and applying Lemma \ref{lem:inv} several times and Lemma \ref{lem:resolvent} again, we have that
\begin{align}
\A_L -  \T_L(\A_L) 
&= 	\T_L(\A_L) \Big( \sum_{k=1}^K  \Ps_k\big(L\bar{\A}_{(k)}\Ps_k+\I_M\big)^{-1} \bar{\A}_{(k)}  - \sum_{k=1}^K \Ps_k \big( L\A_L \Ps_k + \I_M \big)^{-1}\A_L \Big) \label{eq:A-B_aux} \\
&= 	\T_L(\A_L) \sum_{k=1}^K  \big(L\Ps_k\A_L+\I_M\big)^{-1}\Ps_k\big(\A_L-\bar{\A}_{(k)}\big) \big(L\Ps_k\bar{\A}_{(k)}+\I_M\big)^{-1}.  \label{eq:critical}
\end{align}
Given that $\|\T_L(\A_L)\|\leq 1$ and  $\|\A_L^{-1}\|\leq \varepsilon_M^{-1}$, it holds that  
\begin{align}
\big\| \A_L -  \T_L(\A_L) \big\|
&\leq  
\sum_{k=1}^K  \big\|\big(L\Ps_k\A_L+\I_M\big)^{-1}\Ps_k \big\| \big\|\A_L-\bar{\A}_{(k)}\big\| \big\|\big(L\Ps_k\bar{\A}_{(k)}+\I_M\big)^{-1}\big\| 
\leq  
\frac{\varepsilon_M^{-1}}{\alpha_L}\max_{k} \|\A_L-\bar{\A}_{(k)} \|
\end{align}
where $\max_{k}\|\A_L-\bar{\A}_{(k)}\|$ satisfies
\begin{equation} \label{eq:triang_A_k}
\max_{k}\|\A_L-\A_{(k)}\|\leq \max_{k}\|\A_{(k)}-\bar{\A}_{(k)}\|+ \max_{k}\|\A_L-\A_{(k)}\| 
\end{equation}
with $\A_{(k)}  \triangleq \frac{1}{L}\mathrm{blktr} \big( ( \tilde{\P}_{(0)}\Ps_{(0)}\tilde{\P}_{(0)}^\dagger - \tilde{\P}_k\Ps_k\tilde{\P}_k^\dagger + \I_{ML})^{-1}\big)$.

For the first term in the right-hand side of  \eqref{eq:triang_A_k} we follow the same idea as in \cite{Sil95}. More exactly, we use the inequality \eqref{eq:block_trace_main_ineq} in the proof of Proposition \ref{prop:block_trace} in order to state that there exists a constant $C_M$ such that 
\begin{equation}
\Exp\big\{\big\| \A_{(k)} - \bar{\A}_{(k)} \big\|^6\big\} \leq \frac{C_M}{L^3}
\end{equation}
and, then, we apply Boole's inequality \cite[eq.~(2.10)]{Bil95} and Markov's inequality \cite[eq.~(5.31)]{Bil95} to obtain that for any $\varepsilon>0$ 
\begin{align}
\Pr\Big(\max_{1\leq k\leq K} \big\| \A_{(k)} - \bar{\A}_{(k)} \big\|\geq \varepsilon\Big) 
\leq \sum_{k=1}^K \Pr\Big(\big\| \A_{(k)} - \bar{\A}_{(k)} \big\| \geq \varepsilon\Big)
\leq \sum_{k=1}^K \frac{1}{\varepsilon^6} \Exp\Big\{\big\| \A_{(k)} - \bar{\A}_{(k)} \big\|^6 \Big\} 
\leq \frac{C_MK}{\varepsilon^6 L^3}\label{eq:borelboole}.
\end{align}
Since $\tfrac{K}{L^3}$ is summable, we can conclude by the Borel-Cantelli lemma \cite[Thm.~4.3]{Bil95} that $\max_k \big\| \A_{(k)} - \bar{\A}_{(k)} \big\|\xrightarrow{\mathrm{a.s.}} 0$ as $K,L\rightarrow\infty$.

For the second term in the right-hand side of \eqref{eq:triang_A_k}, we use Lemma~\ref{lem:inv} to get 
\begin{equation}
\A_L-\A_{(k)} = \frac{1}{L}\mathrm{blktr}\Big[\tilde{\A}_{(k)}\tilde{\P}_k\Ps_k^{1/2}\big(\I_M + \Ps_k^{1/2}\tilde{\P}_k^\dagger\tilde{\A}_{(k)}\tilde{\P}_k\Ps_k^{1/2}\big)^{-1}\Ps_k^{1/2}\tilde{\P}_k^\dagger\tilde{\A}_{(k)} \Big].
\end{equation}
with $\tilde{\A}_{(k)} = ( \tilde{\P}_{(0)}\Ps_{(0)}\tilde{\P}_{(0)}^\dagger - \tilde{\P}_k\Ps_k\tilde{\P}_k^\dagger + \I_{ML})^{-1}$.
Therefore, given that $\|\tilde{\A}_{(k)}\| \leq 1$, we can upper bound the spectral norm as follows
\begin{align}
\big\|\A_L-\A_{(k)}\big\| &\leq  \frac{1}{L}\Tr\Big(\tilde{\A}_{(k)}^{3/2}\tilde{\P}_k\Ps_k^{1/2}\big(\I_M + \Ps_k^{1/2}\tilde{\P}_k^\dagger\tilde{\A}_{(k)}\tilde{\P}_k\Ps_k^{1/2}\big)^{-1}\Ps_k^{1/2}\tilde{\P}_k^\dagger\tilde{\A}_{(k)}^{1/2}\Big)\\
&\leq  \frac{1}{L}\big\|\tilde{\A}_{(k)}\big\| \Tr\Big(\Ps_k^{1/2}\tilde{\P}_k^\dagger\tilde{\A}_{(k)}\tilde{\P}_k\Ps_k^{1/2}\big(\I_M + \Ps_k^{1/2}\tilde{\P}_k^\dagger\tilde{\A}_{(k)}\tilde{\P}_k\Ps_k^{1/2}\big)^{-1}\Big)\\
&\leq   \frac{M}{L}\Big\|\Ps_k^{1/2}\tilde{\P}_k^\dagger\tilde{\A}_{(k)}\tilde{\P}_k\Ps_k^{1/2}\big(\I_M + \Ps_k^{1/2}\tilde{\P}_k^\dagger\tilde{\A}_{(k)}\tilde{\P}_k\Ps_k^{1/2}\big)^{-1}\Big\|\leq \frac{M}{L}.\label{eq:MoverL}
\end{align}
And this, together with \eqref{eq:triang_A_k} and  \eqref{eq:borelboole}, allows us to conclude that statement \emph{\textbf{(ii)}} in \eqref{eq:A-B} holds. 
\end{proof}

\begin{proof}[Proof of \textbf{(iii)}]  
Recall that $\S_L$ is solution to the fixed point equation $\S_L =\T_L(\S_L)$ with $\T_L(\cdot)$ as defined in \eqref{eq:T}. 
Then, applying Lemma \ref{lem:resolvent}, we have that
\begin{multline}
\A_L^{-1/2}(\T_L(\A_L)  - \T_L(\S_L))\S_L^{-1/2}  \\ = L\sum_{k=1}^K \A_L^{-1/2}\T_L(\A_L) \Ps_k\big(L\A_L\Ps_k+\I_M\big)^{-1}\big(\A_L-\S_L\big)\Ps_k \big(L\S_L\Ps_k+\I_M\big)^{-1} \T_L(\S_L)\S_L^{-1/2}
\end{multline}
whose spectral norm can be bounded using Lemma~\ref{lem:sum_ABC}. Indeed, let us introduce for some Hermitian $M\times M$ matrix $\B$
\begin{align} \label{eq:Mk}
\M_k(\B)\triangleq \B^{1/2}\Ps_k^{1/2}\big(L\Ps_k^{1/2}\B\Ps_k^{1/2}+\I_M\big)^{-1}\Ps_k^{1/2}\B^{1/2}
\end{align}
and write
\begin{multline}
\big\|\A_L^{-1/2}(\T_L(\A_L)  - \T_L(\S_L))\S_L^{-1/2}\big\|  
\leq
L\big\|\A_L^{-1/2}(\A_L  - \S_L)\S_L^{-1/2}\big\| \\ \Big\|\A_L^{-1/2}\T_L(\A_L)\A_L^{-1/2}\sum_{k=1}^K\M_k(\A_L)^2\A_L^{-1/2}\T_L(\A_L)\A_L^{-1/2}\Big\|^{1/2}\Big\|\sum_{k=1}^K\M_k(\S_L)^2\Big\|^{1/2} .\label{eq:contraction}
\end{multline}
We can further bound \eqref{eq:contraction} using again Lemma~\ref{lem:sum_ABC} and the fact that $\M_k(\A_L)$ is positive definite and $\|\M_k(\A_L)\|\leq \frac{1}{L}$:
\begin{align}
\Big\|\A_L^{-1/2}\T_L(\A_L)\A_L^{-1/2}\sum_{k=1}^K&\M_k(\A_L)^2\A_L^{-1/2}\T_L(\A_L)\A_L^{-1/2}\Big\| \nonumber \\
&\leq \max_{1\leq k\leq K} \big\|\M_k(\A_L)\big\| \Big\|\A_L^{-1/2}\T_L(\A_L)\A_L^{-1/2}\sum_{k=1}^K\M_k(\A_L)\A_L^{-1/2}\T_L(\A_L)\A_L^{-1/2}\Big\|\\
&\leq \frac{1}{L} \big\|\A_L^{-1/2}\T_L(\A_L)\big(\T_L(\A_L)^{-1} - \I_M\big)\T_L(\A_L)\A_L^{-1/2}\big\|\\
&\leq \frac{1}{L} \big(\|\A_L^{-1/2}(\T_L(\A_L)-\A_L)\A_L^{-1/2} + \I_M\|\big) \big( \| \T_L(\A_L) - \A_L\| + \|\A_L-\I_M\|\big)\\
&\leq \frac{1}{L} \big(\|\A_L^{-1}\|\|\T_L(\A_L)-\A_L\| + 1\big) \big( \| \T_L(\A_L) - \A_L\| +1- \lambda_{\min}(\A_L) \big).
\end{align}
Recall that by assumption of the theorem  there exists $\epsilon_M>0$ such that $\lim\inf_{K,L} \lambda_{\min}(\A_L)>\epsilon_M$ and this 
proves that there exists $0<\eta_M<1$ such that 
\begin{align}
\Big\|\A_L^{-1/2}\T_L(\A_L)\A_L^{-1/2}\sum_{k=1}^K\M_k(\A_L)^2\A_L^{-1/2}\T_L(\A_L)\A_L^{-1/2}\Big\| <\frac{\eta_M}{L}
\end{align}
for $L$ large enough.
We can similarly prove that $\big\|\sum_{k=1}^K\M_k(\S_L)^2\big\|\leq\frac{1}{L}$, and hence, statement \emph{\textbf{(iii)}} in \eqref{eq:statement_iii} holds.

Finally, it only remains to show that $\S_L$ is the unique fixed point of $\T_L$. First observe that the mapping $\T_L$ is continuous and is defined from $\mathcal{B}$ into $\mathcal{B}$ where 
\begin{equation}
\mathcal{B} = \big\{\S\in\mathbb{C}^{M\times M} \text{\ \ s.t.\ \ } \I_M\succeq \S\succeq \0_M \big\}
\end{equation}
which is a compact convex set and, therefore, $\T_L$ admits a fixed point. Let us suppose that there exist two fixed points $\S_1^{\star}\neq \S_2^{\star} \in \mathcal{B}$. Observe now that $\T_L(\M)$ is invertible for any positive semidefinite matrix $\M \in \mathcal{B}$ and, hence, any fixed point of $\T_L(\M)$ is also invertible. In consequence, we necessarily have that  $\S_1^{\star}\succ 0$ and $\S_2^{\star}\succ 0$. Then, using inequality \eqref{eq:contraction} with $\S_1^{\star}$ and $\S_2^{\star}$ gives us 
\begin{align*}
\|(\S_1^{\star})^{-1/2}(\S_1^{\star}-\S_2^{\star})(\S_2^{\star})^{-1/2}\| &= \|(\S_1^{\star})^{-1/2}(\T_L(\S_1^{\star})-\T_L(\S_2^{\star}))(\S_2^{\star})^{-1/2}\|\\
&\leq \min(\|\I_M - \S_1^{\star}\|, \|\I_M - \S_2^{\star}\| )\|(\S_1^{\star})^{-1/2}(\S_1^{\star}-\S_2^{\star})(\S_2^{\star})^{-1/2}\|
\end{align*}
which for  $\S_1^{\star}\neq \S_2^{\star}$ only holds if $\min(\|\I_M - \S_1^{\star}\|,\|\I_M - \S_2^{\star}\|) = 1$. However, this implies that $\lambda_{\min}(\S_1^{\star})= \lambda_{\min}(\S_2^{\star})=  0$ and this contradicts the fact that both fixed points are positive definite. Then, the fixed point of $\T_L$ is necessarily unique which completes the proof of Theorem \ref{thm:deteq3}.
\end{proof}

\subsection{Proof of Theorem {\rm\ref{thm:deteq3_uni}}} \label{app:deteq3_uni}

\begin{proof}

Considering that constants $\epsilon_M$ and $\nu_M$ do not depend on $M$ and under the existence of some $0<\delta<1$ such that $\lim\sup_M ML^{-\delta}<\infty$ and the assumption that $\p_0$ is uniformly bounded,  we can prove the convergence result in Theorem \ref{thm:deteq3} uniformly in $M$ using similar arguments as in the previous proof.
In particular, we need to use \eqref{eq:concentration1} and \eqref{eq:concentration2} in order to show that $\big\| \A_{(k)} - \bar{\A}_{(k)} \big\|\xrightarrow{\mathrm{a.s.}} 0$ as $K,L\rightarrow\infty$ uniformly in $M$ and $\lim\sup_M ML^{-\delta}<\infty$ ensures that $\frac{M}{L}$ goes to $0$ as $M,L\rightarrow \infty$  in \eqref{eq:MoverL}. Furthermore, the fact that $\epsilon_M$ does not depend on $M$ makes $\eta_M$ in the proof of \emph{\textbf{(iii)}} also independent of $M$. This completes the proof of the uniform convergence in Theorem \ref{thm:deteq3_uni}.
\end{proof}
\phantom\qedhere
\end{proof}
\subsection{Proof of Proposition {\rm \ref{prop:Gam_L}}} \label{app:A_L}
\begin{proof}
We need to prove that there exists a constant $\epsilon_M>0$ such that $\lim\inf_{K,L} \lambda_{\min}(\A_L)\geq \epsilon_M$ a.s. for $\A_L = \frac{1}{L}\mathrm{blktr}\big[(\tilde{\P}_{(0)}\Ps_{(0)}\tilde{\P}_{(0)}^\dagger + \I_{ML})^{-1}\big]$, where $\tilde{\P}_{(0)} = \big(\tilde{\P}_1, \ldots, \tilde{\P}_K\big)$ with $\tilde{\P}_k = \p_k \otimes \I_M$ and $\Ps_{(0)} = \Diag \big( \Ps_1, \ldots, \Ps_K\big)$ with $\Ps_k = \snr_k \Sig_k$, when either the conditions \emph{\textbf{(a)}} or \emph{\textbf{(b)}} in the proposition are satisfied.

Let us define $\y_\ell =(p_0(\ell),\dots, p_K(\ell))^T$,  $\tilde{\Y}_{\ell}=\y_\ell \otimes \I_M$, $\tilde{\Y}=(\tilde{\Y}_1,\dots,\tilde{\Y}_L)$, and $\tilde{\Y}_{(\ell)}=(\tilde{\Y}_1,\dots,\tilde{\Y}_{\ell-1},\tilde{\Y}_{\ell+1},\dots, \tilde{\Y}_{L})$, so that
\begin{equation}
\tilde{\P}^\dagger\tilde{\P}  = \tilde{\Y}\tilde{\Y}^\dagger = \sum_{\ell=1}^L \tilde{\Y}_{\ell}\tilde{\Y}_{\ell}^\dagger.
\end{equation}
Then, including user $0$, we can state that
\begin{align}
\A_L \succeq \frac{1}{L}\mathrm{blktr}\big[(\tilde{\P}_{(0)}\Ps_{(0)}\tilde{\P}_{(0)}^\dagger + \tilde{\P}_{0}\Ps_{0}\tilde{\P}_{0}^\dagger+ \I_{ML})^{-1}\big]
= \frac{1}{L}\mathrm{blktr}\big[(\tilde{\Y}^\dagger\Ps\tilde{\Y}+ \I_{ML})^{-1}\big] %
\end{align}
where $\Ps = \Diag \big(\Ps_0, \Ps_{(0)}\big) = \Diag \big( \Ps_0, \ldots, \Ps_K\big)$, 
Introducing $\Z_{(\ell)} = \tilde{\Y}_{\ell}^\dagger\Ps\big(\I_{M(K+1)}+\tilde{\Y}_{(\ell)}\tilde{\Y}_{(\ell)}^\dagger\Ps\big)^{-1}\tilde{\Y}_{\ell}$ and using the Woodbury identity in Lemma~\ref{lem:inv}, it holds that
\begin{align}
\A_L \succeq \frac{1}{L}\mathrm{blktr}\big[(\Y^\dagger\Ps\Y+ \I_{ML})^{-1}\big]  
= \I_M -  \frac{1}{L}\sum_{\ell=1}^L \mathrm{blktr}\big[\tilde{\Y}_{\ell}\tilde{\Y}_{\ell}^\dagger\Ps(\Y\Y^\dagger\Ps+ \I_{M(K+1)})^{-1}\big]
&  =  \frac{1}{L}\sum_{\ell=1}^L \big(\I_M+\Z_{(\ell)}\big)^{-1} \\ & \succeq \frac{\I_M}{1+\frac{1}{L}\sum_{\ell=1}^L \|\Z_{(\ell)}\|} \label{eq:Z_l}
\end{align}
where we used the convexity of $x\mapsto \frac{1}{1+x}$. Now, for $\La = \Diag(\La_0,\dots,\La_K)$, $\U = \Diag (\U_0,\dots,\U_K)$, and $\tilde{\D}_{\snr} = \Diag(\snr_0, \ldots, \snr_K) \otimes \I_M$, it holds that $\Ps =   \U \tilde{\D}_{\snr} \La \U^{\dagger}$ and we can write 
\begin{equation}
\Ps\big(\I_{M(K+1)}+\tilde{\Y}_{(\ell)}\tilde{\Y}_{(\ell)}^\dagger\Ps\big)^{-1}= \U \big( \tilde{\D}_{\snr}^{-1} \La^{-1}+\U^\dagger\tilde{\Y}_{(\ell)}\tilde{\Y}_{(\ell)}^\dagger\U \big)^{-1}\U^\dagger. 
\end{equation}
Furthermore, under assumption \emph{\textbf{(a)}} in \eqref{eq:sum_snr}, it holds that
\begin{align}
\big( \tilde{\D}_{\snr}^{-1} \La^{-1}+\U^\dagger\tilde{\Y}_{(\ell)}\tilde{\Y}_{(\ell)}^\dagger\U \big)^{-1} \preceq \tilde{\D}_{\snr} \La \preceq  \max_{0\leq k \leq K} \| \Sig_k \|  \tilde{\D}_{\snr}.
\end{align}
Therefore, we have  
\begin{align}
\frac{1}{L}\sum_{\ell=1}^L \|\Z_{(\ell)}\| &\leq  \Big( \max_{0\leq k \leq K} \| \Sig_k \|\Big) \frac{1}{L}\sum_{\ell=1}^L \|\tilde{\Y}_\ell^\dagger\U\tilde{\D}_{\snr}\U^\dagger\tilde{\Y}_\ell\| \leq  \Big(\max_{0\leq k \leq K} \| \Sig_k \|\Big) \frac{1}{L}\sum_{k=0}^K\sum_{\ell=1}^L  |p_k(\ell)|^2 \snr_k\\
&\leq  \Big(\max_{0\leq k \leq K} \| \Sig_k \|\Big)\Big( \max_{1\leq k\leq K}\frac{1}{L}\sum_{\ell=1}^L  |p_k(\ell)|^2\Big) \sum_{k=0}^K\snr_k<\infty
\end{align}
where we have used that $\max_k\frac{1}{L}\sum_{\ell=1}^L  |p_k(\ell)|^2 \xrightarrow[K,L\rightarrow \infty]{\mathrm{a.s.}} 1$ and the assumptions in \emph{\textbf{(a)}}. This combined with \eqref{eq:Z_l} finishes the first part of the proof.

Introducing $\B_{(\ell)} = \tilde{\Y}_{\ell}^\dagger\Ps\big(\I_{M(K+1)}+\tilde{\Y}\tilde{\Y}^\dagger\Ps\big)^{-1}\tilde{\Y}_{\ell}$ and using Lemma~\ref{lem:inv} leads to 
$\Z_{(\ell)}  - \B_{(\ell)} = \B_{(\ell)}(\I_{M} + \B_{(\ell)})^{-1}\B_{(\ell)} \preceq \B_{(\ell)}$.
Thus, we get that
\begin{align}
\max_{1\leq \ell \leq L} \|\Z_{(\ell)}\| &\leq 2\max_{1\leq \ell \leq L} \|\B_{(\ell)}\| \leq 2\max_{1\leq \ell \leq L} \Big\|\tilde{\Y}_{\ell}^\dagger\U(  \tilde{\D}_{\snr}^{-1}  \tilde{\La}^{-1}+\U^\dagger\tilde{\P}^\dagger\tilde{\P}\U)^{-1}\U^\dagger\tilde{\Y}_{\ell}\Big\|\\
&\leq \frac{2}{c_M(K+1)}\max_{1\leq \ell \leq L} \Big\|\tilde{\Y}_{\ell}^\dagger\U\U^\dagger\tilde{\Y}_{\ell}^\dagger\Big\| \leq \frac{2}{c_M}\max_{1\leq \ell \leq L} \frac{\|\y_{\ell}\|^2}{K+1} <\infty
\end{align}
where we have used that $\max_{1\leq \ell \leq L} \frac{\|\y_{\ell}\|^2}{K+1}\xrightarrow[K,L\rightarrow \infty]{\mathrm{a.s.}} 1$ and the assumption in \emph{\textbf{(b)}}. This combined with \eqref{eq:Z_l} finishes the second part of the proof.
\end{proof}

\subsection{Proof of Proposition {\rm \ref{prop:Gam_L2}}} \label{app:Gam_L2}
\begin{proof}
Let $\Gam_L$ be the unique solution to the fixed point equation in \eqref{eq:gam_fp} and let $\gamma_L$ be the unique solution to the fixed point equation in \eqref{eq:gamma}. Then, we need to show that 
\begin{align}
\frac{1}{M}\Big\| (\Gam_L+\I_M)^{-1} - (1+\gamma_L)^{-1}\I_M\Big\|_F^2\xrightarrow[K,L,M\rightarrow \infty]{\mathrm{a.s.}} 0.
\end{align}
Let us introduce $\bar{\Gam}_L \triangleq \gamma_L\I_M$, $\S_L \triangleq (\Gam_L+\I_M)^{-1}$, and  $\bar{\S}_L \triangleq (1+\gamma_L)^{-1}\I_M$ and express the fixed point equation in \eqref{eq:gam_fp} as $\S_L = \T_L(\S_L)$ with $\T_L(\cdot)$ as defined in \eqref{eq:T}. Then, observe that the assumption that $\lim\sup_{L,K}\Big\|\frac{1}{L}\sum_{k=1}^K \U_k\U_k^\dagger\Big\|<1$ uniformly in $M$ made in the statement proposition implies, following the inequality in \eqref{eq:sum_uk}, that $\|\Gam_L\|<\infty$ for any $L$. In consequence, we have that $\|\S_L^{-1}\|<\infty$ uniformly in $M$.
Moreover, $\|\S_L\|\leq 1$ so that it holds $\big\|\S_L - \bar{\S}_L\big\|_F\leq \big\|\S_L^{-1/2}(\S_L - \bar{\S}_L)\big\|_F$, where
\begin{align} \label{eq:S-barS}
\big\|\S_L^{-1/2}(\S_L - \bar{\S}_L)\big\|_F
&\leq \big\|\S_L^{-1/2}\big\| \big\|\T_L(\bar{\S}_L) - \bar{\S}_L\big\|_F + \big\|\S_L^{-1/2}(\T_L(\S_L) - \T_L(\bar{\S}_L))\big\|_F .
\end{align}

We focus first on the second term of the right-hand side of \eqref{eq:S-barS}. Following a similar approach to the one in part \emph{\textbf{(iii)}} in the proof of Theorem \ref{thm:deteq3}, we can control the term 
$\big\|\S_L^{-1/2}\big(\T_L(\S_L) - \T_L(\bar{\S}_L) \big) \big\|_F$  as follows:
\begin{align}
\big\|\S_L^{-1/2} &\big( \T_L(\S_L) - \T_L(\bar{\S}_L) \big) \big\|_F \nonumber \\
\label{eq:normTL_1}
&=  \Big\|L\sum_{k=1}^K \S_L^{-1/2}\T_L(\S_L) \Ps_k\big(L\S_L\Ps_k+\I_M\big)^{-1}\big(\S_L-\bar{\S}_L\big)\Ps_k \big(L\bar{\S}_L\Ps_k+\I_M\big)^{-1} \T_L(\bar{\S}_L)\Big\|_F\\
&\leq  \Big\|L\sum_{k=1}^K \S_L^{1/2} \Ps_k\big(L\S_L\Ps_k+\I_M\big)^{-1}\big(\S_L-\bar{\S}_L\big)\Ps_k \big(L\bar{\S}_L\Ps_k+\I_M\big)^{-1} \T_L^{1/2}(\bar{\S}_L)\bar{\S}_L^{1/2}\Big\|_F \label{eq:normTL_1a} \\
& \hspace{.8cm}+ \Big\|L\sum_{k=1}^K \S_L^{1/2}\Ps_k\big(L\S_L\Ps_k+\I_M\big)^{-1}\big(\S_L-\bar{\S}_L\big)\Ps_k \big(L\bar{\S}_L\Ps_k+\I_M\big)^{-1} \T_L^{1/2}(\bar{\S}_L)(\T_L^{1/2}(\bar{\S}_L) - \bar{\S}_L^{1/2})\Big\|_F
\label{eq:normTL_1b}
\end{align}
where \eqref{eq:normTL_1} comes from applying Lemma \ref{lem:resolvent} twice. The term in \eqref{eq:normTL_1b} can be further bounded as
\begin{align}
\Big\|L\sum_{k=1}^K \S_L^{1/2}&  \Ps_k\big(L\S_L\Ps_k+\I_M\big)^{-1}\big(\S_L-\bar{\S}_L\big)\Ps_k \big(L\bar{\S}_L\Ps_k+\I_M\big)^{-1} \T_L^{1/2}(\bar{\S}_L)(\T_L^{1/2}(\bar{\S}_L) - \bar{\S}_L^{1/2})\Big\|_F \nonumber \\
&\leq L\sum_{k=1}^K \big \|\S_L^{1/2}\Ps_k\big(L\S_L\Ps_k+\I_M\big)^{-1}\big(\S_L-\bar{\S}_L\big)\Ps_k \big(L\bar{\S}_L\Ps_k+\I_M\big)^{-1} \T_L^{1/2}(\bar{\S}_L)\big \| \big\|\T_L^{1/2}(\bar{\S}_L) - \bar{\S}_L^{1/2}\big\|_F\\
&\leq \frac{K}{L} \big \|\S_L^{-1/2}\big\| \big\|\S_L-\bar{\S}_L\big\| \big\|\bar{\S}_L^{-1}\T_L^{1/2}(\bar{\S}_L)\big\| 
\big\| \big(\T_L(\bar{\S}_L) - \bar{\S}_L \big) \big(\T_L^{1/2}(\bar{\S}_L)+\bar{\S}_L^{1/2}\big)^{-1}\big\|_F \\
&\leq \frac{2}{\alpha_L} \big\|\S_L^{-1}\big\|^{1/2} \big\|\bar{\S}_L^{-1}\big\|^{3/2} \big\|\T_L(\bar{\S}_L) - \bar{\S}_L\big\|_F\label{eq:normTL_2}
\end{align}
noting that $\|\T_L(\bar{\S}_L)\|\leq 1$ and $\big\|\S_L-\bar{\S}_L\big\|\leq 2$. The term in \eqref{eq:normTL_1a} satisfies
\begin{align}
\Big\|L\sum_{k=1}^K \S_L^{1/2} &\Ps_k\big(L\S_L\Ps_k+\I_M\big)^{-1}\big(\S_L-\bar{\S}_L\big)\Ps_k \big(L\bar{\S}_L\Ps_k+\I_M\big)^{-1} \T_L^{1/2}(\bar{\S}_L)\bar{\S}_L^{1/2}\Big\|_F^2 \nonumber \\
& \leq L^2 \Big\| \T_L^{1/2}(\bar{\S}_L)\sum_{k=1}^K\bar{\S}_L\Ps_k \big(L\bar{\S}_L\Ps_k+\I_M\big)^{-2} \Ps_k\T_L^{1/2}(\bar{\S}_L)\Big\| \sum_{k=1}^K \Tr\big(\M_k(\S_L)\S_L^{-1}(\S_L - \bar{\S}_L)^2\M_k(\S_L)\big)\\
& \leq L \Big\| \T_L^{1/2}(\bar{\S}_L)\sum_{k=1}^K \Ps_k^{1/2}\big(L\bar{\S}_L\Ps_k+\I_M\big)^{-1} \Ps_k^{1/2}\T_L^{1/2}(\bar{\S}_L)\Big\| \Big\| \sum_{k=1}^K \M_k(\S_L)^2\Big\|  \big\|\S_L^{-1/2}(\S_L - \bar{\S}_L)\big\|_F^2\\
& \leq   \big\|\T_L(\bar{\S}_L) - \I_M\big\| \big\|\S_L^{-1/2}(\S_L - \bar{\S}_L)\big\|_F^2 =   \big(1-\lambda_{\min}\big(\T_L(\bar{\S}_L)\big)\big) \big\|\S_L^{-1/2}(\S_L - \bar{\S}_L)\big\|_F^2 \label{eq:normTL_3}
\end{align}
where we have used Lemma \ref{lem:sum_ABC} with $\M_k(\S_L)$ defined as in \eqref{eq:Mk} and, hence, $\big\|\sum_{k=1}^K\M_k(\S_L)^2\big\|\leq\frac{1}{L}$. On the other hand, it holds that
\begin{align}
\T_L(\bar{\S}_L) =  \Big(\sum_{k=1}^K \Ps_k \big(L \bar{\S}_L \Ps_k +\I_M \big)^{-1} + \I_M\Big)^{-1} 
\succeq \Big(\I_M+\frac{K}{L}\bar{\S}_L^{-1}\Big)^{-1}\succeq \frac{1}{1+\tfrac{K}{L}(1+\gamma_L)}\I_M
\end{align} 
and, since $\gamma_L$ satisfies \eqref{eq:gamma_L_bounds}, we conclude that 
\begin{align}
\lambda_{\min} \big(\T_L(\bar{\S}_L)\big) \geq \frac{1-\frac{\bar{\tau}_M}{\alpha_L}}{1+\frac{1-\bar{\tau}_M}{\alpha_L}}.
\end{align}
Finally, combining \eqref{eq:normTL_2} together with \eqref{eq:normTL_3}, we can substitute back in \eqref{eq:S-barS} and obtain that
\begin{equation}
\|\S_L^{-1/2}(\S_L - \bar{\S}_L)\|_F \leq \frac{\|\S_L^{-1/2}\| + \frac{2}{\alpha_L}\|\S_L^{-1}\|^{1/2} \|\bar{\S}_L^{-1}\|^{3/2}}{1-(1-\lambda_{\min}(\T_L(\bar{\S}_L)))^{1/2}} \|\T_L(\bar{\S}_L) - \bar{\S}_L\|_F.
\end{equation}
In consequence, since from \eqref{eq:gamma_L_bounds} we know that $\|\bar{\S}_L^{-1}\|$ is uniformly bounded, in order to complete the proof it remains to show that 
\begin{equation} \label{eq:TL-SL}
\frac{1}{M}\big\|\T_L(\bar{\S}_L) - \bar{\S}_L \big\|_F^2 \xrightarrow[K,L,M\rightarrow \infty]{\mathrm{a.s.}} 0.
\end{equation}

With this objective, let us first bound $\|\T_L(\bar{\S}_L) - \bar{\S}_L\|_F^2$ as
\begin{align}
\|\T_L(\bar{\S}_L) - \bar{\S}_L\|_F^2 
& = \Big\|  \Big(\I_M - \bar{\S}_L -\bar{\S}_L\sum_{k=1}^K \Ps_k\big(L\bar{\S}_L\Ps_k+\I_M\big)^{-1}\Big)  \T_L(\bar{\S}_L)\Big\|_F^2\\
& \leq  \big\|\T_L(\bar{\S}_L)\big\|^2
\Big\|\I_M -\bar{\S}_L- \bar{\S}_L\frac{1}{L}\sum_{k=1}^K \Ps_k\big(\bar{\S}_L\Ps_k+\tfrac{1}{L}\I_M\big)^{-1}\Big\|_F^2\\
& \leq  \big\|\bar{\S}_L\big\|^2 
\Big\|\frac{\gamma_L}{\gamma_L+1}\I_M-\frac{1}{L}\sum_{k=1}^K \Ps_k\big(\Ps_k+\tfrac{\gamma_L+1}{L}\I_M\big)^{-1}\Big\|_F^2\\
& \leq \frac{1}{\alpha_L^2} 
\Big\|\frac{1}{K}\sum_{k=1}^K \Big(\Ps_k\big(\Ps_k+\tfrac{\gamma_L+1}{L}\I_M\big)^{-1} - \ell_k\big(-\tfrac{\gamma_L+1}{L}\big)\I_M\Big)\Big\|_F^2
\label{eq:app4_auxx}
\end{align}
where we have used that $\|\T_L(\bar{\S}_L)\|^2<1$ and  $\|\bar{\S}_L\|^2<1$, and in the last inequality we have applied the fixed point equation of $\gamma_L$ in \eqref{eq:gamma} and we have defined 
\begin{align}
\ell_k(z) &=  \frac{2\tau_{M,k}}{1+\tau_{M,k}-\frac{\tau_{M,k}z}{\snr_k}+\sqrt{(1+\tau_{M,k}-\frac{\tau_{M,k}z}{\snr_k})^2-4\tau_{M,k}}}, & z<0.
\end{align}
Now we introduce $\De_k = \Ps_k\big(\Ps_k+\tfrac{\gamma_L+1}{L}\I_M\big)^{-1} - \ell_k\big(-\tfrac{\gamma_L+1}{L}\big)\I_M$, which satisfies
\begin{align}
\frac{1}{M} \Big\|\frac{1}{K}\sum_{k=1}^K \De_k\Big\|_F^2 = \frac{1}{MK^2}\sum_{k,k'=1}^K  \Tr(\De_k\De_{k'})= \frac{1}{MK^2}\sum_{k=1}^K  \Tr(\De_k^2) +  \frac{1}{MK^2}\sum_{k\neq k'}  \Tr(\De_k\De_{k'}).
\end{align}
Since the spectral norm of $\De_k$ is almost surely upper bounded by $2$, $\frac{1}{MK^2}\sum_{k=1}^K  \Tr(\De_k^2)\xrightarrow[]{\mathrm{a.s.}} 0$ as $K,L,M\rightarrow \infty$. For the second term, $\frac{1}{MK^2}\sum_{k\neq k'}  \Tr(\De_k\De_{k'})$, we can use the independence between $\De_k$ and $\De_{k'}$ and bound the following function:
\begin{align}\label{eq:f_L_M}
f_{L,M}(z) & = \frac{1}{MK^2}\sum_{k\neq k'}^K  \Tr\big(\big(\Ps_k (\Ps_k-z\I_M)^{-1} - \ell_k(z)\I_M\big)\De_{k'}\big), & z<0.
\end{align}
Indeed, given the independence between $\De_{k}$ and $\De_{k^{\prime}}$,  we can treat $\De_{k^{\prime}}$ as a deterministic matrix with respect to $\De_{k}$. Equivalently, in the following we take the expectations  with respect to the distribution of $\De_k$. Recall that, under the covariance matrix model in \eqref{ass:sig_k}, $\Ps_k = \frac{\snr_k}{r_k}\sum_{i=1}^{r_k}\x_{k,i}\x_{k,i}^\dagger$ with $\x_{k,i}$ vectors of i.i.d. entries of zero-mean and unit variance. 
Then, using Lemma~\ref{lem:inv}, we get
\begin{align}
\Tr\big(\Ps_k\big(\Ps_k-z\I_M\big)^{-1}\De_{k'}\big)
= \frac{\snr_k}{r_k}\sum_{j=1}^{r_k} \x_{k,j}^\dagger(\Ps_k-z\I_M)^{-1}\De_{k'}\x_{k,j}
&= \frac{\snr_k}{r_k}\sum_{j=1}^{r_k} \frac{\x_{k,j}^\dagger (\Ps_{k,(j)}-z\I_M)^{-1}\De_{k'}\x_{k,j}}
{1+\frac{\snr_k}{r_k}\x_{k,j}^\dagger (\Ps_{k,(j)}-z\I_M)^{-1}\x_{k,j}}\\
&= \frac{1}{M} \sum_{j=1}^{r_k} \frac{\x_{k,j}^\dagger (\Ps_{k,(j)}-z\I_M )^{-1}\De_{k'}\x_{k,j}}
{\frac{\tau_{M,k}}{\snr_k}+\frac{1}{M}\x_{k,j}^\dagger(\Ps_{k,(j)}-z\I_M )^{-1}\x_{k,j}}
\end{align}
where $\Ps_{k,(j)} = \frac{\snr_k}{r_k} \sum_{i=1,i\neq j}^{r_k} \x_{k,i}\x_{k,i}^\dagger$. Similarly to \cite{Wag12}, we want first to upper bound  $\zeta_{k,k'}(z)$ defined as
\begin{align}
\zeta_{k,k'}(z)& = \frac{1}{M}\sum_{j=1}^{r_k} \frac{\x_{k,j}^\dagger(\Ps_{k,(j)}-z\I_M)^{-1}\De_{k'}\x_{k,j}}{\frac{\tau_{M,k}}{\snr_k}+\frac{1}{M}\x_{k,j}^\dagger(\Ps_{k,(j)}-z\I_M)^{-1}\x_{k,j}} 
- \Big(\frac{\tau_{M,k}}{\frac{\tau_{M,k}}{\snr_k}+\frac{1}{\snr_k}m_{\tau_{M,k}}(\frac{z}{\snr_k})} \Big) 
\Tr\big((\Ps_k-z\I_M)^{-1}\De_{k'}\big)
\end{align}
with $m_{\tau_{M,k}}$ being the Stieltjes transform of the Marcenko-Pastur distribution \cite{debbahcouillet}
\begin{equation}
m_{\tau_{M,k}}(z) = \frac{\tau_{M,k}-1-\tau_{M,k}z}{2z} - \frac{\sqrt{(\tau_{M,k}+1-\tau_{M,k}z)^2-4\tau_{M,k}}}{2z}.
\end{equation}
Observing that
\begin{align} \label{eq:l_k_over_1_l_k}
\frac{\tau_{M,k}}{\frac{\tau_{M,k}}{\snr_k}+\frac{1}{\snr_k}m_{\tau_{M,k}}(\frac{z}{\snr_k})}
= -\frac{z\ell_k(z)}{1-\ell_k(z)}
\end{align}
it holds that
\begin{align}
\Tr\big(\Ps_k (\Ps_k-z\I_M)^{-1}\De_{k'}\big) 
&= \zeta_{k,k'}(z) - \frac{z\ell_k(z)}{1-\ell_k(z)} \Tr\big( (\Ps_k-z\I_M)^{-1}\De_{k'}\big)\\
&= \zeta_{k,k'}(z) + \frac{\ell_k(z)}{1-\ell_k(z)}\Big(\Tr(\De_{k'}) - \Tr\big(\Ps_k(\Ps_k-z\I_M)^{-1}\De_{k'}\big)\Big)
\end{align}
and, therefore, we can write 
\begin{equation} \label{eq:tr_xi_k}
\Tr\big(\Ps_k (\Ps_k-z\I_M)^{-1}\De_{k'}\big) - \ell_k(z)\Tr(\De_{k'}) = (1-\ell_k(z))\zeta_{k,k'}(z).
\end{equation}
Let us now decompose $\zeta_{k,k'}(z)=\zeta_{k,k'}^{(1)}(z)+\zeta_{k,k'}^{(2)}(z)$, with
\begin{align}
\zeta_{k,k'}^{(1)}(z) 
&= \frac{1}{M} \sum_{j=1}^{r_k} 
\Big( \frac{1}{\frac{\tau_{M,k}}{\snr_k}+\frac{1}{M}\x_{k,j}^\dagger(\Ps_{k,(j)}-z\I_M)^{-1}\x_{k,j}}
+\frac{z\ell_k(z)}{1-\ell_k(z)} \frac{1}{\tau_{M,k}}\Big) \x_{k,j}^\dagger(\Ps_{k,(j)}-z\I_M)^{-1}\De_{k'}\x_{k,j} \\
\zeta_{k,k'}^{(2)}(z) &=  -\frac{z\ell_k(z)}{1-\ell_k(z)} 
\Big( \frac{1}{\tau_{M,k}} \sum_{j=1}^{r_k} \x_{k,j}^\dagger(\Ps_{k,(j)}-z\I_M)^{-1}\De_{k'}\x_{k,j} 
-  \Tr\big( (\Ps_k-z\I_M)^{-1}\De_{k'}\big)\Big).
\end{align}
Using the equality in \eqref{eq:l_k_over_1_l_k}, the term $\zeta_{k,k'}^{(1)}(z)$ can be bounded as
\begin{align}
\frac{1}{M} \big|\zeta_{k,k'}^{(1)} (z)\big| 
&\leq  \frac{|z|\ell_k(z)\|\De_{k'}\|}{1-\ell_k(z)}
\Big| \frac{1}{Mr_k}\sum_{j=1}^{r_k} \x_{k,j}^\dagger(\Ps_{k,(j)}-z\I_M)^{-1}\x_{k,j} - 
\frac{1}{\snr_k}m_{\tau_{M,k}}\big(\tfrac{z}{\snr_k}\big)\Big|\\
&\leq\frac{|z|\ell_k(z)\|\De_{k'}\|}{1-\ell_k(z)}\Big| \frac{1}{Mr_k}\sum_{j=1}^{r_k} \x_{k,j}^\dagger (\Ps_{k,(j)}-z\I_M )^{-1}\x_{k,j} 
- \frac{1}{M}\Tr\big( (\Ps_{k}-z\I_M)^{-1}\big)\Big|\\
&\hspace{2.25cm} +\frac{|z|\ell_k(z)\|\De_{k'}\|}{1-\ell_k(z)}\Big| \frac{1}{M}\Tr\big( (\Ps_{k}-z\I_M )^{-1}\big) -  \frac{1}{\snr_k}m_{\tau_{M,k}}\big(\tfrac{z}{\snr_k}\big)\Big|. \label{eq:bound_xi_1}
\end{align}
Similarly, for the term $\zeta_{k,k'}^{(2)}(z)$ we get 
\begin{align}
\frac{1}{M}\big|\zeta_{k,k'}^{(2)}(z)\big|\leq 
\frac{1}{M} \frac{|z|\ell_k(z)}{1-\ell_k(z)}\Big| \frac{1}{r_k}\sum_{j=1}^{r_k} \x_{k,j}^\dagger(\Ps_{k,(j)}-z\I_M )^{-1}\De_{k'}\x_{k,j} - \Tr\big((\Ps_{k}-z\I_M)^{-1}\De_{k'}\big)\Big| \label{eq:bound_xi_2}.
\end{align}
Given that $\|\De_{k'}\|\leq 2$  almost surely  and $0< \ell_k(z) < \tau_{M,k}$, we can substitute the bounds in \eqref{eq:bound_xi_1} and \eqref{eq:bound_xi_2} back in equation \eqref{eq:tr_xi_k} and obtain that
\begin{multline} 
\frac{1}{M}\big|\Tr\big(\Ps_k (\Ps_k-z\I_M)^{-1}\De_{k'}\big) - \ell_k(z)\Tr(\De_{k'})\big|  \leq 
  \frac{|z|\tau_{M,k}}{M} \Big| \frac{1}{r_k}\sum_{j=1}^{r_k} \x_{k,j}^\dagger(\Ps_{k,(j)}-z\I_M )^{-1}\De_{k'}\x_{k,j} -\Tr\big((\Ps_{k}-z\I_M)^{-1}\De_{k'}\big)\Big| \\
+ \frac{\|\De_{k'}\| |z|\tau_{M,k}}{M} \Big| \frac{1}{r_k}\sum_{j=1}^{r_k} \x_{k,j}^\dagger (\Ps_{k,(j)}-z\I_M)^{-1}\x_{k,j} - \Tr\big( (\Ps_{k}-z\I_M )^{-1}\big)\Big|\\
 + |z|\tau_{M,k}\|\De_{k'}\| \Big| \frac{1}{M}\Tr\big( (\Ps_{k}-z\I_M)^{-1}\big) -  \frac{1}{\snr_k}m_{\tau_{M,k}}\big(\tfrac{z}{\snr_k}\big)\Big|.
\label{eq:ineq_baisilverstein}
\end{multline}
Then, we can use that
\begin{align} \label{eq:bound_tr-l_k_last_term}
|z| \Big| \frac{1}{M}\Tr\big((\Ps_{k}-z\I_M)^{-1}\big) -  \frac{1}{\snr_k}m_{\tau_{M,k}}\big(\tfrac{z}{\snr_k}\big)\Big|= \frac{1}{M} \Big| \Tr\big(\Ps_{k}(\Ps_{k}-z\I_M)^{-1}\big) -  M\ell_k(z)\Big|
\end{align}
and set $\De_{k'}=\I_M$ in \eqref{eq:ineq_baisilverstein} so that we can bound the right hand side of \eqref{eq:bound_tr-l_k_last_term} as
\begin{multline} \label{eq:bound_tr-l_k_last_term_2}
\frac{1}{M}\big|\Tr\big(\Ps_k (\Ps_k-z\I_M)^{-1}\big) - M\ell_k(z)\big|  \leq 
\frac{1}{M} \frac{2|z|\tau_{M,k}}{(1-\tau_{M,k})}  
\Big| \frac{1}{r_k}\sum_{j=1}^{r_k} \x_{k,j}^\dagger (\Ps_{k,(j)}-z\I_M)^{-1}\x_{k,j} - \Tr\big( (\Ps_{k}-z\I_M )^{-1}\big)\Big|.
\end{multline}
Finally, we are in the position to bound function $f_{L,M}(z)$ in \eqref{eq:f_L_M} using \eqref{eq:ineq_baisilverstein} together with 
\eqref{eq:bound_tr-l_k_last_term_2}. Indeed, there exists constants $D_1 \leq \max_k \tau_{M,k} $  and $D_2 \leq 2\frac{1+\max_k\tau_{M,k}}{1-\max_k\tau_{M,k}}$ independent from $K,L,M$ such that 
\begin{align}
|f_{M,L}(z)|&\leq \frac{1}{K^2}\sum_{k\neq k'}^K \frac{1}{M} \Big| \Tr\big(\Ps_k (\Ps_k-z\I_M)^{-1}\De_{k'}\big) - \ell_k(z)\Tr(\De_{k'})\Big| \nonumber\\
&\leq D_1|z| \frac{1}{K^2}\sum_{k\neq k'}^K \frac{1}{M r_k} \Big|\sum_{j=1}^{r_k} \big(\x_{k,j}^\dagger(\Ps_{k,(j)}-z\I_M)^{-1}\De_{k'}\x_{k,j} - 
\Tr\big((\Ps_{k}-z\I_M)^{-1}\De_{k'}\big)\big)\big| \label{eq:f_M_L_terms_1}\\
& \hspace{3.5cm}+D_1D_2|z|  \frac{1}{K^2}\sum_{k\neq k'}^K \frac{1}{M r_k} \Big| \sum_{j=1}^{r_k} \x_{k,j}^\dagger(\Ps_{k,(j)}-z\I_M )^{-1}\x_{k,j} - \Tr\big( (\Ps_{k}-z\I_M)^{-1}\big)\Big| \label{eq:f_M_L_terms_2}.
\end{align}

Let us first focus on the term in \eqref{eq:f_M_L_terms_1}. Defining
\begin{equation}
\delta_{j,k,k'}(z) =  \x_{k,j}^\dagger(\Ps_{k,(j)}-z\I_M)^{-1}\De_{k'}\x_{k,j} -\Tr\big((\Ps_{k}-z\I_M)^{-1}\De_{k'}\big)
\end{equation}
and, using H\"older's inequality \cite[eq.~(5.35)]{Bil95}  on the sum over $j$, we have that
\begin{align}
\Exp_k\Big\{\Big|\frac{1}{Mr_k}\sum_{j=1}^{r_k}\delta_{j,k,k'}(z)\Big|^q\Big\} &\leq \frac{r_k^{q-1}}{(Mr_k)^q}\sum_{j=1}^{r_k}\Exp_k\big\{|\delta_{j,k,k'}(z)|^q\big\} \\
&\leq \frac{2^{q-1}}{M^q r_k}\sum_{j=1}^{r_k}\Big(\frac{C_q}{|z|^q}\|\De_{k'}\|^q M^{q/2} +\frac{(\snr_k)^q}{r_k^q|z|^{2q}}\Exp\big\{\|\x_{k,j}\|^{2q}\big\}\|\De_{k'}\|^q\Big) \label{eq:bound_sum_j}
\end{align}
where the last inequality can be obtained as follows.
From Lemma~\ref{lem:proof_trace} we know that there exists a constant $C_q$ for any $q\geq 1$ such that
\begin{multline}
\Exp_k\big\{\big|\x_{k,j}^\dagger(\Ps_{k,(j)}-z\I_M)^{-1}\De_{k'}\x_{k,j} - \Tr\big((\Ps_{k,(j)}-z\I_M)^{-1}\De_{k'}\big)\big|^q\big\}\\
\leq C_q\big\|(\Ps_{k,(j)}-z\I_M)^{-1}\De_{k'}\big\|_F^q \leq \frac{C_q}{|z|^q}\|\De_{k'}\|^q M^{q/2} \label{eq:bound_sum_j_1}. 
\end{multline}
On the other hand, it holds
\begin{align}
\big|\Tr\big((\Ps_{k,(j)}-z\I_M)^{-1}\De_{k'}\big) &- \Tr\big((\Ps_{k}-z\I_M)^{-1}\De_{k'}\big)\big|^q \nonumber \\
&= \big|\Tr\big((\Ps_{k,(j)}-z\I_M)^{-1}\big( \frac{\snr_k}{r_k}\x_{k,j}\x_{k,j}^\dagger\big)(\Ps_{k}-z\I_M)^{-1}\De_{k'}\big)\big|^q\\
&\leq\frac{(\snr_k)^q}{r_k^q}\|\x_{k,j}\|^{2q}\big\| (\Ps_{k}-z\I_M )^{-1}\De_{k'}(\Ps_{k,(j)}-z\I_M)^{-1}\big\|^q
\leq\frac{(\snr_k)^q}{r_k^q|z|^{2q}}\|\x_{k,j}\|^{2q} \|\De_{k'}\|^q \label{eq:bound_sum_j_2}
\end{align}
which, combined with \eqref{eq:bound_sum_j_1} and applying again H\"older's inequality, gives
\begin{equation}
\Exp_k\big\{|\delta_{j,k,k'}(z)|^q\big\}
\leq 2^{q-1}\Big(\frac{C_q}{|z|^q}\|\De_{k'}\|^qM^{q/2} +\frac{(\snr_k)^q}{r_k^q|z|^{2q}}\Exp\big\{\|\x_{k,j}\|^{2q}\big\} \|\De_{k'} \|^q\Big)
\end{equation}
and this proves the bound in \eqref{eq:bound_sum_j}. Finally, we resort to H\"older's inequality on the sum over $k,k'$ to obtain 
\begin{multline}
\Exp_k\Big\{\Big|\frac{1}{K^2}\sum_{k\neq k'}^K \frac{|z|}{Mr_k}\sum_{j=1}^{r_k}\delta_{j,k,k'}(z)\Big|^q\Big\}
\leq \frac{1}{K^2}\sum_{k\neq k'}^K \Exp_k\Big\{\Big|\frac{|z|}{Mr_k}\sum_{j=1}^{r_k}\delta_{j,k,k'}(z)\Big|^q\Big\}\\
\leq \frac{2^{q-1}}{K^2M^{q/2}}\sum_{k\neq k'}^K\frac{1}{r_k}\sum_{j=1}^{r_k}\Big(C_q\|\De_{k'}\|^q +\frac{(\snr_k)^q}{\tau_{M,k}^qM^{3q/2}|z|^{q}}\Exp\{\|\x_{k,j}\|^{2q}\} \|\De_{k'}\|^q\Big).\label{eq:holdertype}
\end{multline}
Since $\x_{k,j}$ have finite eight-order moment and $\lim\inf_{M,r_k} \tau_{M,k}>0$, we apply \eqref{eq:holdertype} for $q=4$, which results in 
\begin{align}
\frac{1}{K^2}\sum_{k\neq k'}^K \Exp_k\Big\{\Big|\frac{|z|}{Mr_k}\sum_{j=1}^{r_k}\delta_{j,k,k'}(z)\Big|^4\Big\}\leq \frac{8}{M^2}\Big(C^{\prime} + \frac{C^{\prime\prime}}{(M^{3/2}|z|)^4}\Big).
\end{align}
with $C^{\prime}=\sup_k C_4\|\De_{k}\|^4$ and $C^{\prime\prime}=\sup_{k,k'} \frac{(\snr_k)^4}{\tau_{M,k}^4}\Exp\{\|\x_{k,j}\|^{8}\}\|\De_{k'}\|^4$. 

Recall now the bound for $|f_{L,M}(z)|$ in \eqref{eq:f_M_L_terms_1} and \eqref{eq:f_M_L_terms_2}. Since we can use the previous procedure also for the second term, we can conclude that 
\begin{align}
\Exp\big\{ |f_{L,M}(z)|^4 \big\} \leq \frac{8}{M^2}\Big(D^{\prime} + \frac{D^{\prime\prime}}{(M^{3/2}|z|)^4}\Big)
\end{align}
for some constants $D^\prime$ and $D^{\prime\prime}$.
Given that  $\lim\sup_{M,L} \frac{L}{M^{3/2}}<\infty$ since $\lim\inf_{M,L} ML^{-2/3}>0$, and $\gamma_L>0$, we now can use Markov's inequality \cite[eq.~(5.31)]{Bil95} to establish that
\begin{equation}
\Pr\Big(\Big|f_{M,L}\Big(-\frac{\gamma_L+1}{L}\Big)\Big|\geq \epsilon\Big)\leq \frac{1}{\epsilon^4}\Exp\Big\{ \big|f_{M,L}\Big(-\frac{\gamma_L+1}{L}\Big)\Big|^4\Big\}=O\Big(\frac{1}{M^2}\Big).
\end{equation}
Noting that $\frac{1}{M^2}$ is summable, we can finally call Borel-Cantelli lemma \cite[Thm.~4.3]{Bil95} to see that 
\begin{equation}
f_{M,L}\Big(-\frac{\gamma_L+1}{L}\Big) = \frac{1}{MK^2}\sum_{k\neq k'}^K  \Tr(\De_k\De_{k'}) \xrightarrow[K,L,M\rightarrow \infty]{\mathrm{a.s.}} 0.
\end{equation}
Plugging this result back in \eqref{eq:app4_auxx}, shows \eqref{eq:TL-SL} and thus completes the proof.

\end{proof}

\subsection{Proof of Theorem~{\rm \ref{thm:deteq4}}}\label{app:deteq4}
\begin{proof}
Observe first that, under assumptions of the theorem, the convergence of the deterministic equivalent in Theorem~\ref{thm:deteq3} holds uniformly in $M$. Hence, in order to prove Theorem~\ref{thm:deteq4}, we just need to show that 
\begin{align}
L\left|\xi_{0}^{(\mathsf{ii})}\big(\{\Sig_k\}, \{\snr_k\}; \Gam_L\big) - \bar{\xi}_{0}^{(\mathsf{ii})}\big(\Sig_0, \gamma_L\big)\right|
\xrightarrow[K,L,M\rightarrow \infty]{\mathrm{a.s.}} 0 
\end{align}
where  $\xi_{0}^{(\mathsf{ii})}\big(\{\Sig_k\}, \{\snr_k\}; \Gam_L\big)$ is the deterministic equivalent in Theorem~\ref{thm:deteq3}. We use Lemma \ref{lem:resolvent} and write 
\begin{align}
L\left|\xi_{0}^{(\mathsf{ii})}\big(\Gam, \Sig_0\big) - \bar{\xi}_{0}^{(\mathsf{ii})}\big(\Sig_0\big)\right|
= \frac{L\beta_0}{M} \left|\Tr\Big(\Sig_0(\I_M+L\snr_0\S_L\Sig_0)^{-1} L\snr_0(\S_L-\bar{\S}_L)\Sig_0 (\I_M+L\snr_0\bar{\S}_L\Sig_0)^{-1} \Big) \right|
\end{align}
for  $\S_L \triangleq (\Gam_L+\I_M)^{-1}$  and  $\bar{\S}_L \triangleq (1+\gamma_L)^{-1}\I_M$. Then, it holds that
\begin{align}
L^2\left|\xi_{0}^{(\mathsf{ii})}\big(\Gam, \Sig_0\big) - \bar{\xi}_{0}^{(\mathsf{ii})}\big(\Sig_0\big)\right|^2
\leq \frac{1}{M^2\snr_0^2} \left|\Tr\Big(\S_L^{-1}(\S_L-\bar{\S}_L)\bar{\S}_L^{-1} \Big)\right|^2
&\leq \frac{\|\S_L^{-1}\bar{\S}_L^{-1}\|_F^2}{M\snr_0^2}\frac{1}{M}\big\|\S_L-\bar{\S}_L \big\|_F^2\\
&\leq \frac{\|\S_L^{-1}\bar{\S}_L^{-1}\|^2}{\snr_0^2}\frac{1}{M}\big\|\S_L-\bar{\S}_L \big\|_F^2
\end{align} 
using the Cauchy-Schwarz inequality. We can now conclude the proof by applying Proposition \ref{prop:Gam_L2}.
\end{proof}

 \subsection{Proof of Corollary {\rm \ref{cor:deteq4}}} \label{app:deteq4bis}
 \begin{proof}[Proof of \textbf{(a)}]
Under the conditions of Theorem~\ref{thm:deteq4} and the condition in \eqref{eq:beta_far_ues}, we need to show that 
 \begin{align} \label{eq:limit_gamma_2}
 | \gamma_L  |\xrightarrow[ L,K\rightarrow \infty]{} 0
 \end{align}
where $\gamma_L$ is the unique fixed-point of the function  $G_L:\mathbb{R}\rightarrow \mathbb{R}$ defined as 
 \begin{align}
G_L(x) &\triangleq  \frac{1}{L}\sum_{k=1}^K \frac{2\tau_{M,k}(x+1)}{1+\tau_{M,k}+\frac{\tau_{M,k}(x+1)}{L\snr_k}+\sqrt{\big(1+\tau_{M,k}+\frac{\tau_{M,k}(x+1)}{L\snr_k}\big)^2-4\tau_{M,k}}}, & 0\leq x\leq \frac{\bar{\tau}_M}{\alpha_L-\bar{\tau}_M}.
 \end{align}
Observe that $G_L$ is positive and satisfies
\begin{align}
G_L(x) \leq \frac{1}{L}\sum_{k=1}^K \frac{\tau_{M,k}(x+1)}{1+\frac{\tau_{M,k}(x+1)}{L\snr_k}} 
\leq\frac{1}{L}\sum_{k=1}^K \frac{1}{\frac{1}{1+x}+\frac{1}{L\snr_k}}
\leq \frac{1}{L}\sum_{k=1}^K \frac{1}{1-\frac{\bar{\tau}_M}{\alpha_L}+\frac{1}{L\snr_k}}
\end{align}
 which, under the condition in \eqref{eq:beta_far_ues} and using that $\lim\inf_{K,L} \frac{L}{K}>0$, proves \eqref{eq:limit_gamma_2}.
 \end{proof}

 \begin{proof}[Proof of \textbf{(b)}]
 Under the conditions of Theorem~\ref{thm:deteq4}, and the condition in \eqref{eq:beta_close_ues},  we need to show that 
 \begin{align} \label{eq:limit_gamma}
 | (1+\gamma_L)^{-1} - (1+\gamma_{\infty})^{-1}|  \xrightarrow[K,L \rightarrow \infty]{\mathrm{a.s.}} 0
 \end{align}
where $\gamma_{\infty}= \frac{\bar{\tau}_M}{\alpha_L-\bar{\tau}_M}$ and $\gamma_L$ is the unique solution of the fixed-point equation in \eqref{eq:gamma}, which can be rewritten as
 \begin{equation}
 \frac{\gamma_L}{1+\gamma_L} =  \frac{\bar{\tau}_M}{\alpha_L} +\frac{1}{L}\sum_{k=1}^K  \Big(1-\tau_{M,k}+\frac{\tau_{M,k}(1+\gamma_L)}{L\snr_k}-\sqrt{\Big(1-\tau_{M,k}+\frac{\tau_{M,k}(1+\gamma_L)}{L\snr_k}\Big)^2-\frac{4\tau_{M,k}^2(1+\gamma_L)}{L\snr_k}}\Big).
 \end{equation}
 Therefore, $\frac{\gamma_L}{1+\gamma_L}$ is the unique fixed point of the function $F_L:\mathbb{R}\rightarrow \mathbb{R}$ defined as
 \begin{align}
 \label{eq:FL}
 F_L(x) \triangleq \frac{\bar{\tau}_M}{\alpha_L} +\frac{1}{L}\sum_{k=1}^K  \Big(1-\tau_{M,k}+\frac{\tau_{M,k}}{L(1-x)\snr_k}-\sqrt{\Big(1-\tau_{M,k}+\frac{\tau_{M,k}}{L(1-x)\snr_k}\Big)^2-\frac{4\tau_{M,k}^2}{L(1-x)\snr_k}}\Big).
 \end{align}

Since the sequence $x_1=\frac{\gamma_1}{1+\gamma_1},\dots,x_L=\frac{\gamma_L}{1+\gamma_L}$ satisfies $\frac{\gamma_{\infty}}{1+\gamma_{\infty}}\leq x_{\ell}\leq\frac12$ for any $\ell\geq 1$, we can extract a subsequence $x_{\varphi(1)},\dots,x_{\varphi(L)}$ converging to some $x_{\infty}$. 
Furthermore, for any $x$ such that $\frac{\gamma_{\infty}}{1+\gamma_{\infty}}\leq x\leq\frac12$,  we can use that $ 1-\sqrt{x} \leq \sqrt{1-x}$ whenever $x\leq 1$ and obtain 
 \begin{align}
 F_L(x)- \frac{\bar{\tau}_M}{\alpha_L}&\leq \frac{1}{L}\sum_{k=1}^K  \sqrt{\frac{4\tau_{M,k}^2}{L(1-x)\snr_k}}\leq \frac{2\sqrt{2}}{L^{3/2}}\sum_{k=1}^K  \sqrt{\frac{1}{\snr_k}}
 \end{align}
and this converges to zero under the condition in  \eqref{eq:beta_close_ues}, considering that $\lim\inf_{K,L} \frac{L}{K}>0$.
Finally, taking the limit in \eqref{eq:FL} gives $x_{\infty}=\frac{\bar{\tau}_M}{\alpha_L}$ which shows that $\frac{\bar{\tau}_M}{\alpha_L}$ is the limit of any subsequence $(x_{\varphi_i})_i$ and, hence,  proves \eqref{eq:limit_gamma}.
\end{proof}

\end{appendices}

\clearpage

\bibliographystyle{IEEEtran}
\bibliography{RMT_references,ref_snops,IEEEabrv}

\end{document}